\documentclass[letterpaper]{sig-alternate}

\usepackage{amsmath}

\usepackage{mdwlist}%

\usepackage[
colorlinks=true,
linkcolor=linkColor,
citecolor=linkColor,
urlcolor=linkColor,
pdftitle={Extracting and Verifying Cryptographic Models from C Protocol Code by Symbolic Execution},
pdfsubject={},
pdfauthor={Aizatulin,Gordon,Juerjens},
pdfkeywords={},
]{hyperref}

\usepackage[svgnames]{xcolor}

\usepackage{versions}

      \newif\iffull\fulltrue  %
      \newif\ifdraft\drafttrue %

      \draftfalse\fulltrue

      \ifdraft
            \def\Blue{\color{MidnightBlue}}
            \def\Black{\color{black}}

            \newcommand{\displaycomment}[1]{ #1 {}}
            \includeversion{DRAFT}
      \else
            \def\Blue{}\def\Black{}
            \newcommand{\eatspace}[1]{#1}
            \newcommand{\displaycomment}[1]{\eatspace}
            \excludeversion{DRAFT}
      \fi

      \iffull
      \newenvironment{FULL}{\ifdraft\Blue\fi}{\ifdraft\Black\fi}
      \excludeversion{SHORT}
      \else
      \excludeversion{FULL}
      \includeversion{SHORT}
      \fi

\usepackage{etex}
\usepackage[hyperref,thmmarks,amsmath]{ntheorem}
      \theoremnumbering{arabic}
      \theoremstyle{plain}
      \theoremsymbol{\ensuremath{_\Box}}
      \theorembodyfont{\itshape}
      \theoremheaderfont{\normalfont\bfseries}
      \theoremseparator{}
      
      \newtheorem{theorem}{Theorem}
      
      \newtheorem{lemma}{Lemma}

      \theorembodyfont{\upshape}
      
      \newtheorem{definition}{Definition}

      \theoremstyle{nonumberplain}
      \theoremheaderfont{\scshape}
      \theorembodyfont{\normalfont}
      \theoremsymbol{\ensuremath{_\blacksquare}}
      \qedsymbol{\ensuremath{_\blacksquare}}
      
      \renewtheorem{proof}{Proof}

\usepackage{stmaryrd}

\usepackage{braket}
      \let\oldset=\Set
      \renewcommand{\Set}[1]{\oldset{\!\!#1\!\!}}

\usepackage{tikz}
      \usetikzlibrary{positioning,shapes}

      \definecolor{textboxcolor}{rgb}{0.8,0.8,1}
      \colorlet{feedcolor}{blue!50}
      \colorlet{attackercolor}{orange}
      \colorlet{metacolor}{blue!90}
      \colorlet{netcolor}{black}
      \colorlet{problemcolor}{red}
      \colorlet{bgcolor}{white}

      \tikzset{
            thin/.style = {line width = .04em},
            thick/.style = {line width = .08em},
            very thick/.style = {line width = .10em}, %
            figure/.style = {
                  every edge/.append style = {shorten > = .15em, shorten < = .15em,},
                  every node/.style = {draw, inner sep = .3em},
                  entity/.style = {circle, minimum size = 2.5em},
                  very thick,
                  x = \textwidth, y = \textwidth
            },
            chained/.style = {
                  start chain, node distance = #1, every node/.append style = {on chain},
            },
            chained/.default = 3em,
            textbox/.style = {fill = textboxcolor, draw, thin},
            feed/.style = {fill = feedcolor, thin, shape = single arrow, shape border rotate = 270},
            dot/.style = {circle, fill, minimum size = .5em, inner sep = 0pt, label distance = .3em}
      }

\usepackage{listings}
      \lstdefinelanguage{cryptoC}[]{C}{morekeywords = {test, load_buf, apply, store_buf, symL, symN}}      
      \lstdefinelanguage{cvm}[]{}{comment=[l]{//}, commentstyle = {\small\itshape}, fontadjust}
      \lstdefinelanguage{iml}[]{caml}{columns = fullflexible, keepspaces = true, morekeywords = {out, event, suchthat, att}}
      \lstdefinelanguage{trace}{basicstyle = \ttfamily, columns = fullflexible, keepspaces = true}

      \newcommand{\code}[1][]{\lstinline[#1, basicstyle = \ttfamily]}
      \newcommand{\mcode}[1]{\hbox{\code!#1!}}
      \let\singlebar=|
      \lstMakeShortInline[basicstyle = \ttfamily]|
      
      \newcommand{\cvm}[1][]{\lstinline[language = cvm, basicstyle = {\ttfamily\spaceskip=.4em}, #1]}
      \newcommand{\mcvm}[1]{\hbox{\cvm!#1!}}
      \newcommand{\cryptoc}[1][]{\lstinline[#1, language = cryptoC, basicstyle = \ttfamily]}
      \newcommand{\mcryptoc}[1]{\hbox{\cryptoc!#1!}}
      \newcommand{\iml}[1][]{\lstinline[#1, language = iml]}
      \newcommand{\miml}[2][]{\hbox{\iml[#1]!#2!}}

      \lstnewenvironment{lstlistingcentered}[1][]{\lstset{#1}\center\tabular{c}}{\endtabular\endcenter}

\usepackage{cleveref}
      \crefname{figure}{fig.}{figures}
      \crefname{equation}{equation}{equations}

      \newcommand{\eqlabel}[1]{\refstepcounter{equation}\label{#1}}
      \renewcommand{\eqref}{\labelcref}

\usepackage{booktabs}

\usepackage{xspace}

      \newenvironment{restate}[1]%
      {\begin{trivlist}\item[]{\normalsize\bf Restatement of #1}\hspace*{4mm}\it}%
      {\end{trivlist}}

      \def\clap#1{\hbox to 0pt{\hss#1\hss}}
      
      \def\mathrlap{\mathpalette\mathrlapinternal}

      \def\mathrlapinternal#1#2{%
      \rlap{$\mathsurround=0pt#1{#2}$}}

\newcommand{\opname}[1]{\ensuremath{\operatorname{#1}}}

\newcommand{\sA}{\mathcal{A}}

\newcommand{\sM}{\mathcal{M}}
\newcommand{\sP}{\mathcal{P}}
\newcommand{\sS}{\mathcal{S}}

\newcommand{\bC}{\mathbf{C}}
\newcommand{\bD}{\mathbf{D}}

\newcommand{\N}{\mathbb{N}}

\newcommand{\R}{\mathbb{R}}

      \newcommand{\lbars}[1]{\left\singlebar #1 \right\singlebar}

      \newcommand{\concat}{\singlebar}

      \newcommand{\dom}{\opname{dom}}

      \newcommand{\BS}{\textit{BS}}

      \newcommand{\Var}{\textit{Var}}

      \newcommand{\range}[2]{\left\{#1\right\}_{#2}}
      \newcommand{\rangesmash}[2]{\{#1\}_{#2}}
      \newcommand{\segment}[1]{#1}

      \newcommand{\sem}[1]{\llbracket #1\rrbracket}

      \newcommand{\emptybs}{\varepsilon}

      \newcommand{\pto}{\rightharpoonup} %

      \newcommand{\cvmP}{P}
      
      \newcommand{\imlP}{P}
      \newcommand{\imlQ}{Q}

      \newcommand{\hole}{[]}

      \newcommand{\exec}{\opname{Exec}}
      
      \newcommand{\ptsevents}{\opname{Events}}
      \newcommand{\insec}{\opname{insec}}

      \newcommand{\prob}[1]{\opname{Pr}[#1]}

      \newcommand{\PS}{\sP}

      \newcommand{\CVM}{\textit{CVM}}
      
      \newcommand{\bs}{\opname{bs}}
      \newcommand{\val}{\opname{val}}

      \newcommand{\Ops}{\mathbf{Ops}}
      \newcommand{\opfun}{A}
      \newcommand{\arity}{\opname{ar}}

      \newcommand{\addrspace}{\textit{Addr}} %
      \newcommand{\addr}{\opname{addr}}

      \newcommand{\var}{\opname{var}}

      \newcommand{\IML}{\textit{IML}}
      \newcommand{\IExp}{\textit{IExp}}

      \newcommand{\len}{\opname{len}}
      \newcommand{\conc}{\opname{conc}}
      \newcommand{\sub}{\opname{sub}}

      \newcommand{\SExp}{\textit{SExp}}
      \newcommand{\PBase}{\textit{PBase}}

      \newcommand{\ptr}{\opname{ptr}}
      \newcommand{\stack}{\opname{stack}}
      \newcommand{\heap}{\opname{heap}}

      \newcommand{\bop}{_b}
      \newcommand{\nop}{_\N}

      \newcommand{\memc}[1][]{\ensuremath{{\sM^c_{#1}}}}
      \newcommand{\stackc}[1][]{\ensuremath{{\sS^c_{#1}}}}
      \newcommand{\mems}[1][]{\ensuremath{{\sM^s_{#1}}}}
      \newcommand{\stacks}[1][]{\ensuremath{{\sS^s_{#1}}}}
      \newcommand{\memsc}[1][]{\ensuremath{{\sM^{sc}_{#1}}}}
      \newcommand{\stacksc}[1][]{\ensuremath{{\sS^{sc}_{#1}}}}
      \newcommand{\allocc}[1][]{{\sA^c_{#1}}}
      \newcommand{\allocs}[1][]{{\sA^s_{#1}}}
      \newcommand{\allocsc}{{\sA^{sc}}}
      \newcommand{\facts}{\Sigma}

      \newcommand{\simplify}{\opname{simplify}}
      \newcommand{\getLen}{\opname{getLen}}
      \newcommand{\apply}{\opname{apply}}

      \newcommand{\entails}{\vdash}

      \newcommand{\PExp}{\textit{PExp}}

\definecolor{dkblue}{rgb}{0,0.1,0.5}
\definecolor{dkgreen}{rgb}{0,0.4,0}
\definecolor{dkred}{rgb}{0.6,0,0}
\definecolor{linkColor}{rgb}{0,0.1,0.5}

\newcommand{\status}[1]{\ifdraft[#1]\fi}

\title{Extracting and Verifying Cryptographic Models\\ from C Protocol Code by Symbolic Execution}

\author{
Mihhail Aizatulin\\
\affaddr{The Open University}
\and
Andrew~D.~Gordon \\
\affaddr{Microsoft Research}
\and
Jan~J{\"{u}}rjens \\
\affaddr{TU Dortmund \& Fraunhofer ISST}
}

\begin{document}

\maketitle

\begin{abstract}
Consider the problem of verifying security properties of a cryptographic protocol coded in C.
We propose an automatic solution that needs neither a pre-existing protocol description nor manual annotation of source code.
First, symbolically execute the C program to obtain symbolic descriptions for the network messages sent by the protocol.
Second, apply algebraic rewriting to obtain a process calculus description.
Third, run an existing protocol analyser (ProVerif) to prove security properties or find attacks.
We formalise our algorithm and appeal to existing results for ProVerif to establish computational soundness under suitable circumstances.
We analyse only a single execution path, so our results are limited to protocols with no significant branching.
The results in this paper provide the first computationally sound verification of weak secrecy and authentication for (single execution paths of) C code.
\end{abstract}

\section{Introduction \status{good, 06.05.2011}}

Recent years have seen great progress in formal verification of cryptographic protocols, as illustrated by powerful tools like ProVerif \cite{ProVerif}, CryptoVerif \cite{DBLP:conf/sp/Blanchet06} or AVISPA \cite{AVISPA}. There remains, however, a large gap between what we verify (formal descriptions of protocols, say, in the pi calculus) and what we rely on (protocol implementations, often in low-level languages like C). The need to start the verification from C code has been recognised before and implemented in tools like CSur \cite{CSur} and ASPIER \cite{ASPIER}, but the methods proposed there are still rather limited. Consider, for example, the small piece of C code in \cref{fig:example} that checks whether a message received from the network matches a message authentication code. Intuitively, if the key is honestly chosen and kept secret from the attacker then with overwhelming probability the event will be triggered only if another honest participant (with access to the key) generated the message. Unfortunately, previous approaches cannot prove this property: the analysis of CSur is too coarse to deal with authentication properties like this and ASPIER cannot directly deal with code manipulating memory through pointers. Furthermore the previous works do not offer a definition of security directly for C code, i.e. they do not formally state what it means for a C program to satisfy a security property, which makes it difficult to evaluate their overall soundness. The goal of our work is to improve upon this situation by giving a formal definition of security straight for C code and proposing a method that can verify secrecy and authentication for typical memory-manipulating implementations like the one in \cref{fig:example} in a fully automatic and scalable manner, without relying on a pre-existing protocol specification.  

\begin{figure}
\small
      \begin{lstlisting}[language = cryptoC, gobble = 12]
            void * key; size_t keylen;
            readenv("k", &key, &keylen);
            size_t len;
            read(&len, sizeof(len));
            if(len > 1000) exit();
            void * buf = malloc(len + 2 * MAC_LEN);
            read(buf, len);
            mac(buf, len, key, keylen, buf + len);
            read(buf + len + MAC_LEN, MAC_LEN);
            if(memcmp(buf + len, 
                      buf + len + MAC_LEN, 
                      MAC_LEN) == 0)
              event("accept", buf, len);
      \end{lstlisting}
      \hrule
      \begin{lstlisting}[language = iml, gobble = 12]
            in($x_1$); in($x_2$); if $x_2 = mac(k, x_1)$ then event $accept(x_1)$
      \end{lstlisting}
\vspace{-3ex}
      \caption{An example C fragment together with the extracted model.}
\vspace{-2ex}
      \label{fig:example}
\end{figure}

Our method proceeds by extracting a high-level model from the C code that can then be verified using existing tools (we use ProVerif in our work). Currently we restrict our analysis to code in which all network outputs happen on a single execution path, but otherwise we do not require use of any specific programming style, with the aim of applying our methods to legacy implementations. In particular, we do not assume memory safety, but instead explicitly verify it during model extraction. The method still assumes that the cryptographic primitives such as encryption or hashing are implemented correctly---verification of these is difficult even when done manually \cite{ABP09}.

The two main contributions of our work are: 
\begin{itemize*}
      \item 
            formal definition of security properties for source code;
 
      \item
            an algorithm that computes a high-level model of the protocol implemented by a C program. 
\end{itemize*}

We implement and evaluate the algorithm as well as give a proof of its soundness with respect to our security definition.
Our definition of security for source code is given by linking the semantics of a programming language, expressed as a transition system, to a computational security definition in the spirit of \cite{BM84,GM84,Yao82}. We allow an arbitrary number of sessions. We restrict our definition to trace properties (such as weak secrecy or authentication), but do not consider observational equivalence (for strong secrecy, say).

Due to the complexity of the C language we give the formal semantics for a simple assembler-like language into which C code can be easily compiled, as in other symbolic execution approaches such as \cite{CM11}. The soundness of this step can be obtained by using well-known methods, as outlined in \cref{c-to-cvm}.

Our model-extraction algorithm produces a model in an intermediate language without memory access or destructive updates, while still preserving our security definition. The algorithm is based on symbolic execution \cite{Kin76} of the C program, using symbolic expressions to over-approximate the sets of values that may be stored in memory during concrete execution. The main difference from existing symbolic execution algorithms (such as \cite{KLEE} or \cite{SAGE}) is that our variables represent bitstrings of potentially unknown length, whereas in previous algorithms a single variable corresponds to a single byte.

We show how the extracted models can be further simplified into the form understood by ProVerif.  We apply the computational soundness result from \cite{CoSP} to obtain conditions where the symbolic security definition checked by ProVerif corresponds to our computational security definition. Combined with the security-preserving property of the model extraction algorithm this provides a computationally sound verification of weak secrecy and authentication for C.

\paragraph*{Outline of our Method}
The verification proceeds in several steps, as outlined in \cref{fig:outline}. The method takes as input:
\begin{itemize*}
      \item 
            the C implementations of the protocol participants, containing calls to a special function \cryptoc{event} as in \cref{fig:example},
            
      \item
            an environment process (in the modelling language) which spawns the participants, distributes keys, etc., 
            
      \item
            symbolic models of cryptographic functions used by the implementation,

      \item
            a property that event traces in the execution are supposed to satisfy with overwhelming probability.
            
\end{itemize*}

We start by compiling the program down to a simple stack-based instruction language (CVM) using CIL~\cite{CIL} to parse and simplify the C input. 
The syntax and semantics of CVM are presented in \cref{cvm} and the translation from C to CVM is informally described in \cref{c-to-cvm}. 

In the next step we symbolically execute CVM programs to eliminate memory accesses and destructive updates, thus obtaining an equivalent program in an intermediate model language (IML)---a version of the applied pi calculus extended with bitstring manipulation primitives. For each allocated memory area the symbolic execution stores an expression describing how the contents of the memory area have been computed. For instance a certain memory area might be associated with an expression $hmac(01\concat x, k)$, where $x$ is known to originate from the network, $k$ is known to be an environment variable, and $\concat$ denotes concatenation. The symbolic execution does not enter the functions that implement the cryptographic primitives, it uses the provided symbolic models instead. These models thus form the trusted base of the verification. An example of the symbolic execution output is shown at the bottom of \cref{fig:example}. We define the syntax and semantics of IML in \cref{iml} and describe the symbolic execution in \cref{cvm-to-iml}. 

\begin{figure}
      \begin{center}
      \begin{tikzpicture}
      [
            figure,
            feed/.append style = {minimum height = 1.7em, xscale = 1.5},
            node distance = .6em
      ]
            \node[textbox, anchor = west] (T-C) {C source};
            \node[feed, below left = 1.7em and -2.2em of T-C.west] (A-C) {};
            \node[textbox, below = 4.4em of T-C.west, anchor = west] (T-CVM) {C virtual machine (CVM)};
            \node[feed, below left = 1.7em and -2.2em of T-CVM.west] (A-CVM) {};
            \node[textbox, below = 4.4em of T-CVM.west, anchor = west] (T2) {Intermediate model language (IML)};
            \node[feed, below left = 1.7em and -2.2em of T2.west] (A2) {};
            \node[textbox, below = 4.4em of T2.west, anchor = west] (T3) {Applied pi};
            \node[feed, below left = 1.7em and -2.2em of T3.west] (A3) {};
            \node[textbox, below = 4.4em of T3.west, anchor = west] (T4) {Verification Result};

            \node[right = 1em of A-C, draw = none] {CIL};
            \node[right = 1em of A-CVM, draw = none] {Symbolic execution};
            \node[right = 1em of A2, draw = none] {Message format abstraction};
            \node[right = 1em of A3, draw = none] {ProVerif + computational soundness};
      \end{tikzpicture}
      \end{center}
\vspace{-3ex}
      \caption{An outline of the method}
\vspace{-3ex}
      \label{fig:outline}
\end{figure}
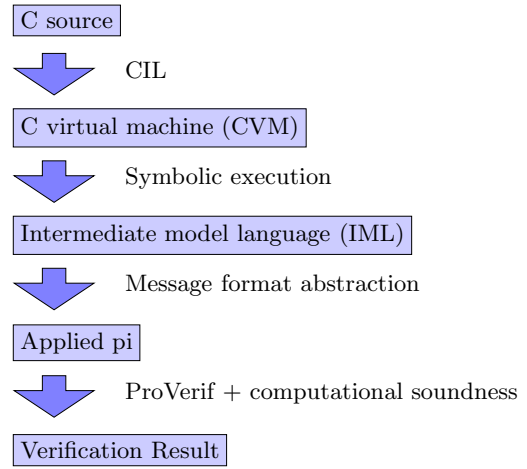

Our definition of security for source code is given in \cref{security}. The definition is generic in that it does not assume a particular programming language. We simply require that the semantics of a language is given as a set of transitions of a certain form, and define a computational execution of the resulting transition system in the presence of an attacker and the corresponding notion of security. This allows one to apply the same security definition to protocols expressed both in the low-level implementation language and in the high-level model-description language, and to formulate a correspondence between the two. 

Given that the transition systems generated by different languages are required to be of the same form, we can mix them in the same execution. This allows us to use CVM to specify a single executing participant, but at the same time use IML to describe an environment process that spawns multiple participants and allows them to interact. In particular, CVM need not be concerned with concurrency, thus making symbolic execution easier. Given an environment process $P_E$ with $n$ holes, we write $P_E[P_1, \ldots, P_n]$ for a process where the $i$th hole is filled with $P_i$, which can be either a CVM or an IML process. The soundness result for symbolic execution (\cref{symex}) states that if $P_1, \ldots, P_n$ are CVM processes and $\tilde P_1, \ldots, \tilde P_n$ are IML models resulting from their symbolic execution then for any environment process $P_E$ the security of $P_E[\tilde P_1, \ldots, \tilde P_n]$ with respect to a trace property $\rho$ relates to the security of $P_E[P_1, \ldots, P_n]$ with respect to $\rho$.

To verify the security of an IML process, we replace its bitstring-manipulating expressions by applications of constructor and destructor functions, thus obtaining a process in the applied pi-calculus (the version proposed in \cite{BAF08} and augmented with events). We can then apply a computational soundness result, such as the one from \cite{CoSP}, to specify conditions under which such a substitution is computationally sound: if the resulting pi calculus process is secure in a symbolic model (as can be checked by ProVerif) then it is asymptotically secure with respect to our computational notion of security. The correctness of translation from IML to pi is captured by \cref{transsound} and the computational soundness for resulting pi processes is captured by \cref{compsound}.  The verification of IML
(and these two theorems in particular) is described in \cref{iml-verification}.

\paragraph*{Theoretical and Practical Evaluation}

\Cref{symex,transsound,compsound} establish the correctness of our approach. In a nutshell, their significance is as follows: given implementations $P_1, \ldots, P_n$ of protocol participants in CVM, which are automatically obtained from the corresponding C code, and an IML process $P_E$ that describes an execution environment, if $P_1, \ldots, P_n$ are successfully symbolically executed with resulting models $\tilde P_1, \ldots, \tilde P_n$, the IML process $P_E[\tilde P_1, \ldots, \tilde P_n]$ is successfully translated to a pi process $P_\pi$, and ProVerif successfully verifies $P_\pi$ against a trace property $\rho$ then $P_1, \ldots, P_n$ form a secure protocol implementation with respect to the environment $P_E$ and property $\rho$. 

We are aiming to apply our method to large legacy code bases like OpenSSL. As a step towards this goal we evaluated it on a range of protocol implementations, including recent code for smart electricity meters \cite{RD10}. We were able to find bugs in preexisting implementations or to verify them without having to modify the code. \Cref{implementation} provides details.

The current restriction of analysis to a single execution path may seem prohibitive at first sight.
In fact, a great majority of protocols (such as those in the extensive SPORE repository~\cite{spore07})
follow a fixed narration of messages between participants, where any deviation from the expected message leads to termination.
For such protocols, our method allows us to capture and analyse the fixed narration directly from the C code.
In the future we plan to extend the analysis to more sophisticated control flow.

\paragraph*{Related Work}
We mention particularly relevant works here and provide a broader survey in \cref{related}. One of the first attempts at cryptographic verification of C code is contained in \cite{CSur}, where a C program is used to generate a set of Horn clauses that are then solved using a theorem prover. The method is implemented in the tool CSur. We improve upon CSur in two ways in particular.

First, we have an explicit attacker model with a standard computational attacker. The attacker in CSur is essentially symbolic---it is allowed to apply cryptographic operations, but cannot perform any arithmetic computations.

Second, we handle authentication properties in addition to secrecy properties. Adding authentication to CSur would be non-trivial, due to a rather coarse over-approximation of C code. For instance, the order of instructions in CSur is ignored, and writing a single byte into an array with unknown length is treated the same as overwriting the whole array. Authentication, however, crucially depends on the order of events in the execution trace as well as making sure that the authenticity of a whole message is preserved and not only of a single byte of it.

ASPIER \cite{ASPIER} uses model checking to verify implementations of cryptographic protocols. The model checking operates on a protocol description language, which is rather more abstract than C; for instance, it does not contain pointers and cannot express variable message lengths. The translation from C to the protocol language is not described in the paper. Our method applies directly to C code with pointers, so that we expect it to provide much greater automation.

Corin and Manzano \cite{CM11} report an extension of the KLEE test-generation tool \cite{KLEE} that allows KLEE to be applied to cryptographic protocol implementations (but not to extract models, as in our work). They do not extend the class of properties that KLEE is able to test for; in particular, testing for trace properties is not yet supported. Similarly to our work, KLEE is based on symbolic execution; the main difference is that \cite{CM11} treats every byte in a memory buffer separately and thus only supports buffers of fixed length.

\begin{FULL}
An appendix includes proofs for all the results stated in this paper. 
\end{FULL}

\begin{SHORT}
Finally, an online technical report includes more details and proofs of all results stated in this paper \cite{AGJ11long}.
\end{SHORT}

\section{C Virtual Machine (CVM) \status{good, 06.05.2011}}\label{cvm}

This section describes our low-level source language CVM (C Virtual Machine). The language is simple enough to formalise, while at the same time the operations of CVM are closely aligned with the operations performed by C programs, so that it is easy to translate from C to CVM. We shall describe such a translation informally in \cref{c-to-cvm}. 

The model of execution of CVM is a stack-based machine with random memory access. All operations with values are performed on the stack, and values can be loaded from memory and stored back to memory. The language contains primitive operations that are necessary for implementing security protocols: reading values from the network or the execution environment, choosing random values, writing values to the network and signalling events. The only kind of conditional that CVM supports is a testing operation that checks a boolean condition and aborts execution immediately if it is not satisfied.

The fact that CVM permits no looping or recursion in the program allows us to inline all function calls, so that we do not need to add a call operation to the language itself. For simplicity of presentation we omit some aspects of the C language that are not essential for describing the approach, such as global variable initialisation and structures. We also restrict program variables to all be of the same size: for the rest of the paper we choose a fixed but arbitrary $N \in \N$ and assume $\mcryptoc{sizeof($v$)} = N$ for all program variables $v$. Our implementation does not have these restrictions and deals with the full C language.

Let $BS = \{0, 1\}^*$ be the set of finite bitstrings with the empty bitstring denoted by $\emptybs$. For a bitstring $b$ let $\lbars{b}$ be the length of $b$ in bits. Let $\Var$ be a countably infinite set of variables. We write $f \colon X \pto Y$ to denote a partial function and let $\dom(f) \subseteq X$ be the set of $x$ for which $f(x)$ is defined. We write $f(x) = \bot$ when $f$ is not defined on $x$ and use the notation $f\{x \mapsto a\}$ to update functions.

Let $\Ops$ be a finite set of operation symbols such that each $op \in \Ops$ has an associated arity $\arity(op)$ and an efficiently computable partial function $\opfun_{op} \colon \BS^{\arity(op)} \pto \BS$. The set $\Ops$ is meant to contain both the primitive operations of the language (such as the arithmetic or comparison operators of C) and the cryptographic primitives that are used by the implementation.
The security definitions of this paper (given later) assume an arbitrary security parameter.
Since real-life cryptographic protocols are typically designed and implemented for a fixed value of the security parameter,
for the rest of the paper we let $k_0 \in \N$ be the security parameter with respect to which the operations in $\Ops$ are chosen.

\begin{figure}
\small
\vspace{-1ex}
\begin{align*}
      \mathrlap{b \in \BS,\, v \in \Var,\, op \in \Ops}\quad\;\; \\
      src & ::= \mcode{read} \mid \mcode{rnd} \hspace{2\bigskipamount} && \hspace{-\bigskipamount}\text{input source} \\
      dest & ::= \mcode{write} \mid \mcode{event} \hspace{2\bigskipamount} && \hspace{-\bigskipamount}\text{output destination} \\
      instr & ::= && \hspace{-\bigskipamount}\text{instruction} \\
            & \mcode{Const}\ b && \text{constant value}\\
            & \mcode{Ref}\ v && \text{pointer to variable} \\
            & \mcode{Malloc} && \text{pointer to fresh memory} \\
            & \mcode{Load} && \text{load from memory} \\
            & \mcode{In}\ v\ src && \text{input} \\
            & \mcode{Env}\ v && \text{environment variable} \\
            & \mcode{Apply}\ op && \text{operation} \\
            & \mcode{Out}\ dest && \text{output} \\
            & \mcode{Test} && \text{test a condition} \\
            & \mcode{Store} && \text{write to memory} \\
      P \in \CVM & ::= \{ instr \mcode{;} \}^* && \hspace{-\bigskipamount}\text{program}
\end{align*}
\vspace{-5ex}
\caption{The syntax of CVM.}
\vspace{-2ex}
\label{fig:CVM-syntax}
\end{figure}

A CVM program is simply a sequence of instructions, as shown in \cref{fig:CVM-syntax}. To define the semantics of CVM we choose two functions that relate bitstrings to integer values, $\val \colon \BS \to \N$ and $\bs \colon \N \to \BS$ and require that for $n < \segment{2^{N}}$ the value $\bs(n)$ is a bitstring of length $N$ such that $\val(\bs(n)) = n$. We allow $\bs$ to have arbitrary behaviour for larger numbers. The functions $\val$ and $\bs$ encapsulate architecture-specific details of integer representation such as the endianness. Even though these functions capture an unsigned interpretation of bitstrings, we only use them when accessing memory cells and otherwise place no restriction on how the bitstrings are interpreted by the program operations. For instance, the set $\Ops$ can contain both a signed and an unsigned arithmetic and comparison operators. Bitstring representations of integer constants shall be written as $i1,\, i20$, etc, for instance, $i10 = \bs(10)$. 

We let $\addrspace = \{1, \ldots, 2^N - 1\}$ be the set of valid memory addresses. The reason we exclude $0$ is to allow the length of the memory to be represented in $N$ bits. The semantic configurations of CVM are of the form $(\allocc, \memc, \stackc, \cvmP)$, where
\begin{itemize*}
      \item 
            $\memc \colon \addrspace \pto \{0, 1\}$ is a partial function that represents concrete memory and is undefined for uninitialised cells, 

      \item
            $\allocc \subseteq \addrspace$ is the set of allocated memory addresses,

      \item
            $\stackc$ is a list of bitstrings representing the execution stack,

      \item
            $\cvmP \in \CVM$ is the executing program.
\end{itemize*}
Semantic transitions are of the form $(\eta,s) \xrightarrow{l} (\eta', s')$, where $s$ and $s'$ are semantic configurations, $\eta$ and $\eta'$ are environments (mappings from variables to bitstrings) and $l$ is a protocol action such as reading or writing values from the attacker or a random number generator, or raising events.
\begin{FULL}
The formal semantics of CVM is given in \cref{cvm-semantics}, in this section we give an informal overview.
\end{FULL}
\begin{SHORT}
We give an informal overview of the semantics of CVM; the details are in the technical report \cite{AGJ11long}.
\end{SHORT}
Before the program is executed, each referenced variable $v$ is allocated an address $\addr(v)$ in $\memc$ such that all allocations are non-overlapping. If the program contains too many variables to fit in memory, the execution does not proceed. Next, the instructions in the program are executed one by one as described below. For $a, b \in \N$ we define $\range{a}{b} = \{a, \ldots, a + b - 1\}$.
\begin{itemize*}
      \item
            \cvm{Const $b$} places $b$ on the stack.
            
      \item
            \cvm{Ref $v$} places $\bs(\addr(v))$ on the stack.
            
      \item
            \cvm{Malloc} takes a value $s$ from the stack, reads a value $p$ from the attacker, and if the range $\range{\val(p)}{\val(s)}$ does not contain allocated cells, it becomes allocated and the value $p$ is placed on the stack. Thus the attacker gets to choose the beginning of the allocated memory area.
            
      \item
            \cvm{Load} takes values $l$ and $p$ from the stack. In case $\range{\val(p)}{\val(l)}$ is a completely initialised range in memory, the contents of that range are placed on the stack. In case some of the bits are not initialised, the value for those bits is read from the attacker.
            
      \item
            \cvm{In $v$ read} or \cvm{In $v$ rnd} takes a value $l$ from the stack. \cvm{In $v$ read} reads a value of length $\val(l)$ from the attacker and \cvm{In $v$ rnd} requests a random value of length $\val(l)$. The resulting value $b$ is then placed on the stack. The environment $\eta$ is extended by the binding $v \mapsto b$. %
            
      \item
            \cvm{Env $v$} places $\eta(v)$ and $\bs(\lbars{\eta(v)})$ on the stack.
            
      \item
            \cvm{Apply $op$} with $\arity(op) = n$ applies $\opfun_{op}$ to $n$ values on the stack, replacing them by the result.
            
      \item
            \cvm{Out write} sends the top of the stack to the attacker and \cvm{Out event} raises an event with the top of the stack as payload. Events with multiple arguments can be represented using a suitable bitstring pairing operation. Both commands remove the top of the stack.
            
      \item
            \cvm{Test} takes the top of the stack and checks whether it is $i1$. If yes, the execution proceeds, otherwise it stops.
            
      \item
            \cvm{Store} takes values $p$ and $b$ from the stack and writes $b$ into memory at position starting with $\val(p)$.
\end{itemize*}

\begin{FULL}
The execution of a program can get stuck if rule conditions are violated, for instance, when the program runs out of memory or attempts to write to uninitialised memory. All these situations would likely result in a crash in a real system. Our work is not focused on preventing crashes, but rather on analysing the sequences of events that occur before the program terminates (either normally or abnormally). Thus we leave crashes implicit in the semantics. An exception is the instruction \cvm{Load}: reading uninitialised memory is unlikely to result in a crash in reality, instead it can silently return any value. We model this behaviour explicitly in the semantics.
\end{FULL}

\section{From C to CVM \status{good, 06.05.2011}}\label{c-to-cvm}

We describe how to translate from C to CVM programs. We start with aspects of the translation that are particular to our approach, after which we illustrate the translation by applying it to the example program in \cref{fig:example}. 

Proving correctness of C compilation is not the main focus of our work, so we trust compilation for now. To prove correctness formally one would need to show that a CVM translation simulates the original C program; 
\begin{SHORT}
an appropriate notion of simulation is defined in the technical report \cite{AGJ11long} and is used to prove soundness of other verification steps. 
\end{SHORT}
\begin{FULL}
an appropriate notion of simulation is defined in \cref{pts} and is used to prove soundness of other verification steps.
\end{FULL}
We believe that work on proving correctness of the CompCert compiler \cite{Compcert} can be reused in this context.

We require that the C program contains no form of looping or function call cycles %
and that all actions of the program (either network outputs or events) happen in the same path (called \emph{main path} in the following). We then prune all other paths by replacing if-statements on the main path by test statements: a statement %
\cryptoc{if(cond) t_block else f_block}
is replaced by \cryptoc{test(cond); t_block} in case the main path continues in the \cryptoc{t_block}, and by \cryptoc{test(!cond); f_block} otherwise. The test statements are then compiled to CVM \cvm{Test} instructions. The main path can be easily identified by static analysis; for now we simply identify the path to be compiled by observing an execution of the program. 

As mentioned in the introduction, we do not verify the source code of cryptographic functions, but instead trust that they implement the cryptographic algorithms correctly. Similarly, we would not be able to translate the source code of functions like \cryptoc{memcmp} into CVM directly, as these functions contain loops. Thus for the purpose of CVM translation we provide an abstraction for these functions. We do so by writing what we call a \emph{proxy function} \cryptoc{f_proxy} for each function \cryptoc{f} that needs to be abstracted. Whenever a call to \cryptoc{f} is encountered during the translation, it is replaced by the call to \cryptoc{f_proxy}. The proxy functions form the trusted base of the verification.

\begin{figure}
\small
      \begin{lstlisting}[language = cryptoC, gobble = 12]
            void mac_proxy(void * buf, size_t buflen, 
                           void * key, size_t keylen, 
                           void * mac){
              load_buf(buf, buflen);
              load_buf(key, keylen);
              apply("mac", 2);
              store_buf(mac);
            }
           
            int memcmp_proxy(void * a, void * b, 
                             size_t len){
              int ret;
              load_buf(a, len);
              load_buf(b, len);
              apply("cmp", 2);
              store_buf(&ret);
              return ret;
            }
      \end{lstlisting}
      
\vspace{-4ex}
      \caption{Examples of proxy functions.}
\vspace{-2ex}
      \label{fig:proxies}
\end{figure}

Examples of proxy functions are shown in \cref{fig:proxies}. The functions \cryptoc{load_buf}, \cryptoc{apply} and \cryptoc{store_buf} are treated specially by the translation. For instance, assuming an architecture with $N = 32$, a call \cryptoc{load_buf(buf, len)}  directly generates the sequence of instructions:
\begin{lstlisting}[language = cvm, gobble = 6, xleftmargin = 1em]
      Ref buf; Const i32; Load;
      Ref len; Const i32; Load; Load;
\end{lstlisting}
Similarly we provide proxies for all other special functions in the example program, such as \cryptoc{readenv}, \cryptoc{read}, \cryptoc{write} or \cryptoc{event}. The proxies essentially list the CVM instructions that need to be generated. 

\begin{FULL}
\Cref{nsl-proxies} shows more examples of proxy functions.
\end{FULL}
\Cref{c-to-cvm-example} shows the CVM translation of our example C program in \cref{fig:example}.

\section{Intermediate Model Language\status{good, 06.05.2011}}\label{iml}

This section presents the intermediate model language (IML) that we use both to express the models extracted from CVM programs and to describe the environment in which the protocol participants execute. IML borrows most of its structure from the pi calculus \cite{AF01,BAF08}. In addition it has access both to the set $\Ops$ of operations used by CVM programs and to primitive operations on bitstrings: concatenation, substring extraction, and computing lengths of bitstrings. Unlike CVM, IML does not access memory or perform destructive updates.

The syntax of IML is presented in \cref{fig:IML-syntax}. In contrast to the standard pi calculus we do not include channel names, but implicitly use a single public channel instead. This corresponds to our security model in which all communication happens through the attacker. The nonce sampling operation $(\nu x[e])$ takes an expression as a parameter that specifies the length of the nonce to be sampled---this is necessary in the computational setting in order to obtain a probability distribution. We introduce a special abbreviation for programs that choose randomness of length equal to the security parameter $k_0$ introduced in \cref{cvm}: let $(\tilde \nu x);\; P$ stand for \iml{$(\nu \tilde x[k_0])$; let $x = nonce(\tilde x)$ in $P$}, where $nonce \in \Ops$. Using $nonce$ allows us to have tagged nonces, which will be necessary to link to the pi calculus semantics from \cite{CoSP}.

\begin{figure}[t]
\small
\vspace{-2ex}
\begin{align*}
      \mathrlap{b \in \BS,\, x \in \Var,\, op \in \Ops}\quad\quad\quad \\
      e \in \IExp & ::= && \hspace{-\bigskipamount}\text{expression} \\
            & b && \text{concrete bitstring} \\
            & x && \text{variable} \\
            & op(e_1, \ldots, e_n) && \text{computation} \\
            & e_1 \concat e_2 && \text{concatenation} \\
            & e\{e_o, e_l\} && \text{substring extraction} \\
            & \len(e) && \text{length} \\
      \imlP,\, \imlQ \in \IML& ::= && \hspace{-\bigskipamount}\text{process} \\
            & 0 && \text{nil} \\
            & !\imlP && \text{replication} \\
            & \imlP \concat \imlQ && \text{parallel composition} \\
            & \miml{$(\nu x[e])$;\ $\imlP$} && \text{randomness}\\
            & \miml{in$(x)$;\ $\imlP$} && \text{input} \\
            & \miml{out$(e)$;\ $\imlP$} && \text{output} \\
            & \miml{event$(e)$;\ $P$} && \text{event} \\
            & \miml{if\ $e$ then\ $\imlP$ [else\ $\imlQ$]} && \text{conditional} \\
            & \miml{let\ $x = e$ in\ $\imlP$ [else\ $\imlQ$]} && \text{evaluation} 
\end{align*}
\vspace{-5ex}
\caption{The syntax of IML.}
\vspace{-2ex}
\label{fig:IML-syntax}
\end{figure}

\begin{figure}[t]
\small
\begin{align*}
      &\sem{b} = b, \;\text{for $b \in \BS$,}\\
      &\sem{x} = \bot, \;\text{for $x \in \Var$,}\\
      &\sem{op(e_1, \ldots, e_n)} = \opfun_{op} (\sem{e_1}, \ldots, \sem{e_n}), \\
      &\sem{e_1\concat e_2} = \sem{e_1}\concat\sem{e_2}, \\
      &\sem{e\{e_o, e_l\}} = \sub(\sem{e}, \val(\sem{e_o}), \val(\sem{e_l})),\\
      &\sem{\len(e)} = \bs(\lbars{\sem{e}}).
\end{align*}
\vspace{-5ex}
\caption{The evaluation of IML expressions, whereby $\bot$ propagates.}
\vspace{-2ex}
\label{fig:IML-eval}
\end{figure}

For a bitstring $b$ let $b[i]$ be the $i$th bit of $b$ counting from $0$. The concatenation of two bitstrings $b_1$ and $b_2$ is written as $b_1 \concat b_2$. 

Just as for CVM, the semantics of IML is parameterised by functions $\bs$ and $\val$. The semantics of expressions is given by the partial function $\sem{\cdot} \colon \IExp \pto \BS$ described in \cref{fig:IML-eval}. The partial function $\sub\colon \BS \times \N \times \N \pto \BS$ extracts a substring of a given bitstring such that $\sub(b, o, l)$ is the substring of $b$ starting at offset $o$ of length $l$:
            \[\sub(b, o, l) = 
                  \begin{cases}
                        b[o]\ldots b[o + l - 1] & \text{if $o + l \leq \lbars{b}$, } \\
                        \bot & \text{otherwise.}
                  \end{cases}
            \]
For a \emph{valuation} $\eta \colon \Var \pto \BS$ we denote with $\sem{e}_\eta$ the result of substituting all variables $v$ in $e$ by $\eta(v)$ (if defined) and then applying $\sem{\cdot}$.

\begin{FULL}
The formal semantics of IML is mostly straightforward and is shown in detail in \cref{iml-semantics}.
\end{FULL}

\begin{SHORT}
The technical report \cite{AGJ11long} describes the detailed semantics of IML, in our formalism of protocol transition systems, introduced in the next section.
\end{SHORT}

\newpage
\section{Security of Protocols \status{good, 27.05.2011}}\label{security}

\begin{FULL}
This section gives an informal overview of our security definition. 
The complete definition is given in \cref{pts}.
\end{FULL}

To define security for protocols implemented by CVM and IML programs we need to specify what a protocol is and give a mapping from programs to protocols. The notion of a protocol is formally captured by a \emph{protocol transition system (PTS)}, which describes how processes evolve and interact with the attacker. A PTS is a set of transitions of the form 
$(\eta, s) \xrightarrow{l} \{(\eta_1, s_1), \ldots, (\eta_n, s_n)\}$,
where $\eta$ and $\eta_i$ are environments (modelled as valuations), $s$ and $s_i$ are semantic configurations of the underlying programming language, and $l$ is an action label. Actions can include reading values from the attacker, generating random values, sending values to the attacker, or raising events. We call a pair $(\eta, s)$ an \emph{executing process}. Multiple processes on the right hand side capture replication.
\begin{SHORT}
Full details are in the technical report \cite{AGJ11long}.
\end{SHORT}
 
The semantics of CVM and IML are given in terms of the PTS that are implemented by programs. For a CVM program $P$ we denote with $\sem{P}_C$ the PTS that is implemented by $P$. Similarly, for an IML process $P$ the corresponding PTS is denoted by $\sem{P}_I$.

Given a PTS $T$ and a probabilistic machine $E$ (an attacker) we can execute $T$ in the presence of $E$. The state of the executing protocol is essentially a multiset of executing processes. The attacker repeatedly chooses a process from the multiset which is then allowed to perform an action according to $T$. The result of the execution is a sequence of raised events. For a resource bound $t \in \N$ we denote with $\ptsevents(T, E, t)$ the sequence of events raised during the first $t$ steps of the execution. We shall be interested in the probability that this sequence of events belongs to a certain ``safe'' set. This is formally captured by the following definition:
\begin{definition}[Protocol security]\label{secdef}
      We define a \emph{trace property} as a polynomially decidable prefix-closed set of event sequences. For a PTS $T$, a trace property $\rho$ and a resource bound $t \in \N$
      let $\insec(T,\rho,t)$ be the probability
      \[\sup \Set{\prob{\ptsevents(T, E, t) \notin \rho} | E\ \text{attacker},\, \lbars{E} \leq t},\]
      where $\lbars{E}$ measures the size of the description of the attacker. 
\end{definition}

Intuitively $\insec(T, \rho, t)$ measures the success probability of the most successful attack against $T$ and property $\rho$ when both the execution time of the attack and the size of the attacker code are bounded by $t$.

Since the semantics of CVM and IML are in the same formalism, we may combine the sets of semantic rules and obtain semantics $\sem{\cdot}_{CI}$ for mixed programs, where a CVM program can be a subprocess of a larger IML process. We add an additional syntactic form $\hole_i$ (a hole) with $i \in \N$ and no reductions to IML. For an IML process $P_E$ with $n$ holes and CVM or IML processes $P_1, \ldots, P_n$ we write $P_E[P_1, \ldots, P_n]$ to denote process $P_E$ where each hole $\hole_i$ is replaced by $P_i$. 
\begin{FULL}
The semantics of the resulting process $P_E'$, denoted with $\sem{P_E'}_{CI}$, is defined in \cref{iml-semantics}. 
\end{FULL}

Being able to embed a CVM program within an IML process is useful for modelling. As an example, let $P_1$ be the CVM program resulting from the translation of the C code in \cref{fig:example} and let $P_2$ be a description of another participant of the protocol, in either CVM or IML. Then we might be interested in the security of the following process:
\[P_E[P_1, P_2] = \;!((\tilde \nu\; k);\; ((!P_1) \concat (!P_2))).\]
A trace property $\rho$ of interest might be, for instance, {``Each event of the form $accept(x)$ is preceded by an event of the form $request(x)$''}, where $request$ is an event possibly raised in $P_2$. The goal is to obtain a statement about probability $\insec(\sem{P_E[P_1, P_2]}_{CI}, \rho, t)$ for various $t$. The next section shows how we can relate the security of $P_E[P_1, P_2]$ to the security of $P_E[\tilde P_1, P_2]$, where IML process $\tilde P_1$ is a model of the CVM process $P_1$, extracted by symbolic execution.

\section{CVM to IML: Symbolic Execution \status{good, 06.05.2011}}\label{cvm-to-iml}

\begin{figure}
\small
\vspace{-2ex}
\begin{align*}
      \mathrlap{v \in \Var,\, i \in \N}\qquad\qquad \\
      pb \in \PBase & ::= && \hspace{-\bigskipamount}\text{pointer base} \\
            & \stack\ v && \text{stack pointer to variable $v$} \\
            & \heap\ i && \text{heap pointer with id $i$} \\
      e \in \SExp & ::= && \hspace{-\bigskipamount}\text{symbolic expression} \\
            & \ptr(pb, e) && \text{pointer} \\
            & \ldots && \text{same as $\IExp$ in \cref{fig:IML-syntax}}
\end{align*}
\vspace{-4ex}
\caption{Symbolic expressions.}
\vspace{-2ex}
\label{fig:SExp}
\end{figure}

\begin{figure*}[t]
\small
\vspace{-2ex}
\begin{align*} 
      &\frac{}
        {(\mcode{Init},\, \cvmP) \rightarrow (\facts_{op},\, \Set{\stack\ v \mapsto \bs(N) | v \in \var(\cvmP)},\, \Set{\stack\ v \mapsto \emptybs | v \in \var(\cvmP)},\, [],\, \cvmP)} \tag{S-Init}\eqlabel{eq:s-init} \\[1ex]
      &\frac{}{(\facts,\, \allocs,\, \mems,\, \stacks,\, \mcode{Const}\ b;\, \cvmP) \rightarrow (\facts,\, \allocs,\, \mems,\, b :: \stacks,\, \cvmP)} \tag{S-Const}\eqlabel{eq:s-const} \\[1ex]
      &\frac{}{(\facts,\, \allocs,\, \mems,\, \stacks,\, \mcode{Ref}\ v;\, \cvmP) \rightarrow (\facts,\, \allocs,\, \mems,\, \ptr(\stack\ v,\, i0) :: \stacks,\, \cvmP)} \tag{S-Ref}\eqlabel{eq:s-ref} \\[1ex]
      &\frac{e_l \in \IExp \quad i \in \N\ \text{minimal s.t.}\ pb = \heap\ i \notin \dom(\mems)}{
                  (\facts,\, \allocs,\, \mems,\, e_l :: \stacks,\, \mcode{Malloc};\, \cvmP) \rightarrow (\facts,\, \allocs\{pb \mapsto e_l\},\, \mems\{pb \mapsto \emptybs\},\, \ptr (pb, i0) :: \stacks,\, \cvmP)
       }  
      \tag{S-Malloc}\eqlabel{eq:s-malloc} \\[1ex]
      & \frac{
                  pb \in \dom(\mems) \quad e = \simplify_\Sigma(\mems(pb)\{e_o, e_l\}) \quad \facts \entails (e_o +\nop e_l \leq \getLen(\mems(pb)))
       }{
            (\facts,\, \allocs,\, \mems,\, e_l :: \ptr(pb, e_o) :: \stacks,\, \mcode{Load};\, \cvmP) \rightarrow (\facts,\, \allocs,\, \mems,\, e :: \stacks,\, \cvmP)
       }
        \tag{S-Load}\eqlabel{eq:s-load} \\[1ex]
      & \frac{
                  e_l \in \IExp \quad l = (\text{if $src = \mcode{read}$ then $\miml{in}(v);$ else $(\nu v[e_l]);$})
      }{
                  (\facts,\, \allocs,\, \mems,\, e_l :: \stacks,\, \mcode{In}\ v\ src;\, \cvmP) \xrightarrow{l} (\facts \cup \{\len(v) = e_l\},\, \allocs,\, \mems,\, v :: \stacks,\, \cvmP)
      }
        \tag{S-In}\eqlabel{eq:s-in} \\[1ex]
      & \frac{}{
                  (\facts,\, \allocs,\, \mems,\, \stacks,\, \mcode{Env}\ v;\, \cvmP) \rightarrow (\facts,\, \allocs,\, \mems,\, \len(v) :: v :: \stacks,\, \cvmP)
      }
        \tag{S-Env}\eqlabel{eq:s-env} \\[1ex]
      & \frac{
            \begin{aligned}
                  & e = \apply(op, e_1, \ldots, e_n) \neq \bot
            \end{aligned}
      }{
                  (\facts,\, \allocs,\, \mems,\, e_1 :: \ldots :: e_n :: \stacks,\, \mcode{Apply}\ op;\, \cvmP) \rightarrow (\facts,\, \allocs,\, \mems,\, \len(e) :: e :: \stacks,\, \cvmP)
      }
        \tag{S-Apply}\eqlabel{eq:s-apply} \\[1ex]
      & \frac{e \in \IExp \quad l = (\text{if $dest = \mcode{write}$ then $\miml{out}(e);$ else $\miml{event}(e);$})}{
                  (\facts,\, \allocs,\, \mems,\, e :: \stacks,\, \mcode{Out}\ dest;\, \cvmP) \xrightarrow{l} (\facts,\, \allocs,\, \mems,\, \stacks,\, \cvmP)
      }
      \tag{S-Out}\eqlabel{eq:s-out} \\[1ex]
      & \frac{e \in \IExp}{ 
             (\facts,\, \allocs,\, \mems,\, e :: \stacks,\, \mcode{Test};\, \cvmP) \xrightarrow{\miml[basicstyle=\scriptsize]{if\ $e$ then}} (\facts \cup \{e\},\, \allocs,\, \mems,\, \stacks,\, \cvmP)
      }
        \tag{S-Test}\eqlabel{eq:s-test} \\[3ex]
      & \frac{
            \begin{aligned}
                  & e_h = \mems(pb) \neq \bot \quad e_s = \allocs(pb) \neq \bot \quad e_{lh} = \getLen(e_h) \quad e_l = \getLen(e)  \\
                  & 
                  \begin{aligned}
                  \text{either}\ & \facts \entails (e_o +\nop e_l < e_{lh})\ \text{and}\ e_h'= \simplify_\Sigma(e_h\{i0, e_o\} \concat e \concat e_h\{e_o +\nop e_l, e_{lh} -\nop (e_o +\nop e_l)\}) \\
                  \text{or}\ & \facts \entails (e_o +\nop e_l \geq e_{lh}) \wedge (e_o \leq e_{lh}) \wedge (e_o +\nop e_l \leq e_s)\ \text{and}\ e_h'= \simplify_\Sigma(e_h\{i0, e_o\} \concat e) 
                  \end{aligned}
            \end{aligned}
      }{
            \begin{aligned}
                  & (\facts,\, \allocs,\, \mems,\, \ptr(pb, e_o) :: e :: \stacks,\, \mcode{Store};\, \cvmP) \rightarrow (\facts,\, \allocs,\, \mems\{pb \mapsto e_h'\},\, \stacks,\, \cvmP)
            \end{aligned}
      }
        \tag{S-Store}\eqlabel{eq:s-store} \\
\end{align*}
\vspace{-5ex}
\caption{The symbolic execution of CVM.}
\vspace{-2ex}
\label{fig:cvm-symex}
\end{figure*}

We describe how to automatically extract an IML model from a CVM program while preserving security properties. The key idea is to execute a CVM program in a symbolic semantics, where, instead of concrete bitstrings, memory locations contain IML expressions representing the set of all possible concrete values at a given execution point.

To track the values used as pointers during CVM execution, we extend IML expressions with an additional construct, resulting in the class of \emph{symbolic expressions} shown in \cref{fig:SExp}.
An expression of the form $\ptr(pb, e_o)$ represents a pointer into the memory location identified by the \emph{pointer base} $pb$ with an offset $e_o$ relative to the beginning of the location. We require that $e_o \in \IExp$, so that pointer offsets do not contain pointers themselves. Pointer bases are of two kinds: a base of the form $\stack\; v$ represents a pointer to the program variable $v$ and a base of the form $\heap\; i$ represents the result of a \cvm{Malloc}.

Symbolic execution makes certain assumptions about the arithmetic operations that are available in $\Ops$. We assume that programs use operators for bitwise addition and subtraction (with overflow) that we shall write as $+\bop$ and $-\bop$. We also make use of addition and subtraction without overflow---the addition operator (written as $+\nop$) is expected to widen its result as necessary and the negation operator (written as $-\nop$) returns $\bot$ instead of a negative result. We assume that $\Ops$ contains comparison operators $=$, $\leq$, and $<$ such that $\opfun_{=}(a, b)$ returns $i1$ if $\val(a) = \val(b)$ and $i0$ otherwise, similarly for the other operators. This way $\leq$ and $<$ capture unsigned comparisons on bitstring values. We assume $\Ops$ contains logical connectives $\neg$ and $\vee$ that interpret $i0$ as false value and $i1$ as true value. These operators may or may not be the ones used by the program itself.

To evaluate symbolic expressions concretely, we need concrete values for pointer bases as well as concrete values for variables.
Given an \emph{extended valuation} $\eta \colon Var \cup \PBase \pto \BS$, we extend the function $\sem{\cdot}_\eta$ from \cref{fig:IML-eval} by the rule: 
\[\sem{\ptr(pb, e_o)}_\eta = \eta(pb) +\bop \sem{e_o}_\eta.\]
When applying arithmetic operations to pointers, we need to make sure that the operation is applied to the pointer offset and the base is kept intact. This behaviour is encoded by the function $\apply$, defined as follows:
\begin{align*}
      & \apply(+\bop, \ptr(pb, e_o), e) = \ptr(pb, e_o +\bop e), \\
      & \qquad \text{for $e \in \IExp$,} \\
      & \apply(-\bop, \ptr(pb, e_o), \ptr(pb, e_o')) = e_o -\bop e_o', \\
      & \apply(op, e_1, \ldots, e_n) = op(e_1, \ldots, e_n), \\
      & \qquad \text{for $e_1, \ldots, e_n \in \IExp$,} \\
      & \apply(...) = \bot, \; \text{otherwise.}
\end{align*}

As well as tracking the expressions stored in memory, we also track logical facts discovered during symbolic execution.
To record these facts, we use symbolic expressions themselves, interpreted as logical formulas with $=$, $\leq$, and $<$ as relations and $\neg$ and $\vee$ as connectives.
We allow quantifiers in formulas, with straightforward interpretation. Given a set $\facts$ of formulas and a formula $\phi$ we write $\Sigma \entails \phi$ iff for each $\Sigma$-consistent valuation $\eta$ (that is, a valuation such that $\sem{\psi}_\eta = i1$ for all $\psi \in \Sigma$) we also have $\sem{\phi}_\eta = i1$.

To check the entailment relation, our implementation relies on the SMT solver \emph{Yices} \cite{yices}, by replacing unsupported operations, such as string concatenation or substring extraction, with uninterpreted functions. This works well for our purpose---the conditions that we need to check during the symbolic execution are purely arithmetic and are supported by Yices' theory. 

The function $\getLen$ returns for each symbolic expression an expression representing its length:
\begin{align*}
      & \getLen(\ptr(\ldots)) = \bs(N), \\
      & \getLen(\len(\ldots)) = \bs(N), \\
      & \getLen(b) = \bs(\lbars{b}), \;\text{for $b \in \BS$,}\\
      & \getLen(x) = \len(x), \;\text{for $x \in \Var$,}\\
      & \getLen(op(e_1, \ldots, e_n)) = \len(op(e_1, \ldots, e_n)), \\
      & \getLen(e_1\concat e_2) = \getLen(e_1) +\nop \getLen(e_2), \\
      & \getLen(e\{e_o, e_l\}) = e_l. 
\end{align*}
We assume that the knowledge about the return lengths of operation applications is encoded in a fact set $\facts_{op}$. As an example, $\facts_{op}$ might contain the facts:
\begin{align*}
      & \forall x,y,a \colon \len(x) = a \wedge \len(y) = a \Rightarrow \len(x +\bop y) = a, \\
      & \forall x \colon \len(sha1(x)) = i20.
\end{align*}
We assume that $\facts_{op}$ is consistent: $\emptyset \entails \phi$ for all $\phi \in \facts_{op}$.

The transformations prescribed by the symbolic semantic rules would quickly lead to very large expressions. Thus the symbolic execution is parametrised by a simplification function $\simplify$ that is allowed to make use of the collected fact set $\facts$. We demand that the simplification function is sound in the following sense: for each fact set $\facts$, expression $e$ and a $\facts$-consistent valuation $\eta$ we have 
\[\sem{e}_\eta \neq \bot \Longrightarrow \sem{\simplify_\Sigma(e)}_\eta = \sem{e}_\eta.\]
\begin{FULL}
The simplifications employed in our algorithm are described in \cref{simplify}.
\end{FULL}
\begin{SHORT}
The simplifications employed in our algorithm are described in the technical report \cite{AGJ11long}.
\end{SHORT}

\begin{figure*}[t]
\small
\begin{center}
      \begin{tabular}{lllll}
               Line no. & C line  & symbolic memory updates  & new facts & generated IML line \\
            \midrule
                  1. & \cryptoc!readenv("k", &key, &keylen);! & $\stack\, key \Rightarrow \ptr(\heap\, 1,\, i0)$ & 
                  \\
                  && $\heap\, 1 \Rightarrow k$ 
                  \\
                  && $\stack\, keylen \Rightarrow \len(k)$
            \\ %
                  2. & \cryptoc!read(&len, sizeof(len));! & $\stack\, len \Rightarrow l$ & $\len(l) = iN$ & \iml!in($l$)!
            \\ %
                  3. & \cryptoc!if(len > 1000) exit();! & & $\neg (l > i1000)$
            \\ %
                  4. & \cryptoc!void * buf = malloc(len + 2 * MAC_LEN);! & $\stack\, buf \Rightarrow \ptr(\heap\, 2,\, i0)$
                  \\
                  && $ \heap\, 2 \Rightarrow \emptybs $
            \\ %
                  5. & \cryptoc!read(buf, len);! & $\heap\, 2 \Rightarrow x_1$ & $\len(x_1) = l$ & \iml!in($x_1$)! 
            \\ %
                  6. & \cryptoc!mac(buf, len, key, keylen, buf + len);! & $\heap\, 2 \Rightarrow x_1 \concat mac(k, x_1)$
            \\ %
                  7. & \cryptoc!read(buf + len + MAC_LEN, MAC_LEN);! & $\heap\, 2 \Rightarrow x_1 \concat mac(k, x_1) \concat x_2$ & $\len(x_2) = i20$ & \iml!in($x_2$)!
            \\ %
                  8. & \cryptoc!if(memcmp(...) == 0)! & & & \iml!if $mac(k, x_1) = x_2$ then! 
            \\ %
                  9. & \cryptoc!event("accept", buf, len);! & & & \iml!event $accept(x_1)$!
      \end{tabular}
\vspace{-1ex}
      \caption{Symbolic execution of the example in \cref{fig:example}.}
\vspace{-4ex}
      \label{fig:symex-example}
\end{center}      
\end{figure*}

The algorithm for symbolic execution is determined by the set of semantic rules presented in \cref{fig:cvm-symex}. The initial semantic configuration has the form $(\mcode{Init},\, \cvmP)$ with the executing program $P \in \CVM$. The other semantic configurations have the form $(\facts,\, \allocs,\, \mems,\, \stacks,\, \cvmP)$, where
\begin{itemize*}
      \item 
            $\facts \subseteq \SExp$ is a set of formulas (the path condition),

      \item
            $\allocs \colon \PBase \pto \SExp$ is the symbolic allocation table that for each memory location stores its allocated size,

      \item
            $\mems \colon \PBase \pto \SExp$ is the symbolic memory. We require that $\dom(\mems) = \dom(\allocs)$, 
            
      \item
            $\stacks$ is a list of symbolic expressions representing the execution stack,
            
      \item
            $\cvmP \in \CVM$ is the executing program.
\end{itemize*}

\begin{FULL}
The symbolic execution rules essentially mimic the rules of the concrete execution. 
\end{FULL}
The crucial rules are \eqref{eq:s-load} and \eqref{eq:s-store} that reflect the effect of storing and loading memory values on the symbolic level. The rule \eqref{eq:s-load} is quite simple---it tries to deduce from $\facts$ that the extraction is performed from a defined memory range, after which it represents the result of the extraction using an IML range expression. The rule \eqref{eq:s-store} distinguishes between two cases depending on how the expression $e$ to be stored is aligned with the expression $e_h$ that is already present in memory. If $e$ needs to be stored completely within the bounds of $e_h$ then we replace the contents of the memory location by $e_h\{\ldots\} \concat e \concat e_h\{\ldots\}$ where the first and the second range expression represent the pieces of $e_h$ that are not covered by $e$. In case $e$ needs to be stored past the end of $e_h$, the new expression is of the form $e_h\{\ldots\} \concat e$. The rule still requires that the beginning of $e$ is positioned before the end of $e_h$, and hence it is impossible to write in the middle of an uninitialised memory location. This is for simplicity of presentation---the rule used in our implementation does not have this limitation (it creates an explicit ``undefined'' expression in these cases).

Since all semantic rules are deterministic there is only one symbolic execution trace. Some semantic transition rules are labelled with parts of IML syntax. The sequence of these labels produces an IML process that simulates the behaviour of the original CVM program. Formally, for a CVM program $\cvmP$, let $L$ be the symbolic execution trace starting from the state $(\mcode{Init},\, \cvmP)$. If $L$ ends in a state with an empty program, let $\lambda_1, \ldots, \lambda_n$ be the sequence of labels of $L$ and set $\sem{\cvmP}_S = \lambda_1\ldots \lambda_n 0 \in \IML$, otherwise set $\sem{\cvmP}_S = \bot$. 

We shall say that a polynomial is \emph{fixed} iff it is independent of the arbitrary values assumed in this paper, such as $N$ or the properties of the set $\Ops$. Our main result relates the security of $\cvmP$ to the security of $\sem{\cvmP}_S$.

\begin{theorem}[Symbolic Execution is Sound]\label{symex}
      There exists a fixed polynomial $p$ such that if $P_1, \ldots, P_n$ are CVM processes and for each $i$ $\tilde P_i := \sem{P_i}_S \neq \bot$ then for any IML process $P_E$, any trace property $\rho$, and resource bound $t \in \N$:
      \begin{align*}
            &\insec(\sem{P_E[P_1, \ldots, P_n]}_{CI}, \rho,t) \\
            &\qquad \leq \insec(\sem{P_E[\tilde P_1, \ldots, \tilde P_n]}_{I}, \rho, p(t)).
      \end{align*}
\end{theorem}

The condition that $p$ is fixed is important---otherwise $p$ could be large enough to give the attacker the time to enumerate all the $2^{2^N - 1}$ memory configurations. 
\begin{FULL}
For practical use the actual shape of $p$ can be recovered from the proof of the theorem given in \cref{symex-proof}.
\end{FULL}
\begin{SHORT}
      For practical use the actual shape of $p$ can be recovered from the proof of the theorem given in the technical report \cite{AGJ11long}.
\end{SHORT}

\Cref{fig:symex-example} illustrates our method by showing how the symbolic execution proceeds for our example in \cref{fig:example}. For each line of the C program we show updates to the symbolic memory, the set of new facts, and the generated IML code if any. In our example \cryptoc{MAC_LEN} is assumed to be $20$ and $N$ is equal to \cryptoc{sizeof(size_t)}. The variables $l$, $x_1$, and $x_2$ are arbitrary fresh variables chosen during the translation from C to CVM  (see \cref{c-to-cvm-example}). Below we mention details for some particularly interesting steps (numbers correspond to line numbers in \cref{fig:symex-example}).

\begin{itemize*}
      \item[1.] 
            The call to \cryptoc{readenv} redirects to a proxy function that generates CVM instructions for retrieving the environment variable $k$ and storing it in memory. 
            
      \item[4.]
            A new empty memory location is created and the pointer to it is stored in \cryptoc{buf}. We make an entry in the allocation table $\allocs$ with the length of the new memory location ($l +\bop i2 * i20$).
            
      \item[5.]
            We check that the stored value fits within the allocated memory area, that is, $l \leq l +\bop i2 * i20$. This is in general not true due to possibility of integer overflow, but in this case succeeds due to the condition $\neg (l > i1000)$ recorded before (assuming that the maximum integer value $2^N - 1$ is much larger than $1000$). Similar checks are performed for all subsequent writes to memory.

      \item[7.]
            The memory update is performed through an intermediate pointer value of the form $\ptr(\heap\, 2,\, l +\bop i20)$. The set of collected facts is enough to deduce that this pointer points exactly at the end of $x_1 \concat mac(k, x_1)$. 

      \item[8.]
            The proxy function for \cryptoc{memcmp} extracts values $e_1 = e\{l, i20\}$ and $e_2 = e\{l +\bop i20, i20\}$, where $e$ is the contents of memory at $\heap\, 2$, and puts $cmp(e_1, e_2)$ on the stack. With the facts collected so far $e_1$ simplifies to $mac(k, x_1)$ and $e_2$ simplifies to $x_2$. With some special comprehension for the meaning of $cmp$ we generate IML \iml{if $e_1 = e_2$ then}. 
\end{itemize*}

\section{Verification of IML\status{good, 05.04.2011}}\label{iml-verification}

The symbolic model extracted in \cref{fig:symex-example} does not contain any bitstring operations, so it can readily be given to ProVerif for verification. In general this is not the case and some further simplifications are required. In a nutshell, the simplifications are based on the observation that the bitstring expressions (concatenation and substring extraction) are meant to represent pairing and projection operations, so we can replace them by new symbolic operations that behave as pairing constructs in ProVerif. We then check that the expressions indeed satisfy the algebraic properties expected of such operations.

We outline the main results regarding the translation to ProVerif.
\begin{SHORT}%
The technical report \cite{AGJ11long} contains the details.%
\end{SHORT}%
\begin{FULL}%
\Cref{iml-verification-details} contains the details.%
\end{FULL}
The pi calculus used by ProVerif can be described as a subset of IML from which the bitstring operations have been removed. Unlike CVM and IML, the semantics of pi is given with respect to an arbitrary security parameter: we write $\sem{P}_\pi^k$ for the semantics of a pi process $P$ with respect to the parameter $k \in \N$. In contrast, we consider IML as executing with respect to a fixed security parameter $k_0 \in \N$. For an IML process $P$ we specify conditions under which it is translatable to a pi process $\tilde P$.

\begin{theorem}[Soundness of the translation]\label{transsound}\ \\
      There exists a fixed polynomial $p$ such that for any $P \in \IML$ translatable to a pi process $\tilde P$, any trace property $\rho$ and resource bound $t \in \N$:
      $\insec(\sem{P}_I, \rho, t) \leq \insec(\sem{\tilde P}^{k_0}_\pi, \rho, p(t))$.
\end{theorem}

Backes et al. \cite{CoSP} provide an example of a set of operations $\Ops^S$ and a set of \emph{soundness conditions} restricting their implementations that are sufficient for establishing computational soundness. The set $\Ops^S$ contains a public key encryption operation that is required to be IND-CCA secure. The soundness result is established for the class of the so-called \emph{key-safe} processes that always use fresh randomness for encryption and key generation, only use honestly generated decryption keys and never send decryption keys around. 

\begin{theorem}[Computational soundness]\label{compsound}
      Let $P$ be a pi process using only operations in $\Ops^S$ such that the soundness conditions are satisfied. If $P$ is key-safe and symbolically secure with respect to a trace property $\rho$ (as checked by ProVerif) then for every polynomial $p$ the following function is negligible in $k$:
      $\insec(\sem{P}^k_\pi,\, \rho,\, p(k))$.
\end{theorem}

Overall, \cref{symex,transsound,compsound} can be interpreted as follows: let $P_1, \ldots, P_n$ be implementations of protocol participants in CVM and let $P_E$ be an IML process that describes an execution environment. Assume that $P_1, \ldots, P_n$ are successfully symbolically executed with resulting models $\tilde P_1, \ldots, \tilde P_n$, the IML process $P_E[\tilde P_1, \ldots, \tilde P_n]$ is successfully translated to a pi process $P_\pi$, and ProVerif successfully verifies $P_\pi$ against a trace property $\rho$. Then we know by \cref{compsound} that $P_\pi$ is a pi protocol model that is (asymptotically) secure with respect to $\rho$. By \cref{symex,transsound} we know that $P_1, \ldots, P_n$ form a secure implementation of the protocol described by $P_\pi$ for the security parameter $k_0$.

\section{Implementation \& Experiments \status{good, 27.05.2011}}\label{implementation}

We have implemented our approach and successfully tested it on several examples. Our implementation performs the conversion from C to CVM at runtime---the C program is instrumented using CIL so that it outputs its own CVM representation when run. This allows us to identify and compile the main path of the protocol easily. Apart from information about the path taken we do not use any runtime information and we plan to make the analysis fully static in future. The idea of instrumenting a program to emit a low-level set of instructions for symbolic execution at runtime as well as some initial implementation code were borrowed from the CREST symbolic execution tool \cite{BS08}.

Currently we omit certain memory safety checks and assume that there are no integer overflows. This allows us to use the more efficient theory of mathematical integers in Yices, but we are planning to move to exact bitvector treatment in future.

The implementation comprises about 4600 lines of OCaml code. The symbolic proxies for over 80 of the cryptographic functions in the OpenSSL library comprise further 2000 lines of C code.

\begin{figure}
\small
\begin{center}
      \begin{tabular}{lccccc}
                  & C LOC & \hspace{-1em} IML LOC \hspace{-1em} & outcome & result type & time \\
            \midrule
            simple mac & $\sim 250$ & $12$ & verified & symbolic & 4s \\
            RPC & $\sim 600$ & $35$ & verified & symbolic & 5s \\
            NSL & $\sim 450$ & $40$ & verified & computat. & 5s \\
            CSur & $\sim 600$ & $20$ & flaw: \cref{fig:csur-bug} & --- & 5s \\
            minexplib & $\sim 1000$ & $51$ & \hspace{-1em} flaw: \cref{fig:minexplib-bug} \hspace{-1em} & --- & 15s \\ 
      \end{tabular}
\vspace{-1ex}
      \caption{Summary of analysed implementations.}
\vspace{-5ex}
      \label{fig:experiments}
\end{center}      
\end{figure}

\begin{figure}[t]
\small
      \begin{lstlisting}[language = cryptoC, gobble = 12]
            read(conn_fd, temp, 128); 
            // BN_hex2bn expects zero-terminated string
            temp[128] = 0;
            BN_hex2bn(&cipher_2, temp);
            // decrypt and parse cipher_2 
            // to obtain message fields
      \end{lstlisting}
      \vspace{-3ex}
      \caption{A flaw in the CSur example: input may be too short.}
      \vspace{-2ex}
      \label{fig:csur-bug}
\end{figure}

\begin{figure}[t]
\small
      \begin{lstlisting}[language = cryptoC, gobble = 12]
            unsigned char session_key[256 / 8];
            ...
            // Use the 4 first bytes as a pad
            // to encrypt the reading
            encrypted_reading = 
              ((unsigned int) *session_key) ^ *reading;
      \end{lstlisting}
\vspace{-3ex}
      \caption{A flaw in the minexplib code: only one byte of the pad is used.}
\vspace{-2ex}
      \label{fig:minexplib-bug}
\end{figure}

\Cref{fig:experiments} shows a list of protocol implementations on which we tested our method. Some of the verified programs did not satisfy the conditions of computational soundness (mostly because they use cryptographic primitives other than public key encryption and signatures supported by the result that we rely on \cite{CoSP}), so we list the verification type as ``symbolic''. 

The ``simple mac'' is an implementation of a protocol similar to the example in \cref{fig:example}. RPC is an implementation of the remote procedure call protocol in \cite{BBF08} that authenticates a server response to a client using a message authentication code. It was written by a colleague without being intended for verification using our method, but we were still able to verify it without any further modifications to the code. 

The NSL example is an implementation of the Needham-Schroeder-Lowe protocol written by us to obtain a fully computationally sound verification result. 
\begin{FULL}
The implementation is designed to satisfy the soundness conditions listed in \cref{iml-verification-details} (modulo the assumption that the encryption used is indeed IND-CCA).
\end{FULL}
\begin{SHORT}
      The implementation is designed to satisfy the soundness conditions of the computational soundness result outlined in the technical report \cite{AGJ11long}.
\end{SHORT}
Masking the second participant's identity check triggers Lowe's attack \cite{Low95} as expected. 
\begin{FULL}
\Cref{nsl-example} shows the source code and the extracted models.
\end{FULL}

The CSur example is the code analysed in a predecessor paper on C verification \cite{CSur}. It is an implementation of a protocol similar to Needham-Schroeder-Lowe. During our verification attempt we discovered a flaw, shown in \cref{fig:csur-bug}: the received message in buffer \cryptoc{temp} is being converted to a \cryptoc{BIGNUM} structure \cryptoc{cipher_2} without checking that enough bytes were received. Later a \cryptoc{BIGNUM} structure derived from \cryptoc{cipher_2} is converted to a bitstring without checking that the length of the bitstring is sufficient to fill the message buffer. In both cases the code does not make sure that the information in memory actually comes from the network, which makes it impossible to prove authentication properties. The CSur example has been verified in \cite{CSur}, but only for secrecy, and secrecy is not affected by the flaw we discovered. The code reinterprets network messages as C structures (an unsafe practise due to architecture dependence), which is not yet supported by our analysis and so we were not able to verify a fixed version of it.

The minexplib example is an implementation of a privacy-friendly protocol for smart electricity meters \cite{RD10} developed at Microsoft Research. The model that we obtained uncovered a flaw shown in \cref{fig:minexplib-bug}: incorrect use of pointer dereferencing results in three bytes of each four-byte reading being sent unencrypted. We found two further flaws: one could lead to contents of uninitialised memory being sent on the network, the other resulted in $0$ being sent (and accepted) in place of the actual number of readings. All flaws have been acknowledged and fixed. An F\# implementation of the protocol has been previously verified \cite{SCF+10}, which highlights the fact that C implementations can be tricky and can easily introduce new bugs, even for correctly specified and proven protocols. The protocol uses low-level cryptographic operations such as XOR and modular exponentiation. In general it is impossible to model XOR symbolically \cite{Unr10}, so we could not use ProVerif to verify the protocol, but we are investigating the use of CryptoVerif for this purpose. 

\begin{DRAFT}
\section{Future Work [stub]}\label{future}

Control flow.  CVM vs LLVM. Observational equivalence.

Our result is merely asymptotic because of computational soundness. Would be better to try CryptoVerif for a more concrete bound.
\end{DRAFT}

\section{Related Work \status{good, 13.02.2011}}\label{related}

\begin{DRAFT}
Justify symbolic length.

Mention the difference from binary format and state machine extraction methods ().

Cedric's distributed computations work.

Mention the Prospex paper.

\url{http://academic.research.microsoft.com/Publication/13283454/a-symbolic-execution-framework-for-javascript}
\end{DRAFT}

\cite{CSur} presents the tool Csur for verifying C implementations of crypto-protocols by transforming them into a decidable subset of first-order logic. It only supports secrecy properties and relies on a Dolev-Yao attacker model. It was applied to a self-made implementation of the Needham-Schroeder protocol.
\cite{ASPIER} presents the verification framework ASPIER using predicate abstraction and model-checking which operates on a protocol description language where certain C concepts such as pointers and variable message lengths are manually abstracted away. In comparison, our method applies directly to C code including pointers and thus requires less manual effort.
\cite{JLW06} presents the C API ``DYC'' which can be used to generate executable protocol implementations of Dolev-Yao type cryptographic protocol messages. By generating constraints from those messages, one can use a constraint solver to search for attacks. The approach presents significant limitations on the C code.
\cite{udrea_rule-based_2006} reports on the Pistachio approach which verifies the conformance of an implementation with a specification of the communication protocol. It does not directly support the verification of security properties.
To prepare the ground for symbolic analysis of cryptographic protocol implementations, \cite{CM11} reports an extension of the KLEE symbolic execution tool. Cryptographic primitives can be treated as symbolic functions whose execution analysis is avoided. A security analysis is not yet supported. The main difference from our work is that \cite{CM11} treats every byte in a memory buffer separately and thus only supports buffers of fixed length.
\cite{DGJN11} shows how to adapt a general-purpose verifier to security verification of C code. This approach does not have our restriction to non-branching code, on the other hand, it requires the code to be annotated (with about one line of annotation per line of code) and works in the symbolic model, requiring the pairing and projection operations to be properly encapsulated.

There is also work on verifying implementations of security protocols in other high-level languages.
These do not compare directly to the work presented here, since our aim is in particular to be able to deal with the intricacies of a low-level language like C.
The tools FS2PV~\cite{bhargavan_verified_2006} and FS2CV translate {F\#} to the process calculi which can be verified by the tools ProVerif~\cite{DBLP:conf/csfw/Blanchet01} and CryptoVerif~\cite{DBLP:conf/sp/Blanchet06} versus symbolic and computational models, respectively. They have been applied to an implementation of TLS~\cite{bhargavan_cryptographically_2008}.
The refinement-type checker F7~\cite{BBF08} verifies security properties of {F\#} programs versus a Dolev-Yao attacker. Under certain conditions, this has been shown to be provably computationally sound~\cite{BMU10:CompSVSC,Fournet11}.
\cite{MukGorRya09} reports on a formal verification of a reference implementation of the TPM's authorization and encrypted transport session protocols in F\#. It also provides a translator from programs into the functional fragment of F\# into executable C code.
\cite{BMU10:CompSVSC} gives results on computational soundness of symbolic analysis of programs in the concurrent lambda calculus RCF. \cite{BacHriMaf11} reports on a type system for verifying crypto-protocol implementations in RCF.
With respect to Java, \cite{jrjens_security_2006} presents an approach which provides a Dolev-Yao formalization in FOL starting from the program's control-flow graph, which can then be verified for security properties with automated theorem provers for FOL (such as SPASS).
\cite{OShea08} provides an approach for translating Java implementations into formal models in the LySa process calculus in order to perform a security verification.
\cite{HubbersOoostdijkPollSPC03} presents an application of the ESC/Java2 static verifier to check conformance of JavaCard applications to protocol models. \cite{Cryptol} describes verification of cryptographic primitives implemented in a functional language Cryptol. CertiCrypt \cite{BGZ09} is a framework for writing machine-checked cryptographic proofs.

\section{Conclusion}

We presented methods and tools for the automated verification of cryptographic security properties of protocol implementations in C. More specifically, we provided a computationally sound verification of weak secrecy and authentication for (single execution paths of) C code. Despite the limitation of analysing single execution paths, the method often suffices to prove security of authentication protocols, many of which are non-branching. We plan to extend the analysis to more sophisticated control flow.

In future, we aim to provide better feedback in case verification fails. In our case this is rather easy to do as symbolic execution proceeds line by line. If a condition check fails for a certain symbolic expression, it is straightforward to print out a computation tree for the expression together with source code locations in which every node of the tree was computed. We plan to implement this feature in the future, although so far we found that manual inspection of the symbolic execution trace lets us identify problems easily.

\paragraph*{Acknowledgements}
Discussions with Bruno Blanchet,\linebreak Fran\c{c}ois Dupressoir, Bashar Nuseibeh, and Dominique Unruh were useful.
We also thank George Danezis, Fran\c{c}ois Dupressoir, and Jean Goubault-Larrecq for giving us access to the code of minexplib, RPC, and CSur, respectively.
Ricardo Corin and  Fran\c{c}ois Dupressoir commented on a draft.

\begin{DRAFT}
\section{Questions}

\begin{itemize}
      \item
            How to do starred arrows?
            
      \item
            Any good notation for ranges?
\end{itemize}
\end{DRAFT}

\bibliographystyle{abbrv}
\bibliography{CSec.M}

\appendix

\section{C to CVM---Example}\label{c-to-cvm-example}

 \Cref{fig:c-to-cvm} shows the CVM translation of the example program from \cref{fig:example}. We use abbreviations for some useful instruction sequences: we write \cvm{Clear} as an abbreviation for \cvm{Store dummy} that stores a value into an otherwise unused dummy variable. The effect of \cvm{Clear} is thus to remove one value from the stack. Often we do not need the length of the result that the instructions \cvm{Env} and \cvm{Apply} place on the stack, so we introduce the versions \cvm{Env'} and \cvm{Apply'} that discard the length: \cvm{Env' $v$} is an abbreviation for \cvm{Env $v$; Clear} and \cvm{Apply' $v$} is an abbreviation for \cvm{Apply $v$; Clear}. The abbreviation \cvm{Varsize} is supposed to load the variable width $N$ onto the stack, for instance, on an architecture with $N = 32$ the meaning of \cvm{Varsize} would be \cvm{Const i32}. For convenience we write operation arguments of \cvm{Apply} together with their arities. 

During the translation we arbitrarily choose fresh variables $l$, $x_1$, and $x_2$ for use in the \cvm{In} operations. %

\begin{figure}[h]
\small
      \begin{lstlisting}[language = cvm, gobble = 12]
            //void * key; size_t keylen;
            //readenv("k", &key, &keylen);
            Env k; Ref keylen; Store; 
            Ref keylen; Varsize; Load; Malloc;
            Ref key; Store; 
            Ref key; Varsize; Load; Store; 
            //size_t len;
            //read(&len, sizeof(len));
            Varsize; In l read; Ref len; Store;
            // if(len > 1000) exit();
            Const i1000; Ref len; Varsize; Load;
            Apply' >/2; Apply' $\neg$/1; Test;
            //void * buf = malloc(len + 2 * 20);
            Ref len; Varsize; Load;
            Const i2; Const i20;
            Apply' */2; Apply' +/2;
            Malloc; Ref buf; Store;
            //read(buf, len);
            Ref len; Varsize; Load; In x1 read;
            Ref buf; Varsize; Load; Store;
            //mac(buf, len, key, keylen, buf + len);
            Ref buf; Varsize; Load; 
            Ref len; Varsize; Load; Load;
            Ref key; Varsize; Load; 
            Ref keylen; Varsize; Load; Load;
            Apply' mac/2;
            Ref buf; Varsize; Load;
            Ref len; Varsize; Load; 
            Apply' +/2; Store;
            //read(buf + len + 20, 20);
            Const i20; In x2 read;
            Ref buf; Varsize; Load;
            Ref len; Varsize; Load;
            Const i20;
            Apply' +/2; Apply' +/2; Store;
            //if(memcmp(buf + len, 
            //         buf + len + 20, 
            //         20) == 0)
            Ref buf; Varsize; Load;
            Ref len; Varsize; Load; Apply' +/2;
            Const i20; Load;
            Ref buf; Varsize; Load;
            Ref len; Varsize; Load;
            Const i20; Apply' +/2; Apply' +/2; 
            Const i20; Load;
            Apply' cmp/2;
            Const 0; Apply' ==/2; Test;
            //  event("accept", buf, len);
            Ref buf; Varsize; Load;
            Ref len; Varsize; Load; Load;
            Event;
      \end{lstlisting}
      \caption{Translation of the example C program (\cref{fig:example}) into CVM.}
      \label{fig:c-to-cvm}
\end{figure}

\begin{FULL}

\section{Protocol Transition Systems \status{good, 27.05.2010}}\label{pts}

This section establishes the definition of security that we use in the paper and gives some sufficient conditions under which a protocol transformation (as done, for instance, by translating from a description of a protocol in C to a description in a more abstract language) preserves security.

In order to define security for a program we first need to define the protocol that the program implements. The notion of a protocol is formally captured by a \emph{protocol transition system (PTS)}, defined as follows: a PTS is a triple $(S, s_I, \to)$, where $S$ is a set, $s_I \in S$ and $\to$ is a labelled transition relation with transitions of the form
\[(\eta, s) \xrightarrow{l} \{(\eta_1, s_1), \ldots, (\eta_n, s_n)\},\]
where $\eta$ and $\eta_i$ are valuations, $s, s_i \in S$, and the right hand side is a nonempty multiset. We call a pair $(\eta, s)$ an \emph{executing process} and think of $\eta$ as an environment in which $s$ executes. We require that each executing process is of one of the following types:
\begin{itemize}
      \item 
            a \emph{reading process}, in which case all outgoing labels are of the form $\mcode{read}\ b$ with $b \in BS$,
            
      \item 
            a \emph{control process}, in which case all outgoing labels are of the form $\mcode{ctr}\ b$ with $b \in BS$,
            
      \item
            a \emph{randomising process}, in which case all outgoing labels are of the form $\mcode{rnd}\ b$ with $b \in BS$ and all $b$ have the same length,
            
      \item
            a \emph{writing process}, in which case there is a single outgoing transition with label of the form $\mcode{write}\ b$ with $b \in BS$,
            
      \item
            an \emph{event process}, in which case there is a single outgoing transition with label of the form $\mcode{event}\ b$ with $b \in BS$.
\end{itemize}
We require that the transition relation is deterministically computable: there should exist a probabilistic algorithm that
\begin{itemize}
      \item 
            given a left hand side which is a reading or a control process and a label computes the right hand side (in particular, the right hand side is uniquely determined),
            
      \item
            given a left hand side which is a randomising process chooses one of the admissible outgoing labels uniformly at random and computes the right hand side,
            
      \item
            given a left hand side which is a writing or an event process computes the outgoing label and the right hand side,
            
      \item
            given inputs for which there is no transition, or malformed inputs, returns \emph{``wrong''}.
      
\end{itemize}
The semantics of languages that we use (CVM and IML) will be given as a function from programs to PTS. 

We now define protocol states and show how they evolve. Intuitively a protocol state is just a collection of executing processes. The attacker repeatedly chooses one of the processes, which is then allowed to perform a transition according to the PTS rules. The executing processes are assigned handles so that the attacker can refer to them. A handle is a sequence of all observable transitions that have been performed by the process so far---this way the handle contains all the information that the attacker has about a process.

Formally, an \emph{observation} is either an integer or one of the reading, control, or writing labels. A \emph{process history} is a sequence of observations. A \emph{protocol state} over a PTS $T$ is a partial map from process histories to executing processes over $T$. We extend the transition relation of $T$ to a transition relation over protocol states as follows: Let $\PS$ be a protocol state and $h \in \dom(\PS)$ a process history such that $T$ contains a transition of the form
\[\PS(h) \xrightarrow{l} \{(\eta_1, s_1), \ldots, (\eta_n, s_n)\}\]
Let 
\[\PS' = \PS_{- h}\Set{h o i  \mapsto (\eta_i, s_i) | 1 \leq i \leq n},\]
where $o = l$ if $l$ is an observation and $o = \emptybs$ otherwise, and we use an abbreviation $f_{-x} = f\{x \mapsto \bot\}$. Then there is a transition $\PS \xrightarrow{c,\, a} \PS'$ between protocol states $\PS$ and $\PS'$ with a \emph{command} $c$ and an \emph{action} $a$, where 
\begin{itemize}
      \item 
            $c = (h, l)$ and $a = \emptybs$ if $l$ is a control label, or a read label,
            
      \item
            $c = (h, \emptybs)$ and $a = l$ if $l$ is a a randomising label, a write label, or an event label.
\end{itemize}
Given an initial protocol state and a command, the action and the resulting state are computable by the assumption that the underlying PTS transitions are computable. We extend the definition to multiple transitions and write $\PS \xrightarrow{c_1\ldots c_n,\, a_1 \ldots a_m}{\!\!\!}^*{\,\,} \PS'$, iff there is a sequence of transitions leading from $\PS$ to $\PS'$ with commands $c_1, \ldots, c_n$ and actions $a_1, \ldots, a_m$.

We shall be interested in the sequence of events raised by a protocol in the presence of an attacker. The execution of a protocol is defined as follows:

\begin{definition}[Protocol execution]\label{compex}
      Given a PTS $T = (S, s_I, \to)$ and an interactive probabilistic machine $E$ (an attacker) we define the execution of the protocol $T$ as a probabilistic machine $\exec(T, E)$ that proceeds as follows:
      
      Maintain a protocol state $\PS$. Initially $\PS = \{\emptybs \mapsto (\emptyset, s_I)\}$.
      Keep receiving commands from the attacker and for each command $c$
      \begin{itemize}
            \item
                  compute a transition $\PS \xrightarrow{c,\, a} \PS'$ and set $\PS := \PS'$. If no such transition exists or if the command is malformed, terminate,

            \item
                  if $a = \mcode{write}\ b$, send $b$ to the attacker,
                  
            \item
                  if $a = \mcode{event}\ b$, raise event $b$.
      \end{itemize}
\end{definition}

We shall assume that $\exec(T, E)$ uses the most efficient algorithm to compute the PTS transitions. For a PTS $T$, an attacker $E$, and a resource bound $t \in \N$ let $\ptsevents(T, E, t)$ be the sequence of events raised by the execution of $\exec(T, E)$ during the first $t$ elementary computation steps (each protocol transition will typically involve multiple steps). We define a \emph{trace property} as a polynomially decidable prefix-closed set of event sequences. This leads us to the definition of security for protocols:

\begin{restate}{\cref{secdef}}
      For a PTS $T$, a trace property $\rho$ and a resource bound $t \in \N$ let
      \begin{align*}
            &\insec(T,\rho,t) \\
            &\;\; = \sup \Set{\prob{\ptsevents(T, E, t) \notin \rho} | E\ \text{attacker},\, \lbars{E} \leq t},
      \end{align*}
      where $\lbars{E}$ is the size of the description of the attacker. 
\end{restate}

Intuitively $\insec(T, \rho, t)$ measures the success probability of the most successful attack against $T$ and property $\rho$ when both the execution time of the attack and the size of the attacker code are bounded by $t$.

In the following we define a simulation relation on PTS that preserves security. This relation will be used as a tool to relate the security of a protocol described by a low-level CVM program $P$ to the security of a protocol described by a more abstract IML process $\tilde P$ that results from the symbolic execution of $P$. 

In the definition of the simulation relation we shall refer to a slightly generalised notion of the protocol execution, parameterised by the initial environment: given a PTS $T$ with an initial state $s_I$, an attacker $E$, and a valuation $\eta$, let $\exec_\eta(T, E)$ be the machine that executes like $\exec(T, E)$, but starts with $\{\emptybs \mapsto (\eta, s_I)\}$ as the initial state. 

We shall be interested in PTS in which the number of steps to reach a certain state is independent of how it is reached, as captured by the following definition:

\begin{definition}[History-independent PTS]\label{history-independent}
      A PTS $T$ is called \emph{history-independent} iff, whenever for any valuation $\eta$ and attackers $E$ and $\tilde E$ the machine $\exec_\eta(T, E)$ reaches a protocol state $\PS$ in $t$ non-attacker steps and the machine $\exec_\eta(T, \tilde E)$ reaches $\PS$ in $\tilde t$ non-attacker steps, $\tilde t = t$.
      
      Given a history-independent PTS $T$, a protocol state $\PS$ over $T$ and a valuation $\eta$ we say that \emph{$T$ reaches $\PS$ from $\eta$ in $t$ steps} iff $t$ is the number of non-attacker steps in which $\exec_\eta(T, E)$ reaches $\PS$ for some attacker $E$.
      
\end{definition}

For the PTS defined in this paper we shall ensure history-independence by recording enough information in the state to be able to reconstruct the set of transitions that lead into that state. 

Intuitively we shall say that a PTS $\tilde T$ simulates a PTS $T$ when an attacker has a way of playing against $\tilde T$ in such a way that it solicits the same sequence of actions as when playing against $T$. In other words, given an execution trace of $T$, it should be feasible to reconstruct an execution trace of $\tilde T$ with the same sequence of actions. The only additional complication is that the reconstruction should happen on-line, that is, the translation of a prefix of a trace should not depend on what follows the prefix. This corresponds to the fact that the attacker cannot see into the future. We achieve the on-line property by demanding that there is an equivalence relation between protocol states of $T$ and $\tilde T$ such that for each transition from $\PS$ to $\PS'$ in $T$ and a state $\tilde \PS$ equivalent to $\PS$ there is a transition from $\tilde \PS$ to $\tilde \PS'$ in $\tilde T$ such that $\tilde \PS'$ is equivalent to $\PS'$. Most of the technicalities of the definition deal with placing restrictions on the computability of these transitions.

\begin{definition}[Simulation relation on PTS]\label{simulation}
      For a polynomial $p$ we say that PTS $\tilde T$ with initial state $\tilde s_I$ \emph{$p$-simulates} a PTS $T$ with initial state $s_I$, writing $T \lesssim_p \tilde T$ iff both $T$ and $\tilde T$ are history-independent and there exists a relation $\lesssim$ between protocol states of $T$ and protocol states of $\tilde T$ and a partial map $\tau$ from commands to sequences of commands such that
      \begin{enumerate}
            \item \label{simulation-initial}
                  for all valuations $\eta$ 
                  \[\{\emptybs \mapsto (\eta, s_I)\} \lesssim \{\emptybs \mapsto (\eta, \tilde s_I)\},\]
                  
            \item \label{simulation-step}
                  if $\PS \lesssim \tilde \PS$ and there exists a transition $\PS \xrightarrow{c,\, a} \PS'$ with a command $c$ and an action $a$ then there exists a protocol state $\tilde \PS'$ of $\tilde T$ such that $\PS' \lesssim \tilde \PS'$ and $\tilde \PS \xrightarrow{\tau(c),\, a}{\!\!\!}^*{\,\,\,} \tilde \PS'$,

            \item \label{simulation-e-efficiency}
                  $\tau(c)$ is computable in $p(\lbars{c} + \lbars{s_I})$ steps,
                  
            \item \label{simulation-t-efficiency}
                  if $\PS \lesssim \tilde \PS$ and for some valuation $\eta$ $T$ reaches $\PS$ from $\eta$ in $t$ steps and $\tilde T$ reaches $\tilde \PS$ from $\eta$ in $\tilde t$ steps then $\tilde t \leq p(t)$.
      \end{enumerate}
\end{definition}

\begin{theorem}[Preservation of security by simulation]\label{simsec}
      For every polynomial $p$ there exists a polynomial $p'$ such that whenever $T \lesssim_p \tilde T$ for PTS $T$ and $\tilde T$, for any trace property $\rho$ and resource bound $t \in \N$
      \[\insec(T, \rho, t) \leq \insec(\tilde T, \rho, p'(t)).\]
\end{theorem}

\begin{proof}
Let $T \lesssim_p \tilde T$ for PTS $T$ and $\tilde T$ and a polynomial $p$. Given an attacker $E$ we shall construct an attacker $\tilde E$ such that whenever the machine $\exec(T, E)$ produces a sequence of events $es$ within the first $t$ steps when running with random tape $R$, the machine $\exec(\tilde T, \tilde E)$ produces the sequence $es$ within at most $p'(t)$ steps when running with $R$, where $p'$ is a polynomial depending on $p$. Thus, given that $\rho$ is defined to be prefix-closed, any violation of $\rho$ happening in $T$ will happen in $\tilde T$ with at least the same probability.

The attacker $\tilde E$ shall run an instance of $E$ and iterate as follows: 
\begin{itemize}
      \item
            Receive a sequence $c_1\ldots c_m$ of commands from $E$ and output $\tau(c_1)\ldots \tau(c_m)$,
      
      \item
            Forward any input to $E$.
\end{itemize}

Let $M$ be the state of the machine $\exec(T, E)$ running with random tape $R$ after having processed commands $c_1\ldots c_n$ and $\tilde M$ the state of the machine $\exec(\tilde T, \tilde E)$ running with $R$ after having processed commands $\tau(c_1)\ldots \tau(c_n)$. By induction using (\ref{simulation-initial})--(\ref{simulation-step}) in \cref{simulation} we can show:
\begin{itemize}
      \item 
            if $\PS$ is the protocol state contained in $M$ and $\tilde \PS$ is the protocol state contained in $\tilde M$ then $\PS \lesssim \tilde \PS$, 
            
      \item
            the instance of $E$ run by $\tilde E$ in $\tilde M$ has executed the same computations as the instance of $E$ in $M$, the same portion of $R$ has been consumed, and the same sequence of events has been raised.
\end{itemize}

To bound the execution time of $\tilde M$ assume that $M$ has executed $t$ steps and $\tilde M$ has executed $\tilde t$ steps. Let $t = t_E + t_T$, where $t_E$ is the number of steps executed by the attacker and $t_T$ is the number of non-attacker steps. Similarly split $\tilde t = \tilde t_E + \tilde t_T$. The attacker $\tilde E$ runs an instance of $E$ which takes time $O(t_E)$ and additionally issues $n$ queries to $\tau$. According to (\ref{simulation-e-efficiency}) in \cref{simulation} the runtime of each query is bounded by $p(\lbars{c_{max}} + \lbars{s_I})$, where $c_{max}$ is the longest command received from $E$. Both $n$ and $\lbars{c_{max}}$ are bounded by $t_E$ and $\lbars{s_I}$ is bounded by $t_T$ as $\exec(T, E)$ needs to construct the initial state. Overall
\begin{align*}
      \tilde t_E & \leq O(t_E) + n \cdot p(\lbars{c_{max}} + \lbars{s_I}) \\
            & \leq O(t_E) + t_E \cdot p(t_E + t_T).
\end{align*}

According to (\ref{simulation-t-efficiency}) in \cref{simulation} $\tilde t_T \leq p(t_T)$. We conclude that $\tilde t \leq t \cdot p(t) + p(t) + O(t)$.

\end{proof}

We shall be interested in executing a PTS in the context of another PTS. This is useful for modelling: we shall specify the threat model for a CVM program by embedding it as a subprocess within an IML process. This way we can formally define a setting with multiple threads and shared key creation and distribution without having to add process control primitives to CVM itself. An important property of embedding that we define is that it preserves the simulation relation. In order to define the embedding we start by adding holes to PTS:

\begin{definition}[PTS with a hole]\label{pts-with-hole}
      Given a polynomial $p$ we define a \emph{PTS with a hole identifiable in $p$-time} as a history-independent PTS with initial state $s_I$ that contains a special state $[]$ such that there are no transitions from $[]$ and such that there exists an algorithm that, given a process history $h$, runs in time $p(\lbars{h} + \lbars{s_I})$ and decides whether $h$ is a \emph{history of a hole}, that is, whether for all protocol states $\PS$ reachable from some environment $\eta$ and such that $h \in \dom(\PS)$ the process $\PS(h)$ is of the form $(\eta', [])$ with some environment $\eta'$.
\end{definition}

The definition intuitively states that the attacker must have an efficient means to decide whether a process is a hole given the observable history of the process. We can now proceed to defining the embedding:

\begin{definition}[Embedding of PTS]\label{embedding}
Given a PTS with a hole $T_E = (S_E, s_{IE}, \to_E)$ and a PTS $T = (S, s_I, \to)$ we define the \emph{embedding} $T_E[T]$ of $T$ within $T_E$ by
\[T_E[T] = (((S_E \setminus []) \times \{s_I\}) \cup S,\, (s_{IE}, s_I),\, \to_{E}' \cup \to),\]
where $\to_{E}'$ is obtained from $\to_E$ by replacing each occurrence of $s \in S_E \setminus []$ by $(s, s_I)$ and by replacing each occurrence of $\hole$ with $s_I$.
\end{definition}

\begin{theorem}[Simulation and embedding]\label{embedsim}
      For each two polynomials $p$ and $p'$ there exists a polynomial $p''$ such that if $T$ and $\tilde T$ are PTS with $T \lesssim_{p} \tilde T$ and $T_E$ is a PTS with a hole identifiable in $p'$-time then $T_E[T] \lesssim_{p''} T_E[\tilde T]$. 
\end{theorem}

\begin{proof}
We start by giving a definition of embedding for protocol states. Given a protocol state $\PS$ that contains holes with histories $h_1, \ldots, h_n$ and protocol states $\PS_1, \ldots, \PS_n$ we define the \emph{embedding} 
\begin{align*}
      &\PS[\PS_1, \ldots, \PS_n]  \\
      & \; = \PS_{-h_1, \ldots, h_n}\set{h_i h_j \mapsto \PS_i(h_j) | 1 \leq i \leq n, \, h_j \in \dom(\PS_i)}.
\end{align*}

Let $T_E = (S_E, s_{IE}, \to_E)$, $T = (S, s_I, \to)$, and $\tilde T = (\tilde S, \tilde s_I, \tilde \to)$ be defined as in the theorem. We show how to extend the relation $\lesssim$ on protocol states and the function $\tau$ given by the definition of simulation relation of $T$ and $\tilde T$ to a corresponding relation $\lesssim_E$ and a function $\tau_E$ for $T_E[T]$ and $T_E[\tilde T]$. For a protocol state $\PS$ over $T_E[T]$ and $\tilde \PS$ over $T_E[\tilde T]$ we set $\PS \lesssim_E \tilde \PS$ iff there exist protocol states $\PS_1, \ldots, \PS_n$ over $T$, $\tilde \PS_1, \ldots, \tilde \PS_n$ over $\tilde T$ and a protocol state $\PS_E$ over $T_E$ such that $\PS_i \lesssim \tilde \PS_i$ for all $i$ and 
\[\PS = \PS_E[\PS_1, \ldots, \PS_n] \; \text{and} \; \tilde \PS = \PS_E[\tilde \PS_1, \ldots, \tilde \PS_n]. \]

Let a command $c = (h, d)$ be given. We compute $\tau_E(c)$ as follows: first check whether $h$ contains a prefix $h_{E}$ such that $h_E$ is a history of a hole. If it doesn't, set $\tau_E(c) = c$, otherwise let $h'$ be a process history such that $h = h_{E} h'$ and let 
\begin{align*}
      & (\tilde h_1, d_1), \ldots, (\tilde h_m, d_m) = \tau((h', d)), \\
      & \tau_E(c) = (h_{E} \tilde h_1, d_1), \ldots, (h_{E} \tilde h_m, d_m).
\end{align*}
It is straightforward to check that (\ref{simulation-initial})--(\ref{simulation-step}) in \cref{simulation} are satisfied for $T_E[T]$ and $T_E[\tilde T]$ with $\tau_E$ and $\lesssim_E$. 

To prove (\ref{simulation-e-efficiency}) we need to bound the evaluation time of $\tau_E(c)$ for a command $c = (h, d)$ in terms of $\lbars{c}$ and $\lbars{s_{IE}'}$ where $s_{IE}' = (s_{IE}, s_I)$ is the initial state of $T_E[T]$. To evaluate $\tau_E(c)$ the following operations are performed:
\begin{itemize}
      \item 
            Run the hole-detection algorithm for each prefix of $h$. According to the assumption on $T_E$ this can be done in $\lbars{h} \cdot p'(\lbars{h} + \lbars{s_{IE}'})$ steps,
            
      \item
            if $h = h_E h'$, where $h_E$ is a history of a hole, evaluate $\tau(c')$ for $c' = (h', d)$. According to the assumption that $T \lesssim_{p} \tilde T$ this takes $p(\lbars{c'} + \lbars{s_I})$ steps.            
\end{itemize}
Given that $\lbars{h} \leq \lbars{c}$, $\lbars{c'} \leq \lbars{c}$, and $\lbars{s_I} \leq \lbars{s_{IE}'}$, the overall evaluation time of $\tau_E$ is bounded by
\[O(\lbars{c} \cdot p'(\lbars{c} + \lbars{s_{IE}'}) + p(\lbars{c} + \lbars{s_{IE}'})).\]

To prove (\ref{simulation-t-efficiency}) choose a valuation $\eta$ and assume that $T_E[T]$ reaches a state $\PS$ from $\eta$ in $t$ steps, $T_E[\tilde T]$ reaches a state $\tilde \PS$ from $\eta$ in $\tilde t$ steps, and $\PS \lesssim \tilde \PS$. By definition the states are of the form 
\[\PS = \PS_E[\PS_1, \ldots, \PS_n] \; \text{and} \; \tilde \PS = \PS_E[\tilde \PS_1, \ldots, \tilde \PS_n], \]
where $\PS_E$ is a state of $T_E$ and $\PS_i \lesssim \tilde \PS_i$ for all $i$. For each $i$ let $\eta_i$ be the environment of the $i$th hole in $\PS_E$. It is easy to see that $t = O(t_E + t_1 + \ldots + t_n)$, where $t_E$ is the time in which $T_E$ reaches $\PS_E$ from $\eta$ and $t_i$ is the time in which $T$ reaches $\PS_i$ from $\eta_i$. Similarly $\tilde t = O(t_E + \tilde t_1 + \ldots + \tilde t_n)$, where $\tilde t_i$ is the time in which $\tilde T$ reaches $\tilde \PS_i$ from $\eta_i$. From the assumption $T \lesssim_p \tilde T$ we know that $\tilde t_i \leq p(t_i)$ for each $i$. Assuming w.l.o.g. that $p$ is at least linear and monotonic, we conclude $\tilde t \leq p(t)$.
\end{proof}

The definition and the theorem can easily be extended to the setting with multiple holes $\hole_1, \ldots, \hole_n$. We shall write $T_E[T_1, \ldots, T_n]$ to denote the corresponding embedding. 

\section{Semantics of CVM \status{good, 27.03.2011}}\label{cvm-semantics}

\begin{figure*}[t]
\small
\begin{align*}
      &\frac{
            \forall v \in \var(P) \colon \range{\addr(v)}{N} \subseteq \addrspace 
            }{
            \eta,\, (\mcode{Init},\, \cvmP) \xrightarrow{\mathtt{ctr}\ \emptybs} \eta,\, (\bigcup_{v \in \var(P)} \range{\addr(v)}{N},\, \emptyset,\, [],\, \cvmP)
            } \tag{C-Init}\eqlabel{eq:c-init} \\
      \\
      &\frac{}{\eta,\, (\allocc,\, \memc,\, \stackc,\, \mcode{Const}\ b;\, \cvmP) \xrightarrow{\mathtt{ctr}\ \emptybs} \eta,\, (\allocc,\, \memc,\, b :: \stackc,\, \cvmP)} \tag{C-Const}\eqlabel{eq:c-const} \\
      \\
      &\frac{}{\eta,\, (\allocc,\, \memc,\, \stackc,\, \mcode{Ref}\ v;\, \cvmP) \xrightarrow{\mathtt{ctr}\ \emptybs} \eta,\, (\allocc,\, \memc,\, \bs(\addr(v)) :: \stackc,\, \cvmP)} \tag{C-Ref}\eqlabel{eq:c-ref} \\
      \\
      &\frac{p \in \BS \quad \lbars{p} = N \quad \range{\val(p)}{\val(l)} \subseteq \addrspace \setminus \allocc}{
                  \eta,\, (\allocc,\, \memc,\, l :: \stackc,\, \mcode{Malloc};\, \cvmP) \xrightarrow{\mathtt{ctr}\ p} \eta,\, (\allocc \cup \range{\val(p)}{\val(l)},\, \memc,\, p :: \stackc,\, \cvmP) 
            } &  
        \tag{C-Malloc}\eqlabel{eq:c-malloc} \\
      \\
      & \frac{
                  b, b_E \in \BS \quad \lbars{b} = \val(l) \leq \lbars{b_E} \quad \forall i \in \segment{\lbars{b}} \colon b[i] = \text{if $\val(p) + i \in \dom(\memc)$ then $\memc(\val(p) + i)$ else $b_E[i]$}
       }{
                  \eta,\, (\allocc,\, \memc,\, l :: p :: \stackc,\, \mcode{Load};\, \cvmP) \xrightarrow{\mathtt{ctr}\ b_E} \eta,\, (\allocc,\, \memc,\, b :: \stackc,\, \cvmP)
        }
        \tag{C-Load}\eqlabel{eq:c-load} \\
      \\
      & \frac{b \in BS \quad \lbars{b} = \val(l) < 2^N}{
                  \eta,\, (\allocc,\, \memc,\, l :: \stackc,\, \mcode{In}\ v\ src;\, \cvmP) \xrightarrow{src\ b} \eta\{v \mapsto b\},\, (\allocc,\, \memc,\, b :: \stackc,\, \cvmP)
            }
        \tag{C-In}\eqlabel{eq:c-in} \\
      \\
      & \frac{v \in \dom(\eta) \quad \lbars{\eta(v)} < 2^{N}}{
                  \eta,\, (\allocc,\, \memc,\, \stackc,\, \mcode{Env}\ v;\, \cvmP) \xrightarrow{\mathtt{ctr}\ \emptybs} \eta,\, (\allocc,\, \memc,\, \bs(\lbars{\eta(v)}) :: \eta(v) :: \stackc,\, \cvmP)
            }
        \tag{C-Env}\eqlabel{eq:c-env} \\
      \\
      & \frac{\arity(op) = n \quad b = \opfun_{op}(b_1, \ldots, b_n) \neq \bot \quad \lbars{b} < 2^{N}}{
                  \eta,\, (\allocc,\, \memc,\, b_1 :: \ldots :: b_n :: \stackc,\, \mcode{Apply}\ op;\, \cvmP) \xrightarrow{\mathtt{ctr}\ \emptybs} \eta,\, (\allocc,\, \memc,\, \bs(\lbars{b}) :: b :: \stackc,\, \cvmP)
            },
        \tag{C-Apply}\eqlabel{eq:c-apply} \\
      \\
      & \frac{}{
                  \eta,\, (\allocc,\, \memc,\, b :: \stackc,\, \mcode{Out}\ dest;\, \cvmP) \xrightarrow{dest\ b} \eta,\, (\allocc,\, \memc,\, \stackc,\, \cvmP)
            }
        \tag{C-Out}\eqlabel{eq:c-out} \\
      \\
      & \frac{b = i1}{
                  \eta,\, (\allocc,\, \memc,\, b :: \stackc,\, \mcode{Test};\, \cvmP) \xrightarrow{\mathtt{ctr}\ 1} \eta,\, (\allocc,\, \memc,\, \stackc,\, \cvmP)
            }
        \tag{C-Test}\eqlabel{eq:c-test} \\
      \\
      & \frac{\range{\val(p)}{\lbars{b}} \subseteq \allocc}{
                  \eta,\, (\allocc,\, \memc,\, p :: b :: \stackc,\, \mcode{Store};\, \cvmP) \xrightarrow{\mathtt{ctr}\ \emptybs} \eta,\, (\allocc,\, \memc\Set{\val(p) + i \mapsto b[i] | i \in \lbars{b}},\, \stackc,\, \cvmP)
            }
        \tag{C-Store}\eqlabel{eq:c-store} \\
\end{align*}
\caption{The concrete semantics of CVM.}
\label{fig:cvm-semantics}
\end{figure*}

This section presents the formal semantics of the CVM language, the syntax of which is introduced in \cref{fig:CVM-syntax}. In the following, let $N$, $\val$, and $\bs$ be chosen as in \cref{cvm}. In order to define the semantics, we associate to each CVM program the protocol transition system that is generated by it. Let a program $P \in \CVM$ be given. Let $\var(P)$ be the set of variables used in \cvm{Ref} instructions within $P$ and choose an allocation function $\addr \colon \var(P) \to \N$. We require that the allocated memory ranges do not overlap, that is 
\begin{align*}
      &\range{\addr(v)}{N} \cap \range{\addr(v')}{N} = \emptyset \; \text{for all $v \neq v'$}.
\end{align*}
We let $\sem{P}_C$ be the PTS with the initial state $(\mcode{Init}, P)$ and all other states of the form $(\allocc, \memc, \stackc, \cvmP)$, as described in \cref{cvm}. The transition rules of $\sem{P}_C$ are presented in \cref{fig:cvm-semantics}. The right hand side of each transition always contains a single process, so we omit the multiset bracket. 

The rule \eqref{eq:c-in} stores the input value in the environment in addition to placing it on the stack. This way the resulting PTS is history-independent---the state contains the information about all inputs so that there is only one trace leading to each state.

\section{Semantics of IML \status{good, 27.03.2011}}\label{iml-semantics}

\begin{figure}[t]
\small
\begin{align*}
      &\frac{}{(\eta,\, !\imlP) \xrightarrow{\mathtt{ctr}\ \emptybs} \{(\eta,\, \imlP),\, (\eta,\, !\imlP)\}} \tag{I-Repl}\eqlabel{eq:i-repl}\\
      \\
      &\frac{}{(\eta,\, \imlP \concat \imlQ) \xrightarrow{\mathtt{ctr}\ \emptybs} \{(\eta,\, \imlP),\, (\eta,\, \imlQ)\}} \tag{I-Par}\eqlabel{eq:i-par}\\
      \\
      &\frac{b \in \BS,\quad \lbars{b} = \val(\sem{e}_\eta)}{(\eta,\, (\nu x[e]); \imlP) \xrightarrow{\mathtt{rnd}\ b} \{(\eta\{x \mapsto b\},\, \imlP)\}} \tag{I-Nonce}\eqlabel{eq:i-nonce}\\
      \\
      &\frac{b \in BS}{(\eta,\,\miml{in}(x); \imlP) \xrightarrow{\mathtt{read}\ b} \{(\eta\{x \mapsto b\},\, \imlP)\}} \tag{I-In}\eqlabel{eq:i-in}\\
      \\
      &\frac{b = \sem{e}_\eta \neq \bot}{(\eta,\, \miml{out}(e); \imlP) \xrightarrow{\mathtt{write}\ b} \{(\eta,\, \imlP)\}} \tag{I-Out}\eqlabel{eq:i-out}\\
      \\
      &\frac{b = \sem{e}_\eta \neq \bot}{(\eta,\, \miml{event}(e); \imlP) \xrightarrow{\mathtt{event}\ b} \{(\eta,\, \imlP)\}} \tag{I-Event}\eqlabel{eq:i-event} \\
      \\
      &\frac{\sem{e}_\eta = i1}{(\eta,\, \miml{if\ $e$ then\ $\imlP$ else\ $\imlQ$}) \xrightarrow{\mathtt{ctr}\ 1} \{(\eta,\, \imlP)\}} \tag{I-Cond-True}\eqlabel{eq:i-cond-true} \\
      \\
      &\frac{\sem{e}_\eta = i0}{(\eta,\, \miml{if\ $e$ then\ $\imlP$ else\ $\imlQ$}) \xrightarrow{\mathtt{ctr}\ 0} \{(\eta,\, \imlQ)\}} \tag{I-Cond-False}\eqlabel{eq:i-cond-false}\\
      \\
      &\frac{b = \sem{e}_\eta \neq \bot}{(\eta,\, \miml{let\ $x = e$ in\ $\imlP$ else\ $\imlQ$}) \xrightarrow{\mathtt{ctr}\ 1} \{(\eta\{x \mapsto b\},\, \imlP)\}} \tag{I-Let-True}\eqlabel{eq:i-let-true} \\
      &\frac{\sem{e}_\eta = \bot}{(\eta,\, \miml{let\ $x = e$ in\ $\imlP$ else\ $\imlQ$}) \xrightarrow{\mathtt{ctr}\ 0} \{(\eta,\, \imlQ)\}} \tag{I-Let-False}\eqlabel{eq:i-let-false}
\end{align*}
\caption{The semantics of IML.}
\label{fig:iml-semantics}
\end{figure}

Just as for CVM, the semantics of IML is given as a protocol transition system. We choose the functions $\bs$ and $\val$ as in \cref{cvm} and let the function $\sem{\cdot}_\eta$ be defined as in \cref{iml}. For an IML process $P$ we let $\sem{P}_I$ be the PTS with IML processes as states, with starting state $P$ and transitions described in \cref{fig:iml-semantics}.

The rules \eqref{eq:i-repl} and \eqref{eq:i-par} are standard replication and parallel composition rules from the pi calculus. The rule \eqref{eq:i-nonce} is interesting in that it restricts the generated nonce to be of a given length. The input rule \eqref{eq:i-in} does not place such a restriction and allows the input to be of any length (this is more permissive than the CVM input rule). The rules \eqref{eq:i-out} and \eqref{eq:i-event} generate an output and an event transition respectively. The conditional rules \eqref{eq:i-cond-true} and \eqref{eq:i-cond-false} have different control labels, so that the attacker can distinguish the branch that has been taken by the process. Unlike CVM there is no explicit rule for reading environment variables, because IML operates on the environment $\eta$ directly.

Consider IML enriched with an additional syntactic form $\hole_i$ (a hole) with $i \in \N$ and without any reductions. For an IML process $P$ with holes the semantics $\sem{P}_I$ is a PTS with holes (\cref{pts-with-hole}). The history of a process uniquely determines its state, for instance, given the process $P = $ \iml{!(if $e$ then $\hole$ else $0$)} and history $h = (\mcode{ctr\ $\emptybs$})\,1\,(\mcode{ctr\ $1$})\,1$, it is easy to see that $h$ is a history of a hole in $\sem{P}_I$. Here it is important that the true and the false branches in \cref{fig:iml-semantics} have different control labels. In general, whether $h$ is a history of a hole in $\sem{P}_I$, is computable in time linear in $\lbars{P} + \lbars{h}$. Just like CVM IML is history-independent because it records all the inputs in the environment. Thus the following holds:

\begin{lemma}[IML with holes]\label{iml-holes}
      For an IML process $P$ with holes the semantics $\sem{P}_I$ is a PTS with holes identifiable in $p$-time for some fixed linear polynomial $p$. 
\end{lemma}

The semantics of mixed IML and CVM processes is defined by using a PTS embedding as follows:
\begin{definition}[Mixed semantics]\label{mixed-semantics}
For a process $P_E \in \IML$ with $n$ holes and processes $P_1, \ldots, P_n \in \CVM$ let
\[\sem{P_E[P_1, \ldots, P_n]}_{CI} = \sem{P_E}_I[\sem{P_1}_C, \ldots, \sem{P_n}_C].\]
\end{definition}

\newpage
\section{Simplifications \status{good, 24.02.2011}}\label{simplify}

\newcommand{\cutL}{\opname{cutL}}
\newcommand{\cutR}{\opname{cutR}}

\begin{figure}
\small
\begin{align*}
      & \cutL_{\facts}(l,\, e_1 \concat \ldots \concat e_n) \\
      & \quad = 
            \begin{cases}
                  e_1 \concat \ldots \concat e_{i - 1} \concat \simplify_\facts(\cutL_\facts(l -\nop l', e_i)) \\
                  \quad \text{if $\facts \entails (l \geq l') \wedge (l \leq l' +\nop \getLen(e_i))$,} \\
                  \quad \text{where $l' = \Sigma_{j = 1}^{i - 1} \getLen(e_j)$,}  \\
                  (e_1 \concat \ldots \concat e_n)\{i0, l\} \;\;\text{otherwise,}
            \end{cases}
      \\
      & \cutR_{\facts}(l,\, e_1 \concat \ldots \concat e_n) = \\
      & \quad = 
            \begin{cases}
                  \simplify_\facts(\cutR_\facts(l -\nop l', e_i)) \concat e_{i + 1} \concat \ldots \concat e_{n} \\
                  \quad \text{if $\facts \entails (l \geq l') \wedge (l \leq l' +\nop \getLen(e_i))$,} \\
                  \quad \text{where $l' = \Sigma_{j = 1}^{i - 1} \getLen(e_j)$,}  \\
                  (e_1 \concat \ldots \concat e_n)\{l, \getLen(e_1 \concat \ldots \concat e_n) -\nop l\} \;\;\text{otherwise,}
            \end{cases}  
      \\
      & \simplify_\facts(e\{e_o, e_l\}) \\
      & \quad = 
            \begin{cases}
                  e & \begin{aligned}
                                        \text{if}\; \facts \entails & (e_o = i0) \\ 
                                        & \wedge (e_l = \getLen(e))
                                  \end{aligned}
                  \\
                  \emptybs & \text{if $\facts \entails (e_l = i0)$} \\
                  e'\{e_o +\nop e_o', e_l\} & \text{if $e = e'\{e_o', e_l'\}$} \\
                  \cutL_\facts(e_l, \cutR_\facts(e_o, e)) & \text{if $e$ is a concatenation} \\
                  e\{e_o, e_l\} & \text{otherwise.}
            \end{cases}
\end{align*}
\caption{Simplification rules.}
\label{fig:simplify}
\end{figure}

\Cref{fig:simplify} presents the simplification rules used in our symbolic execution algorithm. The simplification function is concerned with simplifying range expressions when possible, for instance, an expression of the form $(a \concat b)\{x, y\}$, where $\facts \entails (x = \getLen(a))$ and $\facts \entails (y = \getLen(b))$ will simplify to $b$. The main work is done by two recursive functions $\cutL,\, \cutR \colon \SExp \times \SExp \to \SExp$ that given a length expression $l$ and a concatenation expression $e$ attempt to split $e$ at the position given by $l$. If this succeeds, $\cutL$ returns the part of $e$ to the left of the split position and $\cutR$ returns the part to the right.

In order to simplify an expression of the form $e\{e_o, e_l\}$ the function $\simplify$ first checks two special cases: if $e_o$ is equal to zero and $e_l$ is equal to the length of $e$ then the range can be removed and the expression can be simplified to just $e$. On the other hand if $e_l$ is equal to zero then the range expression can be simplified to $\emptybs$. If $e$ is itself a range expression of the form $e'\{e_o', e_l'\}$ then the two ranges are merged giving the result $e'\{e_o +\nop e_o', e_l\}$. If $e$ is a concatenation then the functions $\cutR$ and $\cutL$ are applied. Finally, if all of the above fails, the original expression is returned without simplification.

We omit the soundness proof for our simplification function.

\section{Symbolic Execution Soundness \status{good, 27.03.2011}}\label{symex-proof}

We prove our main result (\cref{symex}). We shall do so by showing that the PTS $\sem{\sem{P}_S}_I$ resulting from the symbolic execution of a program $P \in \CVM$ simulates (in the sense of \cref{simulation}) the PTS $\sem{P}_C$ resulting from running $P$ directly. This result is captured by \cref{cvm-iml-simulation}. \Cref{symex} then follows by combined application of \cref{simsec,embedsim,iml-holes} together with \cref{mixed-semantics}.

For compactness we shall write $b^\N$ instead of $\val(b)$ for $b \in BS$. When referring to valuations we shall mean extended valuations of the form $\eta\colon \Var \cup \PBase \to \BS_\bot$. For an extended valuation $\eta$ let $\var(\eta)$ be the restriction of $\eta$ to $\Var$. 

We shall make use of the soundness of the function $\getLen$ introduced in \cref{cvm-to-iml} that we state here without proof: for any $e \in \SExp$ and valuation $\eta$
\[(\sem{e}_\eta \neq \bot) \wedge (\lbars{\sem{e}_\eta} < 2^N) \Rightarrow \sem{\getLen(e)}_\eta = \bs(\lbars{\sem{e}_\eta}).\]

The main tool in the proof of \cref{cvm-iml-simulation} is a concretisation function that, given a valuation, maps symbolic execution states to concrete execution states. Given a symbolic state $s = (\facts, \allocs, \mems, \stacks, \cvmP)$ and a valuation $\eta$ we say that $s$ is \emph{$\eta$-consistent} when all expressions in $s$ are well-defined with respect to $\eta$, when $\eta$ maps all symbolic memory locations to disjoint ranges that are within allocated memory bounds, all conditions in $\facts$ hold with respect to $\eta$, and $\eta$ agrees with the $\addr$ function for stack variables. Formally, we say that $s$ is $\eta$-consistent, iff
\begin{enumerate}
      \item
            for all $pb \in \dom(\mems) \colon$
            \begin{gather*}
                  \sem{\mems(pb)}_\eta \neq \bot, \quad \sem{\allocs(pb)}_\eta \neq \bot, \quad \eta(pb) \neq \bot, \\
                  \lbars{\sem{\mems(pb)}_\eta} \leq \sem{\allocs(pb)}_\eta^\N,
            \end{gather*}
             
      \item 
            for all $pb,\, pb' \in \dom(\allocs)$ with $pb \neq pb'$:
            \[\range{\eta(pb)^\N}{\sem{\allocs(pb)}_\eta^\N} \cap \range{\eta(pb')^\N}{\sem{\allocs(pb')}_\eta^\N} = \emptyset,\]

      \item 
            for all $\psi \in \facts \colon \sem{\psi}_\eta = i1$,
            
      \item
            for all $e \in \stacks \colon \sem{e}_\eta \neq \bot$,
            
      \item
            for all $v \in \var(P) \colon \eta(\stack\, v) = \bs(\addr(v))$.
\end{enumerate}
For an $\eta$-consistent state $s$ let $\conc_\eta(s) = (\allocsc, \memsc, \stacksc, \cvmP)$ be the concrete state where $\stacksc$ is obtained from $\stacks$ by applying $\sem{\cdot}_\eta$ to each element and
\begin{align*}
      \allocsc &= \bigcup \Set{ \range{\eta(pb)^\N}{\sem{\allocs(pb')}_\eta^\N} | pb \in \dom(\allocs) },\\
      \memsc &= \Set{\eta(pb)^\N + i \mapsto \sem{\mems(pb)}_\eta[i] | 
                  \begin{aligned}
                              & pb \in \dom(\mems),\, \\
                              & i < \lbars{\sem{\mems(pb)}_\eta}
                  \end{aligned}
                  }.
\end{align*}
The conditions of $\eta$-consistency guarantee that $\memsc$ is well-defined: symbolic expressions will map onto concrete memory without overlapping, that is, for each $p \in \N$ there is only one pair $pb, i$ such that $\eta(pb)^\N + i = p$. 

The special state $(\mcode{Init}, P)$ is defined to be $\eta$-consistent for any $\eta$ with $\dom(\eta) \subseteq \Var$ and we let $\conc_\eta(\mcode{Init}, P) = (\mcode{Init}, P)$.

We start by proving two lemmas relating the symbolic and the concrete execution of a program. \Cref{cvm-to-iml-soundness-1} shows that if a symbolic state $s$ maps to a concrete state $c$ then the state following $s$ in the symbolic execution can be mapped to the state following $c$ in the concrete execution. \Cref{cvm-to-iml-soundness-2} shows that if in a symbolic and a concrete execution the states can be mapped to each other then the IML program generated by the symbolic execution performs the same actions as the concrete execution.

\begin{lemma}\label{cvm-to-iml-soundness-1}
      Let $(\eta_c, c) \xrightarrow{l} (\eta_c', c')$ be a concrete transition (\cref{fig:cvm-semantics}), $\smash{s \xrightarrow{\lambda} s'}$ a symbolic transition (\cref{fig:cvm-symex}), and $\eta$ an extension of $\eta_c$ such that $s$ is $\eta$-consistent and $\conc_{\eta}(s) = c$. Then there exists an extension $\eta'$ of both $\eta$ and $\eta_c'$ such that $s'$ is $\eta'$-consistent and $\conc_{\eta'}(s') = c'$. 
\end{lemma}

\begin{SHORT}
The proof is by case distinction over the value of the executed instruction and is presented in the full version of the paper.
\end{SHORT}

\begin{proof}
      By definition of the concretisation function both the concrete and the symbolic step are executed with the same instruction or both perform the initialisation. We prove the lemma by enumerating the pairs of rules that generate the transitions. For the purpose of this proof we are not interested in the values of transition labels $l$ and $\lambda$.
      
      In the following $\allocc, \ldots$ and $\allocc', \ldots$ refer to components of $c$ and $c'$ respectively, $\allocs, \ldots$ and $\allocs', \ldots$ refer to components of $s$ and $s'$, and $\allocsc, \ldots$ and $\allocsc', \ldots$ refer to components of $\conc_{\eta}(s)$ and $\conc_{\eta'}(s')$.
      
      \begin{enumerate}
            \item
                  \eqref{eq:c-init,eq:s-init}
                  
                  By definition of $\eta$-consistency for the initial state we know that $\stack\, v \notin \dom(\eta)$ for all $v \in \var(P)$. We show that the lemma holds with
                  \[\eta' = \eta\Set{ \stack\, v \mapsto \bs(\addr(v)) | v \in \var(P) }.\]
                  The second condition of $\eta'$-consistency of $s'$ follows by the choice of $\addr$ function (\cref{cvm-semantics}), the other conditions are straightforward to check. In $s'$ each location in the symbolic memory is initialised to $\emptybs$, so applying the definition of $\conc_\eta$ we see that $\memsc' = \emptyset = \memc'$. Finally
                  \begin{align*}
                        \allocsc' & = \bigcup_{v \in \var(P)} \Set{\range{\eta'(\stack\, v)^\N}{\sem{\allocs'(\stack\, v)}_{\eta'}^\N} } \\
                              & = \bigcup_{v \in \var(P)} \Set{ \range{\bs(\addr(v))^\N}{\bs(N)^\N} } \\
                              & = \bigcup_{v \in \var(P)} \{ \range{\addr(v)}{N} \} = \allocc'.
                  \end{align*}
      
            \item
                  \eqref{eq:c-const,eq:s-const} with \cvm{Const $b$}
            
                  Both the concrete and the symbolic transition have the effect of putting the same bitstring $b$ onto the stack. Thus both the $\eta$-consistency and the state correspondence are preserved and the lemma holds with $\eta' = \eta$.
                  
            \item
                  \eqref{eq:c-ref,eq:s-ref} with \cvm{Ref $v$}

                  The concrete transition puts $\bs(\addr(v))$ on the stack and the symbolic transition puts $\ptr(\stack\, v,\, i0)$ on the stack. By $\eta$-consistency $\eta(\stack, v) = \bs(\addr(v))$, thus
                  \[\sem{\ptr(\stack\, v,\, i0)}_\eta = \eta(\stack, v) +\bop i0 = \bs(\addr(v))\]
                  and the lemma holds with $\eta' = \eta$.
                  
            \item
                  \eqref{eq:c-malloc,eq:s-malloc}
                  
                  Let $p$ and $l$ be defined as in rule \eqref{eq:c-malloc} and $pb$ and $e_l$ be defined as in rule \eqref{eq:s-malloc}. We show that the lemma holds with $\eta' = \eta\{pb \mapsto p\}$. It is straightforward to check that the first condition of $\eta'$-consistency of $s'$ holds, taking into consideration that $\sem{\mems'(pb)}_\eta = \emptybs$. To prove the second condition, let $pb' \in \dom(\allocs')$ such that $pb' \neq pb$. In that case $pb' \in \dom(\allocs)$ and by definition of $\conc_\eta$ and the state correspondence of $c$ and $s$
                  \begin{multline*}
                        \range{\eta'(pb')^\N}{\sem{\allocs'(pb')}_{\eta'}^\N} = \range{\eta(pb')^\N}{\sem{\allocs(pb')}_\eta^\N} \\
                        \subseteq \allocsc = \allocc.
                  \end{multline*}
                  By initial state correspondence and the definition of $\eta'$
                  \begin{multline*}
                        \range{\eta'(pb)^\N}{\sem{\allocs'(pb)}_{\eta'}^\N} = \range{\eta'(pb)^\N}{\sem{e_l}_{\eta'}^\N} \\
                        = \range{p^\N}{\sem{e_l}_{\eta}^\N} = \range{p^\N}{l^\N} \subseteq \addrspace \setminus \allocc.
                  \end{multline*}
                  Thus the allocation ranges of $pb$ and $pb'$ are disjoint and the condition (2) holds. Conditions (3) to (5) are straightforward to check. To prove that $\conc_{\eta'}(s') = c'$ observe that 
                  \begin{gather*}
                        \allocsc' = \allocsc \cup \range{\eta'(pb)^\N}{\sem{\allocs'(pb)}_{\eta'}^\N} = \allocc \cup \range{p^\N}{l^\N} = \allocc', \\
                        \memsc' = \memsc = \memc = \memc', \\
                        \sem{\ptr(pb,\, i0)}_{\eta'} = \eta'(pb) +\bop i0 = p.
                  \end{gather*}
      
            \item 
                  \eqref{eq:c-load,eq:s-load}
                  
                  Both the concrete and the symbolic rule have the effect of replacing two values on the stack with a new value. In the concrete transition the new value is $b \in BS$ such that $b[i] = \memc(p + i)$ whenever $\memc(p + i)$ is initialised and $p$ is defined as in rule \eqref{eq:c-load}. In the symbolic transition the new value is $e = \simplify_\Sigma(\mems(pb)\{e_o, e_l\})$, where $pb$, $e_o$, and $e_l$ are defined as in rule \eqref{eq:s-load}. We shall prove that $\sem{e}_\eta = b$ so that the lemma holds with $\eta' = \eta$.
                  
                  Let $b_h = \sem{\mems(pb)}_\eta$. By definition of $\conc_\eta$ and initial state correspondence
                  \begin{align*}
                        b_h & = \memsc\left(\range{b_{pb}^\N}{\lbars{b_h}}\right) = \memc\left(\range{b_{pb}^\N}{\lbars{b_h}}\right),
                  \end{align*}
                  where we use the notation $\memc(I)$ for $I \subseteq \addrspace$ to denote the sequence of bits of $\memc$ with addresses in $I$. Thus $\memc$ is defined in the range $\rangesmash{b_{pb}^\N}{\lbars{b_h}}$, in particular 
                  \begin{equation*}
                        b_{pb}^\N + \lbars{b_h} < 2^N.\tag{*}    
                  \end{equation*}
                  Let $b_o = \sem{e_o}_\eta$ and $b_l = \sem{e_l}_\eta$. Evaluating the conditions of the rule \eqref{eq:s-load} and using the assumption of $\eta$-consistency we obtain $b_o +\nop b_l \leq \sem{\getLen(e_h)}_\eta$. Because $\lbars{b_h} < 2^N$ we can apply soundness of $\getLen$ which together with the definitions of bitstring operations $+\nop$ and $\leq$ gives 
                  \begin{equation*}
                        b_o^\N + b_l^\N \leq \lbars{b_h}.\tag{**}
                  \end{equation*}
                  Using the definition of the function $\sub$
                  \begin{align*}
                        \sem{\mems(pb)\{e_o, e_l\}}_\eta &= \sub(b_h, b_o^\N, b_l^\N) \\
                              &= \sub\left(\memc\left(\range{b_{pb}^\N}{\lbars{b_h}}\right), b_o^\N, b_l^\N\right)  \\
                              &= \memc\left(\range{b_{pb}^\N + b_o^\N}{b_l^\N}\right).
                  \end{align*}
                  This allows us to apply soundness of $\simplify$:
                  \begin{equation*}
                        \begin{aligned}
                              \sem{e}_\eta &= \sem{\simplify_\Sigma(\mems(pb)\{e_o, e_l\})}_\eta \\
                              & = \sem{\mems(pb)\{e_o, e_l\}}_\eta = \memc\left(\range{b_{pb}^\N + b_o^\N}{b_l^\N}\right).
                        \end{aligned}
                  \end{equation*}
                  By the state correspondence of $c$ and $s$ we obtain
                  \begin{align*}
                        p &= \sem{\ptr(pb, e_o)}_\eta = \eta(pb) +\bop \sem{e_o}_\eta = b_{pb} +\bop b_o, \\
                        \lbars{b} &= \sem{e_l}^\N = b_l^\N.
                  \end{align*}
                  By (*) and (**) $b_{pb}^\N + b_o^\N < 2^N$, thus 
                  \[b_{pb}^\N + b_o^\N = (b_{pb} +\bop b_o)^\N = p^\N.\]
                  Substituting this into the above we get
                  \begin{align*}
                         \sem{e}_\eta &= \memc\left(\range{b_{pb}^\N + b_o^\N}{b_l^\N}\right) = \memc\left(\range{p^\N}{\lbars{b}}\right) = b
                  \end{align*}
                  The final equality holds as the referenced memory cells lie within the initialised range $\rangesmash{b_{pb}^\N}{\lbars{b_h}}$.
                  
            \item
                  \eqref{eq:c-in,eq:s-in} with \cvm{In $v$ $src$}
                  
                  The rule \eqref{eq:c-in} takes a value $l$ from the stack and places a value $b$ of length $l^\N$ on the stack. Additionally it updates $\eta_c' = \eta_c\{v \mapsto b\}$. The rule \eqref{eq:s-in} takes an expression $e_l$ from the stack, places $v$ on the stack, and adds the fact $\len(v) = e_l$ to $\facts$. We show that the lemma holds with $\eta' = \eta\{v \mapsto b\}$. Due to initial state correspondence $\sem{e_l}_\eta = l$ and due to the condition of the rule \eqref{eq:c-in} $\lbars{b} < 2^N$, thus
                  \[\sem{\len(v)}_{\eta'}^\N = \bs(\lbars{b})^\N = \lbars{b} = l^\N = \sem{e_l}_{\eta'}^\N,\]
                  so that the new fact is indeed valid.
                  
            \item
                  \eqref{eq:c-env,eq:s-env} with \cvm{Env $v$}
                  
                  The rule \eqref{eq:c-env} places $\eta_e(v)$ together with $\bs(\lbars{\eta_e(v)})$ on the stack (the valuation $\eta$ in \cref{fig:cvm-semantics} corresponds to $\eta_e$ in the lemma). The rule \eqref{eq:s-env} places $v$ and $\len(v)$ on the stack. By assumption of the lemma $\eta(v) = \eta_e(v)$, so it is straightforward to check that the lemma holds with $\eta' = \eta$.
                  
            \item
                  \eqref{eq:c-apply,eq:s-apply} with \cvm{Apply} $op$
                  
                  The rule \eqref{eq:c-apply} places on the stack the bitstring $b = \opfun_{op}(b_1, \ldots, b_n)$ together with its length, whereby $b_1, \ldots, b_n$ are taken from the stack. The rule \eqref{eq:s-apply} places on the stack the value $e = \apply(op, e_1, \ldots, e_n)$ together with $\len(e)$, whereby $e_1, \ldots, e_n$ are taken from the stack. We show that $\sem{e}_\eta = b$ so that the lemma holds with $\eta' = \eta$. By initial state correspondence we have $\sem{e_i}_\eta = b_i$ for all $i$. We enumerate the cases arising from the definition of $\apply$ given $b \neq \bot$:
                  
                  \begin{enumerate}
                        \item
                              $n = 2$, $e_1 = \ptr(pb, e_o)$, $e_2 \in \IExp$, and $op = +\bop$. In this case
                              \begin{align*}
                                    b &= \sem{\ptr(pb, e_o)}_\eta +\bop \sem{e_2}_\eta \\
                                          & = \eta(pb) +\bop \sem{e_o}_\eta +\bop \sem{e_2}_\eta \\
                                          & = \sem{\ptr(pb, e_o +\bop e_2)}_\eta \\
                                          & = \sem{\apply(+\bop, \ptr(pb, e_o), e_2)}_\eta  
                              \end{align*}

                        \item
                              $n = 2$, $e_1 = \ptr(pb, e_o)$, $e_2 = \ptr(pb, e_o')$, $op = -\bop$. In this case
                              \begin{align*}
                                    b &= \sem{\ptr(pb, e_o)}_\eta -\bop \sem{\ptr(pb, e_o')}_\eta \\
                                          & = \eta(pb) +\bop \sem{e_o}_\eta -\bop (\eta(pb) +\bop \sem{e_o'}_\eta) \\
                                          & = \sem{e_o}_\eta -\bop \sem{e_o'}_\eta\\
                                          & = \sem{\apply(-\bop, \ptr(pb, e_o), \ptr(pb, e_o'))}_\eta  
                              \end{align*}

                        \item
                              $e_1, \ldots, e_n \in \IExp$. In this case
                              \begin{align*}
                                    b &= \opfun_{op}(\sem{e_1}_\eta, \ldots, \sem{e_n}_\eta) = \sem{op(e_1, \ldots, e_n)}_\eta \\
                                          & = \sem{\apply(op, e_1, \ldots, e_n)}_\eta  
                              \end{align*}

                  \end{enumerate}
                  
            \item
                  \eqref{eq:c-out,eq:s-out}
                  
                  The lemma holds trivially with $\eta' = \eta$.
                  
            \item
                  \eqref{eq:c-test,eq:s-test}
                  
                  The rule \eqref{eq:c-test} removes a value $b$ from the stack. The rule \eqref{eq:s-test} removes an expression $e$ from the stack and adds $e$ to the set of facts. We show that the lemma holds with $\eta' = \eta$. We only need to prove that $\sem{e}_\eta = i1$, but this follows from the assumption of the lemma that $\sem{e}_\eta = b$ and the condition $b = i1$ of the rule \eqref{eq:c-test}.
                  
            \item
                  \eqref{eq:c-store,eq:s-store}
                  
                  Both the concrete and the symbolic transition perform a memory update. These updates are
                  \begin{align*}
                        \memc' &= \memc\Set{ p^\N + i \mapsto b[i] | i < \lbars{b} }, \\
                        \mems' &= \mems\{ pb \mapsto e_h' \}, \tag{1}\eqlabel{eq:store-1}
                  \end{align*}
                  where $p$ and $b$ are defined as in rule \eqref{eq:c-store}, and $pb$ and $e_h'$ are defined as in rule \eqref{eq:s-store}. 
                  
                  We shall prove that the lemma holds with $\eta' = \eta$. We start by showing that $s'$ is $\eta$-consistent. As the transition only updates the memory, we only need to check that $\sem{e_h'}_\eta \neq \bot$ and $\lbars{\sem{e_h'}_\eta} \leq \sem{\allocs(pb)}_\eta^\N$. Let $e_h$, $e_s$, $e_{lh}$, $e_l$, $e_o$, $e$, and $pb$ be defined as in \eqref{eq:s-store}. For $e_x \in \{e_h, e_s, e_{lh}, e_l, e_o\}$ let $b_x = \sem{e_x}_\eta$ and let $b_{pb} = \sem{pb}_\eta$. 
                  
                  By initial state correspondence $\sem{e}_\eta = b$. The rule \eqref{eq:c-store} assumes $\range{p^\N}{\lbars{b}} \subseteq \allocc$, which implies $\lbars{b} < 2^N$. Using the soundness of $\getLen$ and the definition of $b_l$ we obtain 
                  \[b_l^\N = \sem{\getLen(e)}_\eta^\N = \lbars{\sem{e}_\eta} = \lbars{b}.\] 
                  
                  By initial state correspondence 
                  \begin{align*}
                        b_h & = \memc\left(\range{b_{pb}^\N}{\lbars{b_h}}\right) = \memc\left(\range{b_{pb}^\N}{b_{lh}^\N}\right),
                  \end{align*}
                  where the second equality follows by soundness of $\getLen$ and the fact 
                  \begin{equation*}
                        b_{pb}^\N + \lbars{b_h} < 2^N\tag{2}\eqlabel{eq:store-2}    
                  \end{equation*}
                  established by the first equality. 
                  
                  We shall distinguish between two cases in the premise of the rule \eqref{eq:s-store}. The first case is
                  \begin{align*}
                        & \facts \entails (e_o +\nop e_l < e_{lh}), \quad e_h' = \simplify_\Sigma(e_h''), \, \text{where} \\
                        & e_h''= e_h\{0, e_o\} \concat e \concat e_h\{e_o +\nop e_l, e_{lh} -\nop (e_o +\nop e_l)\}.
                  \end{align*}
                  In this case the same argument as for the rule \eqref{eq:s-load} yields
                  \begin{equation*}
                        b_o^\N + b_l^\N < \lbars{b_h} = b_{lh}^\N.\tag{A3}\eqlabel{eq:store-A3}
                  \end{equation*}                  
                  
                  Substituting the value of $b_h$ and expanding the definition of the function $\sub$ under consideration of \eqref{eq:store-A3} we obtain 
                  \begin{equation*}
                  \begin{aligned}
                        \sem{e_h''}_\eta & = \memc\left.\left(\range{b_{pb}^\N}{b_{o}^\N}\right) \right\concat b \\
                               & \quad \left\concat \memc\left(\range{b_{pb}^\N + b_o^\N + b_l^\N}{b_{lh}^\N - b_o^\N - b_l^\N}\right)\right. .
                  \end{aligned}
                  \tag{A4}\eqlabel{eq:store-A4}
                  \end{equation*}
                  Applying the soundness of the function $\simplify$ we get $\sem{e_h'}_\eta = \sem{e_h''}_\eta \neq \bot$. From $b_l^\N = \lbars{b}$ follows $\lbars{\sem{e_h'}_\eta} = \lbars{\sem{e_h}_\eta}$. By initial state correspondence $\lbars{\sem{e_h}_\eta} \leq \sem{\allocs(pb)}_\eta^\N$. This proves $\eta$-consistency of $s'$ in the first case.
                  
                  The second case in the premise of the rule \eqref{eq:s-store} is 
                  \begin{align*}
                        & \facts \entails (e_o +\nop e_l \geq e_{lh}) \wedge (e_o \leq e_{lh}) \wedge (e_o +\nop e_l \leq e_s), \\
                        & e_h' = \simplify_\Sigma(e_h\{0, e_o\} \concat e).
                  \end{align*}
                  Together with $\eta$-consistency of $s$ this implies the following condition on bitstrings:
                  \begin{equation*}
                        (b_o^\N + b_l^\N \geq b_{lh}^\N) \wedge (b_o^\N \leq b_{lh}^\N) \wedge (b_o^\N + b_l^\N \leq b_s^\N).\tag{B3}\eqlabel{eq:store-B3}
                  \end{equation*}                  
                  This allows us to expand the definition of $\sub$ and apply soundness of $\simplify$ to obtain
                  \begin{equation*}
                        \sem{e_h'}_\eta = \memc\left.\left(\range{b_{pb}^\N}{b_{o}^\N}\right) \right\concat b \neq \bot, \tag{B4}\eqlabel{eq:store-B4}
                  \end{equation*}
                  Using \eqref{eq:store-B3}
                  \[\lbars{\sem{e_h'}_\eta } = b_{o}^\N + \lbars{b} = b_{o}^\N + b_l^\N \leq b_s^\N = \sem{\allocs(pb)}_\eta^\N,\]
                  which proves $\eta$-consistency of $s'$ in the second case.
                  
                  The next step is to show that $\conc_\eta(s') = c'$. Both in the first and in the second case above $\lbars{\sem{e_h'}_\eta} \geq \lbars{\sem{e_h}_\eta}$ (in the first case they are equal, in the second case it follows from \eqref{eq:store-B3}). Comparing the definition of $\memsc'$ and $\memsc$ and using the relation \eqref{eq:store-1} between $\mems'$ and $\mems$
                  \[\memsc' = \memsc\Set{b_{pb}^\N + i \mapsto (\sem{e_h'}_\eta)[i] | i < \lbars{\sem{e_h'}_\eta}}. \]
                  Substituting the value of $\sem{e_h'}_\eta$ from either \eqref{eq:store-A4} or \eqref{eq:store-B4} and using the assumption $\memsc = \memc$ from the initial state correspondence we can simplify this to
                  \[\memsc' = \memc\Set{b_{pb}^\N + b_o^\N + i \mapsto b[i] | i < \lbars{b} }. \]
                  By initial state correspondence
                  \begin{align*}
                        p &= \sem{\ptr(pb, e_o)}_\eta = \eta(pb) +\bop \sem{e_o}_\eta = b_{pb} +\bop b_o, 
                  \end{align*}
                  It is $b_{pb}^\N + b_o^\N < 2^N$ both in the first and in the second case above: in the first case it follows from \eqref{eq:store-2,eq:store-A3}, in the second case it follows from \eqref{eq:store-2,eq:store-B3}. This implies $p^\N = b_{pb}^\N + b_o^\N$. Thus
                  \[\memsc' = \memc\Set{p^\N + i \mapsto b[i] | i < \lbars{b} } = \memc'.\]
      \end{enumerate}
\end{proof}

We call a valuation $\eta'$ \emph{minimal with a property $\phi$} iff $\eta'$ satisfies $\phi$ and $\eta'_{-x}$ does not satisfy $\phi$ for all $x \in \dom(\eta')$.

\begin{lemma}\label{cvm-to-iml-soundness-2}
      Let $\eta_c$, $\eta_c'$, $\eta$, and $\eta'$ be valuations and $s$ and $s'$ be symbolic states such that $s$ is $\eta$-consistent, $s'$ is $\eta'$-consistent, and there are transitions $(\eta_c,\, \conc_\eta(s)) \smash{\xrightarrow{l}} (\eta_c',\, \conc_{\eta'}(s'))$ and $\smash{s \xrightarrow{\lambda} s'}$ with $\lambda \neq \emptybs$. Assume additionally that $\eta'$ is a minimal extension of $\eta$ with the property above. Then for all $P \in \IML$ the following is a valid IML transition (\cref{fig:iml-semantics}):
      \[\{(\var(\eta), lP)\} \xrightarrow{\smash l} \{(\var(\eta'), P)\}.\]
\end{lemma}
 
\begin{proof}
      We prove the lemma by case distinction over all pairs of $l$ and $a$ that can occur. 
      
      \begin{enumerate}
            \item 
                  $l = \mcode{read}\ b$ and $\lambda = \miml{in}(x);$ by rules \eqref{eq:c-in} and \eqref{eq:s-in}.
                  
                  From the correspondence between the symbolic and the concrete transition we obtain $\eta'(x) = b$. Because $\eta'$ was chosen to be minimal $\eta' = \eta\{x \mapsto b\}$, which also implies $\var(\eta') = \var(\eta)\{x \mapsto b\}$. The lemma follows by rule \eqref{eq:i-in}.
                  
            \item
                  $l = \mcode{rnd}\ b$ and $\lambda = (\nu x[e_l]);$ by rules \eqref{eq:c-in} and \eqref{eq:s-in}.
                  
                  The correspondence between the symbolic and the concrete transition implies $\var(\eta') = \var(\eta)\{x \mapsto b\}$. Additionally the correspondence yields $\lbars{b} = \sem{e_l}_{\eta}^\N$, so that the lemma follows by rule \eqref{eq:i-nonce}.
            \item
                  $l = \mcode{write}\ b$ and $\lambda = \miml{out}(e);$ by rules \eqref{eq:c-out} and \eqref{eq:s-out}.
                  
                  From the correspondence between the symbolic and the concrete transition we obtain $\sem{e}_\eta = b$. Additionally $\eta' = \eta$ by minimality of $\eta'$. The lemma follows by rule \eqref{eq:i-out}. 
                  
            \item
                  $l = \mcode{event}\ b$ and $\lambda = \miml{event}(e);$ by rules \eqref{eq:c-out} and \eqref{eq:s-out}.
                  
                  The proof is exactly analogous to the case above, the lemma follows by rule \eqref{eq:i-event}. 
                  
            \item
                  $l = \mcode{ctr}\ 1$ and $\lambda = \miml{if\ $e$ then}$ by rules \eqref{eq:c-test} and \eqref{eq:s-test}.
                  
                  From the correspondence between the symbolic and the concrete transition we obtain $\sem{e}_\eta = i1$. Additionally $\eta' = \eta$ by minimality of $\eta'$. The lemma follows by rule \eqref{eq:i-cond-true}.
      \end{enumerate}

\end{proof}

\begin{lemma}\label{cvm-iml-simulation}
      There exists a fixed polynomial $p$ such that for any $P \in \CVM$ with $\sem{P}_S \neq \bot$ 
      \[\sem{P}_C \lesssim_p \sem{\sem{P}_S}_I.\]
\end{lemma}

\begin{proof}
      Let $P$ be a CVM program such that $\sem{P}_S \neq \bot$. Let $T = \sem{P}_C$ and $\tilde T = \sem{\sem{P}_S}_I$. We shall show that $T \lesssim_p \tilde T$ for some polynomial $p$ by giving  a relation $\lesssim$ between states of $T$ and $\tilde T$ as well as a translation function $\tau$ that satisfy \cref{simulation}. Let $s_1, \ldots, s_n$ be the symbolic execution trace of $P$ with labels $\lambda_1, \ldots, \lambda_{n - 1}$ and let $\tilde P_i = \lambda_i\ldots \lambda_{n - 1} 0 \in \IML$. This way $\tilde P_1 = \sem{P}_S$ and $\tilde P_n = 0$. Let $\PS_1, \ldots, \PS_m$ be protocol states over $T$ such that $\PS_1 = \{\emptybs \mapsto (\eta_1, (\mcvm{Init}, P))\}$ for some initial environment $\eta_1$ and there is a transition $\PS_i \xrightarrow{(h_i, d_i),\, a_i} \PS_{i + 1}$ with a command $(h_i, d_i)$ and an action $a_i$ for each $i$. As CVM does not perform replication, each protocol state will be of the form $\PS_i = \{h_i \mapsto (\eta_i, c_i)\}$ for some state $c_i$ and valuation $\eta_i$.
      
      No concrete trace of CVM is longer than the symbolic trace (both are bounded by the number of instructions in $P$), so clearly $m \leq n$. By definition the initial symbolic state $s_1 = (\mcvm{Init}, P)$ is $\eta_1$-consistent and $\conc_{\eta_1}(s_1) = c_1$. By setting $\tilde \eta_1 = \eta_1$ and repeatedly applying \cref{cvm-to-iml-soundness-1} we obtain a sequence $\tilde \eta_1, \ldots, \tilde \eta_m$ of valuations such that for each $i$ the valuation $\tilde \eta_i$ is an extension of $\eta_i$, the state $s_i$ is $\tilde \eta_i$-consistent and $\conc_{\tilde \eta_i}(s_i) = c_i$. Additionally we can choose the valuations such that $\tilde \eta_{i + 1}$ is a minimal extension of $\tilde \eta_i$ satisfying the property. For each $i = 1, \ldots, m$ we define a protocol state $\tilde \PS_i$ over $\tilde T$ as $\tilde \PS_i = \{\tilde h_i \mapsto (\var(\tilde \eta_i), \tilde P_i)\}$, where $\tilde h_i$ is obtained from $h_i$ as follows: Let $I \subseteq \{1, \ldots, n - 1\}$ be the set of indices $i$ such that $\tilde P_i \neq \tilde P_{i + 1}$. Given a history $h_i$ of the form $h_i = o_1 1 \ldots o_{i - 1} 1$ (every CVM rule only has one process on the right hand side, so the replication identifier is always $1$) let $\tilde h_i = o_{i_1} 1 \ldots o_{i_k} 1$, where $\{i_1, \ldots, i_k\} = I \cap \{1, \ldots, i - 1\}$. 
      
      Given a protocol state $\PS$ over $T$ and a protocol state $\tilde \PS$ over $\tilde T$ we define $\PS \lesssim \tilde \PS$ iff there exist sequences of states $\PS_1, \ldots, \PS_m$ and $\tilde \PS_1, \ldots, \tilde \PS_m$ as above such that $\PS = \PS_i$ and $\tilde \PS = \tilde \PS_i$ for some $i$. We define the function $\tau$ from commands to sequences of commands as follows:
      \[\tau((h, d)) = \begin{cases}
                        (\tilde h, d), & \text{if $h = o_1 1 \ldots o_{i - 1} 1$, and $i \in I$,} \\
                        \emptybs & \text{otherwise.}
                       \end{cases}\]
      We now show that the relation $\lesssim$ and the function $\tau$ satisfy \cref{simulation}, so that $T \lesssim_p \tilde T$ for some polynomial $p$. The conditions in \cref{simulation} are satisfied as follows:
      \begin{enumerate}
            \item 
                  Any initial valuation $\eta_1$ is not an extended valuation so that $\var(\eta_1) = \eta_1$. By definition
                  \begin{align*}
                        \{\emptybs \mapsto (\eta_1, (\mcvm{Init}, P))\} & \lesssim \{\emptybs \mapsto (\var(\eta_1), \tilde P_1)\} \\
                              & = \{\emptybs \mapsto (\eta_1, \sem{P}_S)\}.      
                  \end{align*}

            \item
                  Let $\PS \lesssim \tilde \PS$ and assume that there exists a transition $\PS \xrightarrow{(h, d), \, a} \PS'$. By definition of the relation $\lesssim$ there exist sequences of states $\PS_1, \ldots, \PS_m$ and $\tilde \PS_1, \ldots, \tilde \PS_m$ as above such that $\PS = \PS_i$, $\PS' = \PS_{i + 1}$ and $\tilde \PS = \tilde \PS_i$ for some $i < m$. It suffices to show that 
                  \begin{equation*}
                        \tilde \PS_i \xrightarrow{\tau((h, d)),\, a}\!\!{}^*\; \tilde \PS_{i + 1}. \tag{*}
                  \end{equation*}
                  If $i \in I$ then (*) follows from \cref{cvm-to-iml-soundness-2}. Let $i \notin I$, that is $\lambda_i = \emptybs$ in the symbolic execution. Inspecting the proof of \cref{cvm-to-iml-soundness-1} we see that then $\var(\tilde \eta_i) = \var(\tilde \eta_{i + 1})$ and so $\tilde \PS_i = \tilde \PS_{i + 1}$. The program performs no action so that $a = \emptybs$ and by definition $\tau((h, d)) = \emptybs$, thus (*) is satisfied.

            \item
                  To compute $\tau$ it is necessary to know $I$, but this can be computed by an inspection of $P$ in linear time: it is $i + 1 \in I$ iff the $i$th instruction in $P$ is one of \cvm{In}, \cvm{Out}, or \cvm{Test}, that is, an instruction that generates a nonempty label $\lambda$ in the symbolic execution. Thus $\tau(c)$ is computable in time linear in $\lbars{c} + \lbars{P}$.
                  
            \item
                  Assume that for some valuation $\eta$ and attackers $E$ and $\tilde E$ the machine $M = \exec_\eta(T, E)$ reaches a state $\PS$ in $t$ steps and the machine $\tilde M = \exec_\eta(\tilde T, \tilde E)$ reaches a state $\tilde \PS$ in $\tilde t$ steps and $\PS \lesssim \tilde \PS$. If $\eta$ is the environment of the process in $\PS$ and $\tilde \eta$ is the environment of the process in $\tilde \PS$ then $\tilde \eta$ is an extension of $\eta$, in fact $\tilde \eta = \eta$, as both environments get updated by rules \eqref{eq:c-in} and \eqref{eq:i-nonce}, \eqref{eq:i-in} in the same way. It is easy to see that
                  \[\tilde t = O(\tilde n_{tr} \cdot (\tilde t_e + \lbars{\sem{P}_S} + \lbars{\tilde \eta})),\]
                  where $\tilde n_{tr}$ is the number of transitions performed by $\tilde M$ and $\tilde t_e$ is the number of steps to evaluate the most expensive IML expression during the execution of $\tilde M$. All of these values can be bounded in terms of $t$ as follows: The IML model $\sem{P}_S$ performs at most the same number of transitions as the PTS program $P$, so that $\tilde n_{tr} \leq n_{tr} \leq t$, where $n_{tr}$ is the number of transitions executed by $M$. By construction of the symbolic execution $\lbars{\sem{P}_S} = O(\lbars{P}) = O(t)$. Furthermore $\lbars{\tilde \eta} = \lbars{\eta} \leq t$. Finally we shall prove by induction that if $\tilde t_e$ is the number of steps to evaluate $\sem{e}_{\tilde \eta} = \sem{e}_{\eta}$ for some expression $e$ then $\tilde t_e = O(t) \cdot \lbars{e}$. Consider the following cases:
                  \begin{itemize}
                        \item 
                              $e = b$ for some $b \in \BS$. In this case $\tilde t_e = \lbars{e}$.
                              
                        \item
                              $e = x$ for some $x \in \Var$. In this case $\sem{e}_{\eta} = \eta(x)$, so that $\tilde t_e \leq t$.
                              
                        \item
                              $e = op(e_1, \ldots, e_n)$ with some $op \in \Ops$. For bitstrings $b_1, \ldots, b_n \in \BS$ let $t_{op}(b_1, \ldots, b_n)$ be the number of steps to evaluate $\opfun_{op}(e_1, \ldots, e_n)$ and let $\tilde t_i$ be the number of steps to evaluate $\sem{e_i}_{\eta}$. Every operation in $e$ is also performed by $M$, thus
                              \begin{align*}
                                    \tilde t_e & = t_{op}(\sem{e_1}_\eta, \ldots, \sem{e_n}_\eta) + \tilde t_1 + \ldots + \tilde t_n \\
                                          &\leq t_{op}(\sem{e_1}_\eta, \ldots, \sem{e_n}_\eta)  + \sum_i \lbars{e_i} \cdot O(t) \\
                                          &\leq O(t) \cdot \left(\sum_i \lbars{e_i} + 1 \right) \leq O(t) \cdot \lbars{e}.
                              \end{align*}

                        \item
                              $e = e_1 \concat e_2$. Let $\tilde t_1$ and $\tilde t_2$ be the number of steps to evaluate $\sem{e_1}_\eta$ and $\sem{e_2}_\eta$ respectively. Then
                              \begin{align*}
                                    \tilde t_e & \leq \tilde t_1 + \tilde t_2 + \lbars{\sem{e_1}_\eta} + \lbars{\sem{e_2}_\eta} \\
                                    & \leq 2 \cdot (\tilde t_1 + \tilde t_2) \leq O(t) \cdot (\lbars{e_1} + \lbars{e_2}) \leq O(t) \cdot \lbars{e}.
                              \end{align*}

                        \item
                              The cases $e = e'\{e_o, e_l\}$ and $e = \len(e')$ are proved analogously to the case $e = e_1 \concat e_2$.
                  \end{itemize}
      \end{enumerate}      
\end{proof}

\begin{restate}{\cref{symex}}
      There exists a fixed polynomial $p$ such that if $P_1, \ldots, P_n$ are CVM processes and for each $i$ $\tilde P_i := \sem{P_i}_S \neq \bot$ then for any IML process $P_E$, any trace property $\rho$, and resource bound $t \in \N$
      \begin{align*}
            &\insec(\sem{P_E[P_1, \ldots, P_n]}_{CI}, \rho,t) \\
            &\qquad \leq \insec(\sem{P_E[\tilde P_1, \ldots, \tilde P_n]}_{I}, \rho, p(t)).
      \end{align*}
\end{restate}

\begin{proof}
      By \cref{cvm-iml-simulation} there exists a polynomial $p_1$ such that $\sem{P_i}_C \lesssim_{p_1} \sem{\tilde P_i}_I$ for each $i$. By \cref{iml-holes} the PTS $\sem{P_E}_I$ is a PTS with holes identifiable in $p_2$-time for some fixed polynomial $p_2$. Applying \cref{mixed-semantics} and \cref{embedsim} we see that there exists a polynomial $p_3$ depending only on $p_1$ and $p_2$ (and thus fixed) such that 
      \begin{align*}
                  \sem{P_E[P_1, \ldots, P_n]}_{CI} = & \sem{P_E}_I[\sem{P_1}_C, \ldots, \sem{P_n}_C] \\
                  \lesssim_{p_3} & \sem{P_E}_I[\sem{\tilde P_1}_I, \ldots, \sem{\tilde P_n}_I] \\
                   = & \sem{P_E[\tilde P_1, \ldots, \tilde P_n]}_{I}.
      \end{align*}
      By \cref{simsec} there exists a polynomial $p_4$ depending only on $p_3$ (and thus fixed) such that \cref{symex} holds with $p = p_4$.
\end{proof}

\section{Verification of IML---Details \status{good, 27.05.2011}}\label{iml-verification-details}

We show how to simplify IML to the applied pi calculus that can be verified using ProVerif. As ProVerif works in the symbolic model, we shall employ a computational soundness result from \cite{CoSP} to justify its use. The result will guarantee that if ProVerif successfully verifies the translated pi calculus process then the process is asymptotically secure in our computational model. We start by illustrating the method on an example and then give a general description.

The main challenge when translating IML to the pi calculus is that IML processes contain bitstring manipulation primitives that are not valid in pi. An example of such a process is shown in \cref{fig:example-NSL-IML}---it is an adapted excerpt from an IML model of the Needham-Shroeder-Lowe protocol implementation used in one of our experiments (the full model is shown in \cref{nsl-example}). The key observation is that the bitstring manipulation expressions in IML are most commonly employed to provide the tupling functionality. In our example the process \iml{A} uses concatenations to construct a computational representation of the pair of $n_A$ and $pk_A$. Similarly, process \iml{B} uses range expressions to extract the second element of the pair. The idea of the translation is thus to enrich $\Ops$ with encoding and parsing operations with meanings given by the bitstring manipulation expressions. This way we hide the direct bitstring manipulation inside new opaque operations. Of course, to obtain a soundness result we need to prove certain properties of the extracted operations to make sure that they correctly implement tupling.

\begin{figure}[t]
\small
\begin{lstlisting}[language = iml, gobble = 6, xleftmargin = 1em]
      A =
        ($\tilde \nu\, n_A$); ($\tilde \nu\, r$); 
        let $m_1 = $ "msg1"$\concat \len(n_A) \concat n_A \concat pk_A$ in
        let $e_1 = encrypt(pk_X, m_1)$ in
        out($e_1$); ...
        
      B = 
        in($e_1$); 
        let $m_1 = decrypt(sk_B, e_1)$ in
        if $m_1\{i4, iN\} +\bop iN +\bop i4 \leq \len(m_1)$ then
        if $m_1\{i0, i4\} = $ "msg1" then
        let $x_1 = m_1\{i4 +\bop iN +\bop m_1\{i4, iN\},$
                       $ \len(m_1) -\bop i4 -\bop iN -\bop m_1\{i4, iN\}\}$ in
        if $x_1 = pk_X$ then ...
\end{lstlisting}
\caption{An excerpt from the IML process for the NSL protocol. An expression $\len(\ldots)$ produces a result of fixed length $iN$.}
\label{fig:example-NSL-IML}
\end{figure}

\begin{figure}[t]
\small
\begin{lstlisting}[language = iml, gobble = 6, xleftmargin = 1em]
      A =
        ($\tilde \nu\, n_A$); ($\tilde \nu\, r$); 
        out($encrypt(pk_X, conc_1(n_A, pk_A))$); ...
        
      B = 
        in($e_1$); 
        let $m_1 = decrypt(sk_B, e_1)$ in
        let $x_1 = parse_2(m_1)$ in
        if $x_1 = pk_X$ then ...        
\end{lstlisting}
\caption{An excerpt from the pi calculus translation for the NSL protocol.}
\label{fig:example-NSL-pi}
\end{figure}

In our example we introduce new operations $conc_1$ and $parse_2$ with implementations given by
\begin{align*}
      & \opfun_{conc_1}(b_1, b_2) = \sem{\miml{"msg1"}\concat \len(b_1) \concat b_1 \concat b_2}, \\
      & \opfun_{parse_2}(b) = \\ 
            &\;\;\text{if $\sem{\neg(b\{i4, iN\} +\bop iN +\bop i4 \leq \len(b))}$ then $\bot$ else} \\
            &\;\;\text{if $\sem{\neg(b\{i0, i4\} = \miml{"msg1"})}$ then $\bot$ else} \\
            &\;\;\llbracket b\{i4 +\bop iN +\bop b\{i4, iN\}, \\
            &\;\;\;\;\;\;\; \len(b) -\bop i4 -\bop iN -\bop b\{i4, iN\}\}\rrbracket.
\end{align*}
In the implementation of the parsing expression we keep all the condition checks that are performed by the IML process before applying the parser. We follow the convention of IML that $i1$ and $i0$ represent truth values of bitstrings. Using the new operations we can simplify our example IML process to the pi calculus process shown in \cref{fig:example-NSL-pi}, removing the if-statements that have been absorbed into the implementation of $parse_2$. 

\begin{figure}
\small
\begin{align*}
      \mathrlap{b \in \BS,\, x \in \Var,\, op \in \Ops}\quad\quad\quad \\
      e \in \PExp & ::= && \hspace{-\bigskipamount}\text{expression} \\
            & x && \text{variable} \\
            & op(e_1, \ldots, e_n) && \text{constructor/destructor} \\
      \imlP,\, \imlQ & ::= && \hspace{-\bigskipamount}\text{process} \\
            & 0 && \text{nil} \\
            & !\imlP && \text{replication} \\
            & \imlP \concat \imlQ && \text{parallel composition} \\
            & \miml{$(\tilde\nu x)$;\ $\imlP$} && \text{randomness}\\
            & \miml{in$(x)$;\ $\imlP$} && \text{input} \\
            & \miml{out$(x)$;\ $\imlP$} && \text{output} \\
            & \miml{event$(b)$;\ $P$} && \text{event} \\
            & \miml{let\ $x = e$ in\ $\imlP$ [else\ $\imlQ$]} && \text{evaluation} \\
\end{align*}
\caption{The syntax of the applied pi calculus.}
\label{fig:pi-syntax}
\end{figure}

The syntax of the applied pi calculus is shown in \cref{fig:pi-syntax}. It is a strict subset of the IML syntax with the following differences:
\begin{itemize}
      \item 
            The bitstring operations are no longer available.
            
      \item
            The only allowed form of the restriction operator is $(\tilde\nu x)$ with the same meaning as described in \cref{iml}. 
            
      \item
            Parameters of events are restricted to be fixed bitstrings. This is a limitation of the result in \cite{CoSP}.
            
      \item
            The conditional expression of IML with truth meanings for bitstrings $i0$ and $i1$ is no longer available. Instead we can use let expressions to conditionally choose based on equality of bitstrings by assuming that there exists an operation $eq \in \Ops$ such that $\opfun_{eq}(b, b) = b$ and $\opfun_{eq}(b, b') = \bot$ for all $b \neq b'$. 

      \item
            The input and output expressions only accept variables as parameters---all computations must be performed in let-expressions.
\end{itemize}

The calculus shown in \cref{fig:pi-syntax} is a restricted version of the pi calculus presented in \cite{CoSP}, as we do not need the full generality used there. Our restrictions are as follows:
\begin{itemize}
      \item 
            There is only one public communication channel.
            
      \item
            We do not make a distinction between variables and names, as they behave identically for the purpose of the computational execution.
            
      \item
            We only allow computations in let-expressions, so that we do not make a distinction between constructors and destructors in the syntax.
\end{itemize}

\begin{figure}
\small
\begin{align*}
      &\sem{x}^k_\eta = \eta(x), \;\text{for $x \in \Var$,}\\
      &\sem{op(e_1, \ldots, e_n)}^k_\eta = \tilde\opfun_{op}(k, \sem{e_1}^k_\eta, \ldots, \sem{e_n}^k_\eta).
\end{align*}
\caption{The evaluation of pi expressions, whereby $\bot$ propagates.}
\label{fig:pi-eval}
\end{figure}

\begin{figure}
\small
\begin{align*}
      &\frac{b \in \BS,\quad \lbars{b} = k, \quad r = \tilde \opfun_{nonce}(k, b)}{(\eta,\, (\tilde \nu x); \imlP) \xrightarrow{\mathtt{rnd}\ b} \{(\eta\{x \mapsto r\},\, \imlP)\}}. \tag{pi-Nonce}\eqlabel{eq:pi-nonce}\\
\end{align*}
\caption{Randomness generation in pi calculus.}
\label{fig:pi-semantics}
\end{figure}

Unlike CVM and IML which execute with regards to a fixed security parameter $k_0$ introduced in \cref{cvm}, the computational semantics of the applied pi calculus is parameterised by a security parameter. In order to achieve that we assume that the operations in $\Ops$ possess a \emph{generalised implementation} $\tilde \opfun$ such that $\tilde\opfun_{op} \colon \N \times \BS^{\arity(op)} \pto \BS$ is the implementation of an operation $op \in \Ops$ that takes the security parameter as the first argument. For a security parameter $k$ and inputs $\underline m$ the value $\tilde\opfun_{op}(k, \underline m)$ should be computable in time polynomial in $k + \lbars{\underline m}$. We require that $\tilde\opfun_{op}(k_0, \cdot) = \opfun_{op}$ for each $op \in \Ops$. 

The semantics of the pi calculus is directly derived from the semantics of IML. Given a pi process $P$ and a security parameter $k$, we define the semantics $\sem{P}^k_\pi$ as follows: The expression evaluation uses $\tilde \opfun$ instead of $\opfun$ as shown in \cref{fig:pi-eval}. The semantics rules are obtained from the IML rules (\cref{fig:iml-semantics}) by substituting all expression evaluations $\sem{e}_\eta$ with $\sem{e}^k_\eta$. The syntactic form $(\tilde \nu x)$ behaves as described in \cref{iml}, but now it is not a syntactic sugar anymore, so we add a new semantic rule shown in \cref{fig:pi-semantics}.

We now give details regarding the translation procedure from IML to pi. In the following we shall assume that the IML processes do not contain else-branches; this is true for the processes produced by the symbolic execution. Removing if-statements from such processes does not reduce the set of traces and thus does not reduce insecurity. We shall therefore divide all if-statements into two groups: the \emph{cryptographic} statements, that are likely to be relevant for the security of the process and should be kept in the translation, and the \emph{auxiliary} statements that can be removed from the process without affecting security. The exact choice does not affect the soundness of the approach, but removing too many if statements might make the resulting pi process insecure, and removing too few may prevent the successful translation from IML to pi. We use the following heuristic: an if-statement is considered to be cryptographic iff it is of the form \iml{if $e_1 = e_2$ then $P$}, where both $e_1$ and $e_2$ are variables or applications of cryptographic operations. 

Given an IML process $P$ we perform on it the following operations:
\begin{itemize}
      \item
            Introduce intermediate let-statements so that all out-statements only contain variables, all cryptographic if-statements are of the form \iml{if $x_1 = x_2$ then $P$} with variables $x_1$ and $x_2$ and every expression in the new let statements is of one of three types:
            \begin{itemize}
                  \item
                        an \emph{encoding expression}, that is, an expression containing only concrete bitstrings, $\len()$, concatenations, arithmetic operations, and variables,
                        
                  \item
                        a \emph{parsing expression}, that is, an expression containing only concrete bitstrings, $\len()$, substring extraction, arithmetic operations, and a single variable,
                        
                  \item
                        a \emph{cryptographic expression}, that is, an expression containing only variables and cryptographic operations.
            \end{itemize}
            
            As an example, the IML processes in \cref{fig:example-NSL-IML} are already written in such a form.
            
      \item
            For each subprocess \iml{$P' = $ (let $y = e$ in $P''$)}, where $e$ is an encoding expression with variables $x_1, \ldots, x_n$, add a new \emph{encoding operation} $c$ of arity $n$ to $\Ops$ with the implementation given by
            \[\opfun_{c}(b_1, \ldots, b_n) = \sem{e[b_1 / x_1, \ldots, b_n / x_n]}.\]
            Now substitute $P'$ by \iml{let $y = c(x_1, \ldots, x_n)$ in $P''$}. 
            
            In order to justify modelling the encoding operations as tuples symbolically, we need to check that their computational implementations fulfil certain conditions. The first condition is:
            \begin{itemize}
                  \item[(C1)]
                        the ranges of the functions $A_c$ introduced above are disjoint.
            \end{itemize}
            Checking the side conditions is described in \cref{parsing-conditions}.
           
      \item
            For each subprocess \iml{$P' = $ (let $y = e$ in $P''$)}, where $e$ is a parsing expression with a variable $x$, add a new \emph{parsing operation} $p$ of arity $1$ to $\Ops$. We need to check that before computing $e$ the process $P$ makes sure that $x$ contains a result of a suitable encoding operation. More specifically, we check that there exists an encoding operation $c$ such that the process rejects any $x$ with the value outside the range of $\opfun_c$ and such that $e$ computes an inverse of $A_c$. Let $e_1, \ldots, e_n$ be expressions such that $P$ contains an auxiliary if-statement of the form \iml{if $e_i$ then $\ldots$} above $P'$ for some $i$. Let $x_1, \ldots, x_m$ be the variables of $P$ with exception of $x$ and let
            \[\phi_p = \exists x_1,\,\ldots,\, x_m \colon e_1 \wedge \ldots \wedge e_n.\]
            This way, whenever $(\eta', P')$ is an executing process in a protocol state reached by $\sem{P}_I$ from some environment $\eta$, we have $\sem{\phi_p}_{\eta'} = i1$. We check the following conditions:
            \begin{itemize}

                  \item[(C2)]
                        there exists an encoding operation $c$ such that for every $b$ not in the range of $\opfun_c$ it is $\sem{\phi_p[b/x]} = i0$. We say that $c$ \emph{matches} $p$,
                        
                  \item[(C3)]
                        the function $f_p \colon b \mapsto \sem{e[b/x]}$ is an $i$th inverse of $\opfun_c$ for some $i$, that is, $f_p(\opfun_c (b_1, \ldots, b_n)) = b_i$ where $n$ is the arity of $c$. 
            \end{itemize}
            \Cref{parsing-conditions} shows how to check the conditions (C1)--(C3) and how a successful check results in a quantifier-free formula $\phi_p'$ with $x$ as the only variable such that $\phi_p$ implies $\phi_p'$ and the condition (C2) is still satisfied with $\phi_p'$. Additionally $\phi_p'$ satisfies
            \begin{itemize}
                  \item[(C4)]
                        for the encoding operation $c$ that matches $p$ and any $b$ in the range of $A_c$ it is $\sem{\phi_p[b/x]} = i1$.
            \end{itemize}
            We define the computational implementation for $p$ as 
            \[A_p(b) = \text{if $\sem{\phi_p'[b/x]}$ then $\sem{e[b/x]}$ else $\bot$}\]
            and substitute $P'$ by \iml{let $y = p(x)$ in $P''$}. 
            
      \item
            Remove all auxiliary if-statements: for every such statement replace \iml{if $e$ then $P'$} by $P'$. Translate all cryptographic if-statements into the form expected by the pi-calculus: replace every occurrence of \iml{if $x_1 = x_2$ then $P$} by \iml{let $\_ = eq(x_1, x_2)$ in $P$}.
\end{itemize}

If the process $P$ does not contain any else-branches and the above procedure yields a valid pi process $\tilde P$ then we say that $P$ \emph{is translatable to $\tilde P$}. A complete example of an IML program and its resulting pi calculus translation for the NSL protocol is shown in \cref{nsl-example}. 

In order to obtain the computational semantics for the translated process, we need to specify the generalised implementations $\tilde \opfun_c$ and $\tilde \opfun_p$ for the newly introduced encoders and parsers. We can assume \emph{any} generalisation of these operations to arbitrary security parameters that satisfies the conditions (C1)--(C4).

Clearly the translation preserves all the action sequences of the original process so the following holds:

\begin{lemma}
      There exists a fixed polynomial $p$ such that for any IML process $P$ translatable to a pi process $\tilde P$
      \[\sem{P}_I \lesssim_p \sem{\tilde P}^{k_0}_\pi.\]
\end{lemma}

Applying \cref{simsec} we obtain a statement that links the security of the pi translation to the security of the original IML process:

\begin{restate}{\cref{transsound}}
      There exists a fixed polynomial $p$ such that for any IML process $P$ translatable to a pi process $\tilde P$, any trace property $\rho$ and resource bound $t \in \N$
      \[\insec(\sem{P}_I, \rho, t) \leq \insec(\sem{\tilde P}^{k_0}_\pi, \rho, p(t)).\]
\end{restate}

Now that we have translated IML to pi, we can enumerate the conditions under which the resulting pi process can be soundly verified using ProVerif. For this purpose we shall make use of a computational soundness result from \cite{CoSP}, which places restrictions on the operation set $\Ops$ as well as on the shape of the pi process. More specifically, the computational soundness theorem is proved there for the set of constructors $\bC = \{E/3, ek/1, dk/1, pair/2\}$ and destructors $\bD = \{D/2, isenc/1, isek/1, ekof/1, fst/1, snd/1, eq/2\}$. The result includes soundness for signatures, but we omit them as they have not been used in our experiments so far. For simplicity the result presented here uses only one pairing construct (as in \cite{CoSP}), but it can be easily extended to an arbitrary number of tupling constructors and destructors, to correspond to our encoding and parsing operations introduced during the translation from IML. The symbolic behaviour of the operations is defined by the following equations:
\begin{align*}
      D(dk(t_1), E(ek(t_1), m, t_2)) & = m, \\
      isenc(E(ek(t_1), t_2, t_3)) & = E(ek(t_1), t_2, t_3), \\
      isek(ek(t)) & = ek(t), \\
      ekof(E(ek(t_1), m, t_2)) & = ek(t_1), \\
      fst(pair(x, y)) & = x, \\
      snd(pair(x, y)) & = y, \\
      eq(x, x) & = x.
\end{align*}

Let $\Ops^S = \bC \cup \bD \cup \{nonce\}$. The \emph{soundness conditions} that the implementations $\tilde \opfun_x$ for $x \in \Ops^S$ need to satisfy are as follows:
\begin{enumerate}
      \item 
            There are disjoint and efficiently computable sets of bitstrings representing the types nonces, ciphertexts, encryption keys, decryption keys, and pairs. Let $Nonces_k$ denote the set of all nonces for a security parameter $k$.
            
      \item
            Given $b \in \BS$ with $\lbars{b} = k$ chosen uniformly at random, $\tilde \opfun_{nonce}(k, b)$ returns $r \in Nonces_k$ uniformly at random.
            
      \item
            The functions $\tilde \opfun_E$, $\tilde \opfun_{ek}$, $\tilde \opfun_{dk}$, and $\tilde \opfun_{pair}$ are length-regular---the length of their result depends only on the lengths of their parameters. All $m \in Nonces_k$ have the same length.
            
      \item
            Every image of $\tilde \opfun_E$ is of type ciphertext, every image of $\tilde \opfun_{ek}$ and $\tilde \opfun_{ekof}$ is of type encryption key, every image of $\tilde \opfun_{dk}$ is of type decryption key.
            
      \item
            For all $m_1, m_2 \in \BS$ we have $\tilde \opfun_{fst}(\tilde \opfun_{pair}(m_1, m_2)) = m_1$ and $\tilde \opfun_{snd}(\tilde \opfun_{pair}(m_1, m_2)) = m_2$. Every $m$ of type pair is in the range of $\tilde \opfun_{pair}$. If $m$ is not of type pair, $\tilde \opfun_{fst}(m) = \tilde \opfun_{snd}(m) = \bot$.
            
      \item
            $\tilde \opfun_{ekof}(\tilde \opfun_E(p, x, y)) = p$ for all $p$ of type encryption key, $x \in \BS$, and a nonce $y$. $\tilde \opfun_{ekof}(e) \neq \bot$ for any $e$ of type ciphertext and $\tilde \opfun_{ekof}(e) = \bot$ for any $e$ that is not of type ciphertext. 
            
      \item
            $\tilde \opfun_E(p, m, y) = \bot$ if $p$ is not of type encryption key.
            
      \item
            $\tilde \opfun_D(\tilde \opfun_{dk}(r), m) = \bot$ if $r \in Nonces_k$ and $\tilde \opfun_{ekof}(m) \neq \tilde \opfun_{ek}(r)$.
            
      \item
            $\tilde \opfun_D(\tilde \opfun_{dk}(r), \tilde \opfun_E(\tilde \opfun_{ek}(r), m, r')) = m$ for all $r, r' \in Nonces_k$.
            
      \item
            $\tilde \opfun_{isek}(x) = x$ for any $x$ of type encryption key. $\tilde \opfun_{isek}(x) = \bot$ for any $x$ not of type encryption key.
            
      \item
            $\tilde \opfun_{isenc}(x) = x$ for any $x$ of type ciphertext. $\tilde \opfun_{isenc}(x) = \bot$ for any $x$ not of type ciphertext.

      \item
            We define an encryption scheme $(KeyGen, Enc, Dec)$ as follows: $KeyGen$ picks a random $r$ in $Nonces_k$ and returns $(\tilde \opfun_{ek}(r), \tilde \opfun_{dk}(r))$. $Enc(p, m)$ picks a random $r$ in $Nonces_k$ and returns $\tilde \opfun_E(p, m, r)$. $Dec(k, c)$ returns $\tilde \opfun_D(k, c)$. We require that the defined encryption scheme is IND-CCA secure.
            
      \item
            For all $e$ of type encryption key and $m \in \BS$ the probability that $\tilde \opfun_E(e, m, r) = \tilde \opfun_E(e, m, r')$ for uniformly chosen $r, r' \in Nonces_k$ is negligible.
\end{enumerate}

The conditions on the pairing operations follow from the conditions (C1)--(C4) checked during the translation (length-regularity is fulfilled for any function given by an IML encoding expression), the other conditions (in particular that the encryption is IND-CCA) shall be assumed, because we are treating cryptographic operations as black boxes and not trying to verify them. The condition that all functions have disjoint ranges is quite restrictive and is unlikely to be fulfilled in actual implementations. For this reason in future we would like to use CryptoVerif to verify our models, to bypass the need for complex soundness conditions.

\newcommand{\widebar}{\;\singlebar\;}
\begin{figure}
\small
\begin{align*}
      m,\, n ::= & \; x \widebar pair(m, n)  \\
      e  ::= & \; m \;\widebar isek(e) \widebar isenc(e) \widebar D(x_d, e) \widebar fst(e)  \\
         & \widebar snd(e) \widebar ekof(e) \widebar eq(e, e) \\
      \imlP,\, \imlQ ::= 
            & \; \miml{out$(x)$;\ $\imlP$} \widebar \miml{in$(x)$;\ $\imlP$} \widebar 0 \widebar !\imlP \widebar (\imlP \concat \imlQ) \widebar \miml{$(\tilde\nu x)$;\ $\imlP$} \\
            & \widebar \miml{let\ $x = e$ in\ $\imlP$ [else\ $\imlQ$]} \widebar \miml{event$(b)$;\ $P$} \\
            & \widebar \miml{$(\tilde\nu r)$;\ let\ $x = ek(r)$ in let\ $x_d = dk(r)$ in\ $\imlP$} \\
            & \widebar \miml{$(\tilde\nu r)$;\ let\ $x = E(isek(D_1), D_2, r)$ in\ $\imlP$ [else\ $\imlQ$]} 
\end{align*}
\caption{The syntax of key-safe processes.}
\label{fig:key-safe}
\end{figure}

The soundness result of \cite{CoSP} is proved for a class of the so-called \emph{key-safe} processes. In a nutshell, key-safe processes always
use fresh randomness for encryption and key generation and only use honestly generated (that is, through key generation) decryption keys for decryption. Decryption keys may not be sent around (in particular, this avoids the key-cycle problems). The grammar of key-safe processes is summarised in \cref{fig:key-safe}. We let $x$, $x_d$, $k_s$, and $r$ stand for different sets of variables: general purpose, decryption key, signing key, and randomness variables. 

\begin{lemma}[Computational soundness \cite{CoSP}]
      If a closed key-safe process symbolically satisfies a trace property $\rho$ then it computationally satisfies $\rho$.
\end{lemma}

We now proceed to sketching out the proof of \cref{compsound} from \cref{iml-verification}. For a process $P$ let $\Ops_P$ be the set of operations used by $P$ (including the $nonce$ operation). The symbolic semantics and security of pi are defined in \cite{CoSP}. We do not detail the semantics here, as we only need to know that it is exactly the semantics that is used by ProVerif. %

A function $f \colon \N \to \R$ is called \emph{negligible} if for every $c \in \N$ there exists $n_0 \in \N$ such that $f(n) < 1/n^c$ for all $n > n_0$.

\begin{restate}{\cref{compsound}}
      Let $P$ be a pi process such that $\Ops_{P} \subseteq \Ops^S$ and the soundness conditions are satisfied. If $P$ is key-safe and symbolically secure with respect to a trace property $\rho$ then for every polynomial $p$ the following function is negligible in $k$:
      \[\insec(\sem{P}^k_\pi,\, \rho,\, p(k)).\]
\end{restate}

The main issue in the proof is to relate the notion of computational execution in \cite{CoSP} (their definition 18) to our notion of computational execution (\cref{compex}). Both definitions are very similar. In \cite{CoSP} the state of the protocol consists of a single executing process together with valuations for variables in the process. In each step the attacker chooses an execution context to specify which subprocess of the complete process is supposed to perform a reduction. In our definition the attacker interacts with a multiset of processes, selecting the process to be executed by an attached handle. It is easy to see that both definitions of the security game are equivalent.

\subsection{Parsing Conditions}\label{parsing-conditions}

\begin{figure}[t]
\small
\begin{lstlisting}[language = iml, gobble = 6, xleftmargin = 1em]
      A =
        ($\tilde \nu\, n_A$); ($\tilde \nu\, r$); 
        let $m_1 = $ "msg1"$\concat \len(n_A) \concat n_A \concat pk_A$ in
        let $e_1 = encrypt(pk_X, m_1)$ in
        out($e_1$); ...
        
      B = 
        in($e_1$); 
        let $m_1 = decrypt(sk_B, e_1)$ in
        if $\len(pk_X) + \bop iN +\bop i20 +\bop i4 = \len(m_1)$ then
        if $m_1\{i0, i4\} = $ "msg1" then
        if $m_1\{i4, iN\} = i20$ then
        let $x_1 = m_1\{i4 +\bop iN +\bop m_1\{i4, iN\},$
                       $ \len(m_1) -\bop i4 -\bop iN -\bop m_1\{i4, iN\}\}$ in
        if $x_1 = pk_X$ then ...
\end{lstlisting}
\caption{An excerpt from the IML process for the NSL protocol (full version).}
\label{fig:example-NSL-IML-2}
\end{figure}

We show how we check conditions (C1)--(C4) arising during the translation from IML to pi. The checks we perform are by no means complete (we might fail to detect that the conditions actually hold), but they are suitable for the protocols that we encountered so far. We shall use the excerpt from the IML process of the NSL protocol shown in \cref{fig:example-NSL-IML-2} as an example (\cref{fig:example-NSL-IML} contained a slightly simplified version).

For each encoding operation $c$ and parsing operation $p$ let $e_c$ and $e_p$ be the IML expressions that they replace. Let $\phi_p$ represent the set of facts that the IML process establishes before applying $e_p$, as described previously.

To prove (C1) we check that all encoding expressions $e_c$ contain a concrete bitstring (a tag) at the same positions and that all tags are different. In the example of \cref{fig:example-NSL-IML-2} the bitstring \iml{"msg1"} would be such a tag, and we would expect other messages to contain tags like \iml{"msg2"}, \iml{"msg3"}, etc.

To prove (C3) for an encoder $c$ and a parser $p$ we check that $\simplify_{\facts_{op}}(e_p[e_c/x]) = x_i,$
where $x$ is the variable of $e_p$ and $x_i$ is one of the variables of $e_c$. As an example, for the operations $conc_1$ and $parse_2$ introduced at the beginning of \cref{iml-verification-details},
\begin{align*}
      e_{conc_1} &= \miml{"msg1"}\concat \len(x_1) \concat x_1 \concat x_2, \\
      e_{parse_2} & = x\{i4 +\bop iN +\bop x\{i4, iN\}, \\
                  & \quad \quad \len(x) -\bop i4 -\bop iN -\bop x\{i4, iN\}.
\end{align*}
Substituting $e_{conc_1}$ for $x$ in $e_{parse_2}$ we obtain an expression that simplifies to $x_2$, thus we know that $e_{parse_2}$ computes the second inverse of $e_{conc_1}$.

Given a parser $p$ and a candidate encoder $c$, we check whether $c$ matches $p$ (C2) as follows: first check that $e_c$ is a concatenation of expressions, each of which is either a variable (a concatenation parameter), a length of a variable, or a constant expression. Formally $e_c$ is required to be of a form $e_1 \concat \ldots \concat e_n$, where $\{1, \ldots, n\} = I_x \cup I_l \cup I_t$ such that for all $i \in I_x$ it is $e_i = x_i$ for some variable $x_i$, for all $i \in I_l$ it is $e_i = \len(x_j)$ for some $j \in I_x$ and for all $i \in I_t$ it is $e_i = b_i$ for some constant bitstring $b_i$. We require that all variables and length expressions are distinct (no variable repeats twice) and that $\lbars{I_x} = \lbars{I_l} + 1$, that is, the expression $e_c$ contains lengths for all parameters except one---the missing length can then be derived from knowing the total length of the concatenation. 

Given a bitstring $b$, in order to check that $b$ is in the range of $A_c$, it is sufficient to check all the constant (tag) fields and to check that the sum of the length fields is consistent with the actual length of $b$. The following makes this precise.

Given a parsing expression $p_i$, we say that $p_i$ \emph{extracts the $i$th field from $e_c$} if the following holds: for an expression $e$ let $e_c[e / e_i]$ be the expression obtained from $e_c$ by substituting $e_i$ with $e$. Then for a fresh variable $x'$
\begin{align*}
      & \simplify_{\Sigma}(p_i[e_c[x'/e_i] / x]) = x', \\
      & \text{where} \; \Sigma = \facts_{op} \cup \{\len(x') = \getLen(e_i)\}.
\end{align*}

\begin{theorem}
      Let $c$ and $p$ be an encoding and a parsing expression such that $e_c$ is of a form $e_1 \concat \ldots \concat e_n$ with $\{1, \ldots, n\} = I_x \cup I_l \cup I_t$ as described above. Assume that for each $i \in I_l \cup I_t$ the formula $\phi_p$ contains a parsing expression $p_i$ as a term, such that $p_i$ extracts the $i$th field from $e_c$. Let
      \begin{align*}
            \phi_{tag} &= \bigwedge_{i \in I_t} p_i = b_i, \\
            \phi_{len} &= \sum_{i \in I_l} p_i + \sum_{i \in I_t \cup I_l} \getLen(e_i) \leq \len(x).
      \end{align*}
      Then a bitstring $b$ is in the range of $A_c$ iff 
      \[\sem{\phi_{tag} \wedge \phi_{len}}_{x \mapsto b} = i1.\]
\end{theorem}

\begin{proof}[sketch]
      Let $b \in \BS$ satisfy the premises of the theorem. For each $i \leq n$ we obtain the length $l_i \in \N$ of the $i$th field in $b$ as follows: for each $i \in I_l$ such that $e_i = \len(x_j)$ for some $j \in I_x$ let $l_j = \sem{p_i[b / x]}^\N$. For each $i \in I_l \cup I_t$ let $l_i = \sem{\getLen(e_i)}^\N$. For the single $i \in I_x$ such that $\len(x_i)$ is not one of the fields of $e_c$ let 
      $l_i = \lbars{b} - \sum_{j \neq i} l_j$.
      Knowing the lengths allows us to split $b$ into fields as follows: for each $i \leq n$ let $b_i = b\{\sum_{j = 1}^{i - 1} l_j,\, l_i \}$. This is well-defined according to $\phi_{len}$. Clearly $b = b_1 \concat \ldots \concat b_n$. We show that for each $i$ it is $b_i = \sem{e_i[b_j/x_j \singlebar j \in I_x]}$ as follows.
      \begin{itemize}
            \item 
                  If $i \in I_x$ then $e_i = x_i$ and the equality holds trivially.
                  
            \item
                 If $i \in I_l$ then $e_i = \len(x_j)$ for some $j \in I_x$. By construction $b_i = \bs(l_i) = \bs(\lbars{b_j})$.
                  
            \item
                  If $i \in I_t$ then the equality follows from $\phi_{tag}$.
      \end{itemize}
      Overall we have shown that $b = \sem{e_c[b_j/x_j \singlebar j \in I_x]}$, so that $b$ is in the range of $A_c$.
\end{proof}

Thus checking (C2) reduces to finding appropriate parsers $p_i$ among the terms of $\phi_p$ and checking that $\phi_p \vdash \phi_{tag} \wedge \phi_{len}$. Furthermore, by choosing $\phi_p' = \phi_{tag} \wedge \phi_{len}$, we obtain a quantifier-free formula that satisfies (C2) and (C4), as required by the translation.

As an example, we can show that (C2) holds for $conc_1$ and $parse_2$ with respect to \cref{fig:example-NSL-IML-2} as follows: the conditions checked by the process $B$ contain references to parsing expressions $m_1\{i0, i4\}$ and $m_1\{i4, iN\}$. We check that the first expressions extracts the first field (the tag) from $e_{conc_1}$ and the second expression extracts the second field (the length of the first parameter). We then observe that the conditions checked by $B$ imply
\begin{align*}
      \phi_{tag} & = (m_1\{i0, i4\} = \miml{"msg1"}), \\
      \phi_{len} & = (iN +\bop m_1\{i4, iN\} +\bop i4 \leq \len(m_1)).
\end{align*} 
Thus both the tag and the length consistency are properly checked.

Our implementation currently checks all the conditions automatically except $\phi_p \vdash \phi_{len}$. The reason is that we are planning to use CryptoVerif as a verification backend and expect to be able to relax the parsing conditions there.

\section{NSL Example Code \status{good, 27.05.2011}}\label{nsl-example}

We show all the stages of the verification of the NSL example, discussed in \cref{implementation}

\subsection{Client Source}

The source code of the client is shown below. In our example $N = \mcryptoc{sizeof(size_t)} = 8$ and $k_0$ corresponds to \cryptoc{SIZE_NONCE}, which is set to be $20$.

\begin{lstlisting}[language=cryptoc, basicstyle = \small]

#include <net.h>
#include <lib.h>

#include <proxies/common.h>

#include <string.h>
#include <stdio.h>

// #define LOWE_ATTACK

int main(int argc, char ** argv)
{
  unsigned char * pkey, * skey, * xkey;
  size_t pkey_len, skey_len, xkey_len;

  unsigned char * m1, * m1_all;
  unsigned char * Na;
  size_t m1_len, m1_e_len, m1_all_len;

  unsigned char * m2, * m2_e;
  unsigned char * xNb;
  size_t m2_len, m2_e_len;
  size_t m2_l1, m2_l2;

  unsigned char * m3_e;
  size_t m3_e_len;

  unsigned char * p;

  // for encryption tags
  unsigned char * etag = malloc(4);

  BIO * bio = socket_connect();

  pkey = get_pkey(&pkey_len, 'A');
  skey = get_skey(&skey_len, 'A');
  xkey = get_xkey(&xkey_len, 'A');

  /* Send message 1 */

  m1_len = SIZE_NONCE + 4 + pkey_len
      + sizeof(size_t);
  p = m1 = malloc(m1_len);

  memcpy(p, "msg1", 4);
  p += 4;
  * (size_t *) p = SIZE_NONCE;
  p += sizeof(size_t);
  Na = p;
  nonce(Na);
  p += SIZE_NONCE;
  memcpy(p, pkey, pkey_len);

  m1_e_len = encrypt_len(xkey, xkey_len,
                         m1, m1_len);
  m1_all_len = m1_e_len + sizeof(size_t) + 4;
  m1_all = malloc(m1_all_len);
  memcpy(m1_all, "encr", 4);
  m1_e_len =
      encrypt(xkey, xkey_len, m1,
              m1_len,
              m1_all + sizeof(m1_e_len) + 4);
  m1_all_len = m1_e_len + sizeof(size_t) + 4;
  * (size_t *) (m1_all + 4) = m1_e_len;

  send(bio, m1_all, m1_all_len);

  /* Receive message 2 */

  recv(bio, etag, 4);
  recv(bio, (unsigned char*) &m2_e_len,
       sizeof(m2_e_len));
  m2_e = malloc(m2_e_len);
  recv(bio, m2_e, m2_e_len);

  m2_len = decrypt_len(skey, skey_len,
                       m2_e, m2_e_len);
  m2 = malloc(m2_len);
  m2_len =
      decrypt(skey, skey_len,
              m2_e, m2_e_len, m2);

  if(xkey_len + 2 * SIZE_NONCE
      + 2 * sizeof(size_t) + 4 != m2_len)
  {
    printf("A: m2 has wrong length\n");
    exit(1);
  }

  if(memcmp(m2, "msg2", 4))
  {
    printf("A: m2 not properly tagged\n");
    exit(1);
  }

  m2_l1 = *(size_t *) (m2 + 4);
  m2_l2 = *(size_t *) (m2 + 4 + sizeof(size_t));

  if(m2_l1 != SIZE_NONCE)
  {
    printf("A: m2 has wrong length for xNa\n");
    exit(1);
  }

  if(m2_l2 != SIZE_NONCE)
  {
    printf("A: m2 has wrong length for xNb\n");
    exit(1);
  }

  if(memcmp(m2 + 4 + 2 * sizeof(size_t),
            Na, m2_l1))
  {
    printf("A: xNa in m2 doesn't match Na\n");
    exit(1);
  }

#ifndef LOWE_ATTACK
  if(memcmp(m2 + m2_l1 + m2_l2
            + 2 * sizeof(size_t) + 4,
            xkey,  xkey_len))
  {
    printf("A: x_xkey in m2 doesn't match xkey\n");
    exit(1);
  }
#endif

  xNb = m2 + m2_l1 + 2 * sizeof(size_t) + 4;

  /* Send message 3 */

  m3_e_len = encrypt_len(xkey, xkey_len,
                         xNb, m2_l2);
  m3_e = malloc(m3_e_len + sizeof(size_t) + 4);
  memcpy(m3_e, "encr", 4);
  m3_e_len =
      encrypt(xkey, xkey_len, xNb,
              m2_l2, 
              m3_e + sizeof(m3_e_len) + 4);
  * (size_t *)(m3_e + 4) = m3_e_len;

  send(bio, m3_e, 
       m3_e_len + sizeof(m3_e_len) + 4);

  return 0;
}
\end{lstlisting}

\subsection{Proxy Functions}\label{nsl-proxies}

We show examples of proxy functions that replace calls to \cryptoc{nonce}, \cryptoc{encrypt}, etc. in the symbolic execution. Each function starts by calling the actual function that it replaces so that the concrete execution can proceed as usual---recall that we observe a run of the program in order to identify the main path. The proxy functions then call the special \emph{symbolic interface} functions to create new symbolic values and place them in memory. These symbolic interface functions are interpreted specially by the symbolic execution and perform the following actions:
\begin{itemize}
      \item 
            \begin{lstlisting}[language=cryptoc, gobble = 18]
                  load_buf(const unsigned char * buf, 
                           size_t len, const char * hint)
            \end{lstlisting}

            Retrieves from memory the expression located at \cryptoc{buf} of length \cryptoc{len} and places it on the stack. The value \cryptoc{hint} is attached to the expression for naming purposes. For instance, the names of variables in the IML model shown in \cref{nsl-iml} are derived from hints.
            
      \item 
            \begin{lstlisting}[language=cryptoc, gobble = 18]
                  store_buf(const unsigned char * buf)
            \end{lstlisting}
      
            Takes an expression from the stack and stores it in the location in memory pointed to by \cryptoc{buf}.
            
      \item 
            \begin{lstlisting}[language=cryptoc, gobble = 18]
                  symL(const char * sym, const char * hint, 
                       size_t len, int deterministic)
            \end{lstlisting}
      
            Applies the operation \cryptoc{sym} to all the expressions on the stack as parameters. Sets the length of the new expression to be equal to \cryptoc{len}. The last parameter can be used to specify that the application is non-deterministic, that is, conceptually it takes an extra random argument, without having to specify that argument explicitly. Calls to this function are also used to model random variable generation. For instance, the symbol $nonce$ created in \cryptoc{nonce_proxy} is treated specially and translates to the $\nu$ operator of IML.

      \item 
            \begin{lstlisting}[language=cryptoc, gobble = 18]
                  symN(const char * sym, const char * hint, 
                       size_t * len, int deterministic)
            \end{lstlisting}
      
            Behaves like \cryptoc{symL}, but instead of assigning a known length to the new expression $e$, keeps its length unrestricted and writes $\len(e)$ into \cryptoc{len}.
\end{itemize}

The proxy functions are trusted to represent the true behaviour of the actual cryptographic operations. For instance, the function \cryptoc{encrypt} is supposed to check the well-formedness of the key (corresponding to the symbolic operation $isek$). The actual cryptographic functions are required to satisfy the conditions listed in \cref{iml-verification-details} for the soundness result to hold. 

\begin{lstlisting}[language=cryptoc, basicstyle = \small]

void nonce_proxy(unsigned char * N)
{
  nonce(N);

  symL("nonce", "nonce", SIZE_NONCE, FALSE);
  store_buf(N);
}

size_t encrypt_len_proxy(unsigned char * key, 
                         size_t keylen, 
                         unsigned char * in, 
                         size_t inlen)
{
  size_t ret = 
    encrypt_len(key, keylen, in, inlen);

  symL("encrypt_len", "len", sizeof(ret), FALSE);
  store_buf(&ret);

  if(ret < 0) exit(1);

  return ret;
}

size_t encrypt_proxy(unsigned char * key, 
                     size_t keylen, 
                     unsigned char * in, 
                     size_t inlen, 
                     unsigned char * out)
{
  size_t ret = 
    encrypt(key, keylen, in, inlen, out);

  unsigned char nonce[SIZE_NONCE];

  nonce_proxy(nonce);

  load_buf(key, keylen, "key");
  symN("isek", "key", NULL, TRUE);
  load_buf(in, inlen, "msg");
  load_buf(nonce, SIZE_NONCE, "nonce");
  symN("E", "cipher", &ret, TRUE);
  store_buf(out);

  if(ret > encrypt_len_proxy(key, keylen, 
                             in, inlen))
    fail("encrypt_proxy: bad length");

  return ret;
}

unsigned char * get_pkey_proxy(size_t * len, 
                               char side)
{
  unsigned char * ret = get_pkey(len, side);

  char name[] = "pkX";
  name[2] = side;

  readenv(ret, len, name);

  return ret;
}

\end{lstlisting}

\subsection{IML Model}\label{nsl-iml}

The IML model extracted from both the client and the server is shown below. The notation $e\langle l\rangle$ is a shorthand for ``$e$ such that $\len(e) = l$''. For instance, \iml{in(c, var1<8>);} means \iml{in(c, var1); if len(var1) = 8 then}. 

The model contains several \iml{castToInt} expressions. These result from the fact that the implementation uses \cryptoc{size_t} as the length type, but the OpenSSL functions that we call use \cryptoc{int}. These type conversions are recorded during the symbolic execution. For now we assume no numeric overflows, as mentioned in \cref{implementation}, so the casts are removed before translating to pi.

\begin{lstlisting}[language=iml, basicstyle = \small, columns = flexible, keepspaces = true]
let A = 
new nonce1<20>;
let msg1 = 6d736731|i20|nonce1|pkA in
new nonce2<20>;
let cipher1 = E(isek(pkX), msg1, nonce2) in
let msg2 = 656e6372|len(cipher1)<8>|cipher1 in
out(c, msg2);
in(c, msg3<8>);
let var1 = (msg3 castToInt TSBase(int )) 
                 castToInt TSBase(unsigned long ) in
in(c, msg4<var1>);
let msg5 = D(skA, msg4) in
if len(pkX)<8> + i40 + i16 + i4 = len(msg5)<8> then 
if msg5{0, 4} = 6d736732 then 
if msg5{4, 8} = i20 then 
if msg5{12, 8} = i20 then 
let var2 = msg5{20, msg5{4, 8}} in
if var2 = nonce1 then 
let var3 = 
  msg5{msg5{4, 8} + msg5{12, 8} + i16 + i4, 
       len(msg5) - (msg5{4, 8} + msg5{12, 8} + i16 + i4)} in
if var3 = pkX then 
let msg6 = msg5{msg5{4, 8} + i16 + i4, msg5{12, 8}} in
new nonce3<20>;
let cipher2 = E(isek(pkX), msg6, nonce3) in
let msg7 = 656e6372|len(cipher2)<8>|cipher2 in
out(c, msg7); 0.

let B = 
in(c, msg1<8>);
let var1 = (msg1 castToInt TSBase(int )) 
                 castToInt TSBase(unsigned long ) in
in(c, msg2<var1>);
let msg3 = D(skB, msg2) in
if len(pkX)<8> + i8 + i20 + i4 = len(msg3)<8> then 
if msg3{0, 4} = 6d736731 then 
if msg3{4, 8} = i20 then 
let var2 = msg3{i8 + msg3{4, 8} + i4, 
                len(msg3) - (i8 + msg3{4, 8} + i4)} in
if var2 = pkX then 
let var3 = msg3{4, 8} in
let var4 = msg3{12, msg3{4, 8}} in
new nonce1<20>;
let msg4 = 6d736732|var3|i20|var4|nonce1|pkB in
new nonce2<20>;
let cipher1 = E(isek(pkX), msg4, nonce2) in
let msg5 = 656e6372|len(cipher1)<8>|cipher1 in
out(c, msg5);
in(c, msg6<8>);
let var5 = (msg6 castToInt TSBase(int )) 
                 castToInt TSBase(unsigned long ) in
in(c, msg7<var5>);
let msg8 = D(skB, msg7) in
if len(msg8)<8> = i20 then 
if msg8 = nonce1 then 
event endB(); 0.

\end{lstlisting}

\subsection{ProVerif Model}

The ProVerif model resulting from the translation of the IML process is shown below. The processes $A$ and $B$ as well as the symbolic rules for the new encoding and parsing expressions $conc_i$ and $parse_i$ are generated automatically from the source IML process. The rules for encryption and decryption, the query, and the environment process (including $A'$ and $B'$) are specified by hand.

The events are used without parameters---this is a limitation of the result in \cite{CoSP}, but our symbolic execution as well as ProVerif can easily deal with parameterised events. The modelling is similar to \cite{ProVerif,CoSP}. There the client $A'$ executes an event $beginA()$ only if it is supposed to talk to $B$ and $B'$ executes an event $endB()$ only if it supposed to talk to $A$. The event $endB()$ is executed at the end, so conceptually $B'$ needs to execute
\begin{lstlisting}[language = iml]
      if pkX = pkA then B; event endB(). else B.
\end{lstlisting}
Unfortunately, \iml{B; event endB().} does not form a valid process, so we use an equivalent formulation using an event $notA()$ instead---$endB()$ is always executed, but it is counted only if $notA()$ has not been executed.

The meaning of if-statements in pi is different from their meaning in IML. A pi calculus statement \iml{if $e_1 = e_2$ then $P$} corresponds to the IML \iml{let $\_ = eq(e_1, e_2)$ in $P$}.

\begin{lstlisting}[language=iml, basicstyle = \small, columns = flexible, keepspaces = true]
free c.
fun ek/1.
fun dk/1.
fun E/3.
reduc
  D(dk(a), E(ek(a), x, r)) = x.
reduc
  isek(ek(a)) = ek(a).
  

data conc2/2.

data conc5/3.

data conc11/1.

reduc 
  parse2(conc5(x0, x1, x2)) = x0.
reduc 
  parse3(conc5(x0, x1, x2)) = x2.
reduc 
  parse4(conc5(x0, x1, x2)) = x1.
reduc 
  parse6(conc2(x0, x1)) = x1.
reduc 
  parse7(conc2(x0, x1)) = x0.

query
  ev:endB() ==> ev:beginA() | ev:notA().

query 
  ev:endB() ==> ev:notA().

let A = 
new nonce1;
new nonce2;
let var1 = 
  conc11(E(isek(pkX), conc2(nonce1, pkA), nonce2)) in
out(c, var1);
in(c, msg1);
in(c, var2);
let var3 = parse2(D(skA, var2)) in
if var3 = nonce1 then 
let var4 = parse3(D(skA, var2)) in
if var4 = pkX then 
new nonce3;
let var5 = 
  conc11(E(isek(pkX), parse4(D(skA, var2)), nonce3)) in
out(c, var5); 0.

let B = 
in(c, msg2);
in(c, var27);
let var28 = parse6(D(skB, var27)) in
if var28 = pkX then 
new nonce4;
new nonce5;
let var29 = 
  conc11(E(isek(pkX), 
           conc5(parse7(D(skB, var27)), nonce4, pkB), 
           nonce5)) in
out(c, var29);
in(c, msg3);
in(c, var30);
let var31 = D(skB, var30) in
if var31 = nonce4 then 
event endB(); 0.


let A' =
  in(c, pkX);

  if pkX = pkB then
  event beginA(); A
  else A.
 
let B' = 
  in(c, pkX);

  if pkX = pkA then B else 
  event notA(); B.

process
  !
  new A; new B;
  let pkA = ek(A) in
  let skA = dk(A) in
  let pkB = ek(B) in
  let skB = dk(B) in
  out(c, pkA); out(c, pkB);
  (!A' | !B')
\end{lstlisting}

\end{FULL}

\end{document}